\documentclass[11pt]{article}
\usepackage{preamble}
\usetikzlibrary{positioning}
\usepackage{xcolor}
\usepackage{placeins}
\definecolor{gacEdge}{RGB}{0,150,170}
\definecolor{ghostEdge}{RGB}{200,200,200}
\definecolor{treeEdge}{RGB}{33,113,181}
\definecolor{pathEdge}{RGB}{228,26,28}
\definecolor{vertexA}{RGB}{235,250,252}
\definecolor{vertexB}{RGB}{230,244,255}
\definecolor{vertexC}{RGB}{255,235,235}  
\definecolor{diaggray}{RGB}{110,110,110}
\definecolor{xblue}{RGB}{38,92,155}
\definecolor{zzorange}{RGB}{210,120,40}

\definecolor{datablue}{HTML}{2166AC}
\definecolor{bandblue}{HTML}{9ECAE1}
\definecolor{refcoral}{HTML}{D6604D}
\definecolor{gridgray}{HTML}{CCCCCC}
\usepackage{pgfplots}
\pgfplotsset{compat=1.18}
\pgfplotsset{plot coordinates/math parser=false}
\usepackage{subcaption}

\usepackage[margin=1in]{geometry}
\usepgfplotslibrary{groupplots}
\PassOptionsToPackage{noend}{algpseudocode}
\usepackage{algorithm}
\usepackage{algpseudocode}
\usepackage{amsmath, amssymb}
\usepackage{enumitem}
\usepackage{placeins}

\numberwithin{equation}{section}

\pgfplotsset{
  paperplot/.style={
    width              = 8cm,
    height             = 6cm,
    axis lines         = left,
    axis line style    = {semithick, black},
    tick align         = outside,
    major tick length  = 4pt,
    minor tick num     = 3,
    xtick style        = {black, semithick},
    ytick style        = {black, semithick},
    tick label style   = {font=\small},
    label style        = {font=\small},
    title style        = {font=\small\bfseries},
    legend style       = {
      font        = \small,
      draw        = gray!50,
      fill        = white,
      fill opacity= 0.9,
      text opacity= 1,
      inner sep   = 5pt,
      row sep     = 1pt,
    },
    grid               = both,
    grid style         = {gridgray, very thin},
    minor grid style   = {gridgray!40, ultra thin},
    clip               = false,
  },
}

\allowdisplaybreaks

\usepackage{defs}
\usepackage{float}

\usepackage{textgreek}
\usepackage{geometry}

\vspace{1.5cm}
\author{
Constantin Cedillo Vayson de Pradenne  \\
Caltech \\
\href{mailto:ccedillo@caltech.edu}{\texttt{ccedillo@caltech.edu}}
\and
Jordan Cotler \\
Harvard University \\
\href{mailto:jcotler@fas.harvard.edu}{\texttt{jcotler@fas.harvard.edu}}
\and
Hsin-Yuan Huang \\
Oratomic, Caltech \\
\href{mailto:hhuang@oratomic.com}{\texttt{hhuang@oratomic.com}}
}

\renewcommand{\wt}{\mathrm{wt}}

\renewcommand{\top}{\intercal}

\usepackage[titles]{tocloft}

\usepackage{stackengine}
\stackMath
\newcommand\tsup[2][2]{%
 \def\useanchorwidth{T}%
  \ifnum#1>1%
    \stackon[-.5pt]{\tsup[\numexpr#1-1\relax]{#2}}{\scriptscriptstyle\sim}%
  \else%
    \stackon[.5pt]{#2}{\scriptscriptstyle\sim}%
  \fi%
}

\title{Learning Hamiltonians at Long Times}

\begin{document}
\captionsetup[algorithm]{hypcap=false}
\emergencystretch=3em
\hbadness=10000
\hfuzz=60pt
\pagestyle{empty}
{
  \renewcommand{\thispagestyle}[1]{}
  \maketitle

\begin{abstract}
We study the problem of learning an unknown $n$-qubit Hamiltonian $H$ from $U = e^{-iHt}$ for a single time $t$, where $t$ may be arbitrarily large. For broad families of local Hamiltonians, we prove that, with high probability over $H$ and $t$, any sum of local observables $A$ that is normalized and orthogonal to $H$ satisfies $\tfrac{1}{2^n}\|[U(t),A]\|_F^2 \geq 1/\poly(n)$. The Hamiltonian is therefore the unique approximately conserved local observable, and we can efficiently recover~$H$, up to scale, as the approximate null vector of a data matrix built from random product-state inputs and classical shadows. As a corollary, we obtain a weak equilibration statement: the infinite-temperature autocorrelation of every sum of local observables orthogonal to $H$ decays by at least an inverse-polynomial amount.
\end{abstract}

}

\clearpage
\pagestyle{plain}
\pagenumbering{arabic}

\tableofcontents

\newpage

\begin{figure}
    \centering
\includegraphics[width=1.0\linewidth]{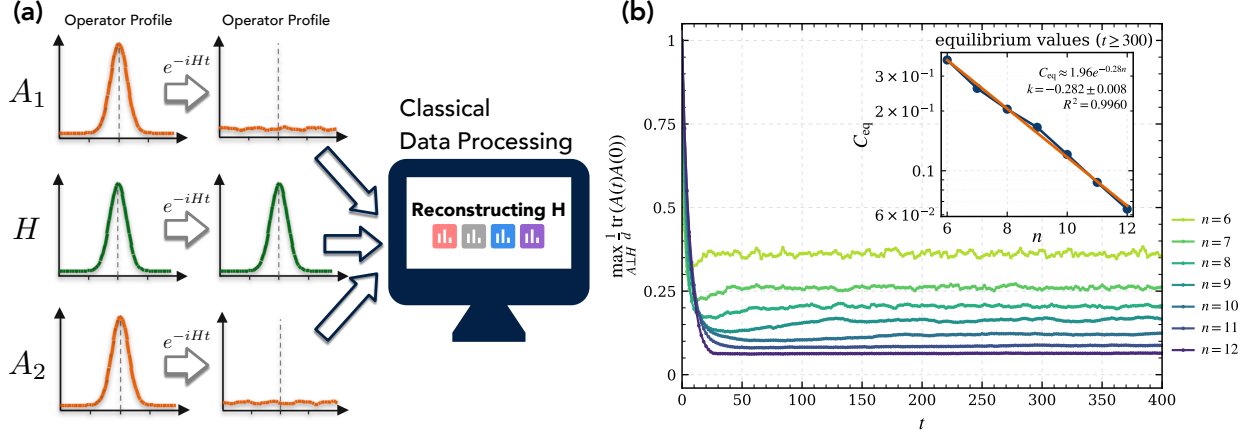}
\caption{
\textbf{Hamiltonian learning from weak equilibration.}
(a) \emph{Schematic of the mechanism}. Under long-time evolution, local observables in the search space generically lose their infinite-temperature autocorrelation, while the Hamiltonian direction is exactly conserved. Classical shadow measurements on random product-state probes produce a data matrix whose approximate null vector recovers $H$.
(b) \emph{Weak-equilibration diagnostic for random local Hamiltonians}. For each system size~$n$, we track how strongly the most slowly decorrelating normalized local observable orthogonal to $H$ retains its initial value under the dynamics. This quantity starts at $1$ and rapidly settles to a plateau below $1$. The inset shows the late-time plateau values with a semi-log fit, indicating that the plateau of the strongest non-Hamiltonian local autocorrelation decays exponentially with $n$. Additional numerical results are in Appendix~\ref{Additional Numerics}.
}
    \label{fig:main}
\end{figure}

\section{Introduction}
\label{sec:intro}

A basic question in quantum learning is whether an unknown local Hamiltonian $H$ can be recovered from a single time-evolution unitary $U(t) = e^{-iHt}$, when the evolution time $t$ may be large and unknown. Since the scale of $H$ is inseparable from the unknown time, the natural target is the Hamiltonian \emph{direction} within a prescribed space of local operators. Without further assumptions, worst-case recovery is impossible or at least ill-posed: different Hamiltonians can yield the same $U(t)$ at isolated values of $t$, and nongeneric Hamiltonians may have many local conserved quantities, so identifying $H$ as a local direction fixed by the dynamics is underdetermined. We therefore consider an average-case version of the question, and prove that for broad ensembles of random local Hamiltonians, a single long-time evolution suffices to determine the Hamiltonian direction among local observables, with high probability over $H$ and~$t$.

The argument rests on a simple observation about commutators. For any local observable $A$ in the search space, the Frobenius norm $\|[U(t),A]\|_F$ quantifies how strongly the dynamics rotate $A$: if $A$ is approximately conserved, this commutator is small, while a lower bound on it certifies that $A$ is not preserved by the long-time evolution. The Hamiltonian itself is trivially conserved, since $[U(t),H]=0$. Thus, in the absence of other approximately conserved local quantities, the local observable most preserved by conjugation by $U(t)$ must be $H$, up to scale and sign. This motivates a learning procedure that minimizes $\|[U(t),A]\|_F$ over local observables~$A$. Our main contribution is to establish a weak but \emph{uniform} version of this principle, strong enough to make the procedure rigorous even when $t$ is large.

Quantitatively, for $n$ qubits and $d=2^n$, we study the normalized commutator $\tfrac{1}{d}\|[U(t),A]\|_F^2$. Our central result is that, for the random local Hamiltonian classes considered here, the Hamiltonian direction is the \emph{only} approximately conserved direction in the chosen space. More precisely, subject to the orthogonality and normalization conditions $\tr(AH)=0$ and $\tfrac{1}{d}\tr(A^2)=1$, we prove an inverse-polynomial lower bound on $\tfrac{1}{d}\|[U(t),A]\|_F^2$ that holds simultaneously for all admissible local observables $A$, with high probability over both the sampled Hamiltonian $H$ and the sampled time~$t$. The order of quantifiers matters: $t$ is sampled first, and the observable $A$ may be chosen afterward, so the lower bound must be uniform in $A$.

These bounds connect Hamiltonian learning to the broader study of equilibration. Closed quantum systems evolve reversibly, yet local observables in many-body systems often appear to lose memory of their initial values at long times. Foundational results establish that closed systems generically equilibrate under broad spectral or typicality assumptions~\cite{Reimann2008,LindenPopescuShortWinter2009,Short2011Equilibration,ShortFarrelly2012,GogolinEisert2016}, and the eigenstate thermalization hypothesis provides a microscopic mechanism for thermal behavior in chaotic systems~\cite{Deutsch1991,Srednicki1994}. These previous results, however, do not yield inverse-polynomial lower bounds on commutators of local observables with a long-time unitary. We prove such bounds for several random local Hamiltonian models, and use them to learn $H$ from a single long-time evolution.

The picture is more subtle in the worst case. Adversarially chosen Hamiltonians can carry many conserved charges, and integrable systems in particular often admit extensive families of local or quasi-local conservation laws. No theorem can therefore assert that every local Hamiltonian has a unique conserved local direction. Our viewpoint is instead average-case: even when a Hamiltonian admits additional conserved quantities, a small generic local perturbation typically breaks those conservation laws within the low-weight local sector.  We make this average-case statement precise for the model classes analyzed in Appendix~\ref{sec:commutator-gaps}. There, we show that generic local randomness removes accidental conserved directions in the low-weight sector, leaving the Hamiltonian direction isolated by an inverse-polynomial commutator gap. The same framework also gives smoothed-analysis guarantees: after a small continuous perturbation of the local coefficients, the same isolation holds with high probability.

Existing approaches to Hamiltonian learning span a wide range of access models, including local measurements on eigenstates and Gibbs states, short-time dynamics, controlled dynamics, and structured Hamiltonian classes~\cite{granade2012robust, wiebe2014hamiltonian, wang2017experimental, BaireyAradLindner2019, evans2019scalable, AnshuArunachalamKuwaharaSoleimanifar2021, HuangTongFangSu2023, yu2023robust, gu2024practical, caro2024learning, HaahKothariTang2024, BakshiLiuMoitraTang2024Structure, BakshiLiuMoitraTang2024Temperature, DutkiewiczOBrienSchuster2024, StilckFrancaetal2024, castaneda2025hamiltonian, zhao2025learning, ArunachalamDuttEscudero2025, hu2025ansatz}. Closest to our setting is the real-time evolution model: \cite{HuangTongFangSu2023} achieved Heisenberg-limited scaling using quantum control, \cite{DutkiewiczOBrienSchuster2024} analyzed the role of control more broadly, and \cite{BakshiLiuMoitraTang2024Structure} solved a structure-learning problem in this setting. Contemporaneous work~\cite{ShinLeeOh2026} also studies $e^{-iHt}$ without access to arbitrarily small times, but assumes that the evolution time can be chosen by the learner from a range bounded below by an inverse-polynomial, or by a constant, depending on the setting. All of these works leave open the question we address: whether an unknown $n$-qubit Hamiltonian $H$ can be learned from $U = e^{-iHt}$ at a single, unknown, and possibly arbitrarily large $t$. A more detailed comparison with prior work, including equilibration theory and conservation-law discovery, is given in Appendix~\ref{Related work}.

We prove that for broad random local Hamiltonian ensembles, a single unknown and possibly long-time evolution can identify the Hamiltonian direction (see Fig.~\ref{fig:main}). The proof has two main ingredients. First, a \emph{dynamical reduction} (Appendix~\ref{section-sinc}), based on a uniform sinc identity, converts a static commutator lower bound for $[H,A]$ into a dynamical commutator lower bound for $[U(t),A]$, uniformly over all local observables in the search space. Second, we establish an inverse-polynomial lower bound (Appendix~\ref{sec:commutator-gaps}) on the quantity
\begin{align}
\min_{\substack{A\in\mathcal{O}_{\mathcal{V}}^{\mathrm{op}} \\ \tfrac{1}{d}\tr(A^2)=1,\ \tr(AH)=0}}
\frac{1}{d}\,\|[A,H]\|_F^2 \, ,
\end{align}
where $\mathcal{O}_{\mathcal{V}}^{\rm op}$ is a vector space of local Paulis out of which the Hamiltonian can be built.  We establish the lower bound in two regimes. In the \emph{Pauli-injective regime} (Appendix~\ref{subsec:pauli-injective-commutator-gaps}), the Pauli support has a connected anticommutation graph and satisfies an edge-product injectivity condition: distinct anticommuting pairs of Pauli strings produce distinct Pauli products, up to a phase. Under this condition, the commutator decomposes as an orthogonal sum over the edges of the anticommutation graph, and the static gap reduces to a graph Poincar\'{e}, or spectral gap, estimate. In the \emph{geometrically local regime} (Appendix~\ref{subsec:geometric-local-nondegeneracy}), the main complication is that different commutator terms can overlap and interfere. We handle this by organizing the commutator locally: each bounded-size neighborhood gives a local constraint on possible conserved observables, and these local constraints are then patched together across the interaction graph to isolate the Hamiltonian direction. Certified examples include the transverse-field Ising family in the Pauli-injective regime; the periodic XYZ-chain support family certified by exact local rank checks; the dense finite-range Pauli families on the $1$D chain, square, triangular, honeycomb, and cubic lattices with support size at most $k\in \{2,3,4\}$ and and graph diameter at most $R\in \{1,2\}$; and the exact two-body no-field Pauli families for $R \in \{1,2\}$ on the chain, square, triangular, honeycomb, and cubic lattices, with open or periodic boundary conditions.

The resulting commutator gaps yield both a learning algorithm (Appendix~\ref{sec:learning-hamiltonian}) and a weak equilibration statement (Appendix~\ref{sec:equilibration}). Although we do not prove thermalization, nor that all local autocorrelations decay to zero, we show that every normalized local observable orthogonal to $H$ loses at least an inverse-polynomial fraction of its infinite-temperature autocorrelation at the sampled time. Within the local sector, the Hamiltonian direction is therefore the unique direction that remains almost fixed under the dynamics, and the algorithm recovers $H$, up to scale and sign, as the local observable least disturbed by conjugation by $U(t)$. Operationally, this minimization is implemented by building an empirical data matrix from random product-state probes and classical shadows, and extracting its approximate null vector. In this sense, weak equilibration in the local sector is what makes Hamiltonian learning from a single long-time evolution possible.

\section{Results}
\label{sec:results}

\subsection{Setup and notation}

We work on an $n$-qubit Hilbert space of dimension $d=2^n$, and write $\mathcal{P}_n$ for the set of $n$-qubit Pauli strings modulo overall phases. Throughout, we use the normalized Hilbert--Schmidt inner product $\langle A,B\rangle_{HS} := \tfrac{1}{d}\tr(A^\dagger B)$, with associated norm $\|A\|_{HS}^2 := \tfrac{1}{d}\tr(A^\dagger A)$.  For a given subset of Pauli strings $\mathcal{V} \subset \mathcal{P}_n \setminus \{I\}$, the associated real local span is
\begin{align}
\mathcal{O}_{\mathcal{V}}^{\mathrm{op}} := \operatorname{span}_{\mathbb{R}}\{P_u : u \in \mathcal{V}\}\,,
\end{align}
and we expand the Hamiltonian and observable in this basis as $H = \sum_{u\in\mathcal{V}} h_u P_u$ and $A = \sum_{u\in\mathcal{V}} a_u P_u$. To remove the exactly conserved Hamiltonian direction, we impose $\tr(AH)=0$, and we normalize local observables by $\tfrac{1}{d}\tr(A^2)=1$.
We adopt the extensive normalization $\tfrac{1}{d}\tr(H^2) = n$. Since the Pauli strings are orthonormal under the normalized Hilbert--Schmidt inner product, this is the same as requiring the coefficient vector to satisfy $\|h\|_2^2 = n$. It mirrors the physical convention that an $n$-qubit local Hamiltonian is a sum of $\Theta(n)$ bounded local terms.

In our analysis, we consider two regimes for the search space defined by the subset~$\mathcal{V}$, which we refer to as the search family. The unknown $n$-qubit Hamiltonian $H = \sum_{u \in \mathcal{V}} h_u P_u$ is in the span of $\mathcal{V}$ and is considered to be sampled according to a broad range of distributions over $h_u$. In the \emph{geometrically local regime}, $\mathcal{V}$ consists of Pauli strings of bounded support size and bounded graph diameter on a bounded-degree interaction graph. In the \emph{Pauli-injective regime}, $\mathcal{V}$ has a connected anticommutation graph and satisfies an edge-product injectivity condition. Throughout the main text, $\mathcal{V}$ denotes a generic finite Pauli family. We specialize to the Pauli-injective regime $\mathcal{V}_{\mathrm{inj}}$ in Appendix~\ref{subsec:pauli-injective-commutator-gaps}, and to the geometrically local regime $\mathcal{V}_{k,R}$ in Appendix~\ref{subsec:geometric-local-nondegeneracy}; the generic notation $\mathcal{V}$ and $\mathcal{O}^{\mathrm{op}}_{\mathcal{V}}$ are used in Appendices~\ref{section-sinc}, \ref{sec:learning-hamiltonian}, and~\ref{sec:equilibration}. The certified geometric examples are the support families for which the structural local nondegeneracy condition has been verified by the exact certificates described in Appendix~\ref{subsubsec:deterministic-certificates}. These include the dense finite-range cases on the listed lattices for the certified $(R,k)$ ranges stated there, and the exact two-body no-field families for $R \in \{1,2\}$.

\subsection{Learning from a single unknown time evolution}

Our learning algorithm searches for the local observable least disturbed by the unknown unitary. The Hamiltonian is always one such observable, and the commutator-gap theorem implies that, in the model classes considered here, it is essentially the only one.

\begin{theorem}[Single-time learning, informal]
\label{thm:informal-learning}
Suppose $U = e^{-iHt}$, where $H$ belongs to a local Hamiltonian ensemble in Appendix~\ref{sec:commutator-gaps} and $t$ is uniformly distributed in $[0, T]$ for any $T \geq 1/\text{\rm poly}(n)$. For any inverse-polynomial target accuracy~$\varepsilon$, there is a polynomial-time algorithm, using polynomially many random product-state inputs and randomized Pauli measurements~\cite{HuangKuengPreskill2020}, that outputs an estimate $\widehat H$ satisfying
\begin{align}
\min_{s\in\{\pm1\}} \|\widehat H - sH\|_{\mathrm{op}} \leq \varepsilon,
\end{align}
with high probability over the Hamiltonian $H$, time $t$, and randomized experimental outcomes.
\end{theorem}

The precise algorithm and proof are given in Appendix~\ref{sec:learning-hamiltonian}. The algorithm builds a data matrix from $M_{\mathrm{probe}}$ random product-state probes $\rho^{(\ell)}$ drawn from the six-state ensemble for $\ell = 1, \dots, M_{\mathrm{probe}}$. For each Pauli string $P_v$ in the search family $\mathcal{V}$, classical shadows~\cite{HuangKuengPreskill2020} provide an estimator $\widehat m_{\ell v}$ of the post-evolution expectation $m_{\ell v} := \tr\!\left(P_v U\rho^{(\ell)}U^\dagger\right)$. We then assemble the entries
\begin{align}
\widehat X_{\ell v} = \tr\!\left(P_v\rho^{(\ell)}\right) - \widehat m_{\ell v}.
\end{align}
In the ideal noiseless matrix $X$, the unit Hamiltonian direction $h_{\mathrm{unit}}:=h/\sqrt{n}$ is a right null vector, since $U^\dagger H U=H$. The estimator $\widehat h_{\mathrm{unit}}$ is the right singular vector of $\widehat X$ with the smallest singular value, and we set $\widehat H=\sqrt{n}\sum_{v\in\mathcal V}(\widehat h_{\mathrm{unit}})_v P_v$.

The analysis rests on a covariance gap. For any unit coefficient vector $a$ with associated observable $A = \sum_v a_v P_v$, consider the row variance
\begin{align}
\operatorname{Var}_\rho(A; U) := \mathbb{E}_{\rho}\!\left[ \bigl( \tr(A\rho) - \tr(A\, U\rho U^\dagger) \bigr)^2 \right],
\end{align}
where the expectation is over $\rho$ drawn from the six-state ensemble. A second-moment identity for this ensemble lower bounds $\operatorname{Var}_\rho(A;U)$ by the local projection of $U^\dagger A U - A$, and a projection lemma in turn lower bounds the latter by the dynamical commutator gap
\begin{align}
\pi_U := \min_{\substack{A \in \mathcal{O}_{\mathcal{V}}^{\mathrm{op}} \\ \tfrac{1}{d}\tr(A^2)=1,\ \tr(AH)=0}} \tfrac{1}{d}\|[U,A]\|_F^2 \,.
\end{align}
The commutator gap yields a strictly positive lower bound on the probe variance in every direction $a \perp h$, equivalently a spectral gap in the expected row covariance. A matrix Chernoff bound then transfers this covariance gap to the finite-sample matrix $X$, giving
\begin{align}
\sigma_2(X) \geq \sqrt{\frac{M_{\mathrm{probe}}\, 3^{-k}\, \pi_U^2}{8}}
\end{align}
with high probability, provided $M_{\mathrm{probe}} = \poly(n)$. Standard singular-vector perturbation bounds then show that the smallest right singular vector of $\widehat X$ remains close to $h_{\mathrm{unit}}$, provided the entrywise classical-shadow error is sufficiently small. As a byproduct, within the promised model class, the spectrum of $\widehat X$ provides a consistency check: if $U = e^{-iHt}$ with $H$ in the promised class, then $\widehat X$ has an isolated approximately null direction.

Appendix~\ref{sec:time-estimation} gives an optional post-processing routine that assigns a representative time once a normalized Hamiltonian direction has been learned. This routine is not part of the product-state classical-shadow learning guarantee: it uses a stronger coherent access model, including preparation of a maximally entangled state and coherent implementation of $e^{i\widehat H g}$. However, it does not require controlled access to $U(t_\star)$. With ordinary coherent uses of $U(t_\star)$, global phase is not observable, so the natural target is the projective Frobenius loss
\begin{equation}
D_{\rm proj}(s,g) := \min_{\phi \in \mathbb R} \frac{1}{2d}\left\|e^{i\phi}e^{-is\widehat H g} - U(t_\star)\right\|_F^2 = 1 - \left|\frac{1}{d}\operatorname{tr}\!\left(e^{is\widehat H g}U(t_\star)\right)\right|.
\end{equation}
The routine estimates the associated Choi overlap by preparing $\ket{\Phi_d}$, applying $U(t_\star)$ and then $e^{is\widehat H g}$ to one half, and measuring the projector onto $\ket{\Phi_d}$.

\subsection{Commutator gaps and local identifiability}

Theorem~\ref{thm:informal-learning} depends on showing that no local observable other than $H$ is nearly conserved by $U(t)$. The next result gives this identifiability statement as a commutator gap: every normalized local observable orthogonal to $H$ has an inverse-polynomial commutator with the sampled unitary.

\begin{theorem}[Commutator-gap theorem, informal]
\label{thm:informal-commutator-gap}
Let $H = \sum_{u \in \mathcal{V}} h_u P_u$ with $\tfrac{1}{d}\tr(H^2) = n$, drawn from either \text{\rm (i)} a Pauli-injective ensemble or \text{\rm (ii)} a certified geometrically local ensemble considered in Appendix~\ref{sec:commutator-gaps}. With high probability over $H$,
\begin{align}
\text{\rm (static commutator gap)} \quad \min_{\substack{A \in \mathcal{O}_{\mathcal{V}}^{\mathrm{op}} \\ \tfrac{1}{d}\tr(A^2)=1,\ \tr(AH)=0}} \frac{1}{d}\|[A,H]\|_F^2 \geq \frac{1}{\poly(n)}.
\end{align}
Furthermore, if $t$ is then sampled uniformly from $[0,T]$ for any $T \geq 1/\text{\rm poly}(n)$, the uniform sinc reduction of Appendix~\ref{section-sinc} implies that, with high probability over $t$,
\begin{align}
\text{\rm (dynamic commutator gap)} \quad \min_{\substack{A \in \mathcal{O}_{\mathcal{V}}^{\mathrm{op}} \\ \tfrac{1}{d}\tr(A^2)=1,\ \tr(AH)=0}} \frac{1}{d}\|[U(t),A]\|_F^2 \geq \frac{1}{\poly(n)}.
\end{align}
\end{theorem}

\noindent The precise dynamical reduction is Theorem~\ref{relation-commutator-A-U}, and the Pauli-injective and geometric gap statements are Theorem~\ref{thm:pauli-injective-gap}, Lemma~\ref{lem:dominating-extension}, Proposition~\ref{prop:deterministic-geometric-gap-from-local-gaps}, and Theorem~\ref{thm:geo-global-gap}.

The two regimes have different sources of genericity. In the Pauli-injective regime, the static gap is in fact \emph{deterministic}: Theorem~\ref{thm:pauli-injective-gap} gives an explicit gap whenever the anticommutation graph is connected, edge-product injectivity holds, and the nonzero coefficients are bounded above and below in absolute value. This regime therefore includes non-random examples, such as the uniform-coefficient transverse-field Ising model on a connected bounded-degree graph with non-vanishing couplings; randomness enters only as one convenient way to ensure coefficient regularity. By contrast, the geometrically local regime relies on random or smoothed coefficients to supply the local nondegeneracy and anti-concentration estimates required by the patching argument.

The proof of Theorem~\ref{thm:informal-commutator-gap} begins with an exact identity. Diagonalizing $H = \sum_{p = 1}^{d} E_p\, |p\rangle\langle p|$ and writing $\Delta_{pq} = E_p - E_q$, one has
\begin{align}
\|[U(t),A]\|_F^2 = t^2 \sum_{p,q = 1}^{d} |[A,H]_{pq}|^2\, \operatorname{sinc}^2\!\left(\frac{\Delta_{pq}\, t}{2}\right).
\end{align}
For a fixed $A$, the only obstruction to comparing $[U(t),A]$ with $[H,A]$ is that some of the relevant sine factors may be small at the sampled time. The learning problem requires a uniform statement: $t$ is sampled once, and then $A$ may be chosen adversarially from the local search space. We therefore have to show that the small-sinc coordinates cannot dominate any vector in the commutator image
\begin{align}
\mathcal{W}_H := \{[A,H] : A \in \mathcal{O}_{\mathcal{V}}^{\mathrm{op}},\ \tr(AH) = 0\}.
\end{align}
For each $t$, we mark an energy-gap coordinate $(p,q)$ as \emph{bad} if its sinc factor is too small. A trace bound on the bad-coordinate projector restricted to $\mathcal{W}_H$ then shows that, with high probability over $t$, the bad coordinates carry less than half the Frobenius mass of any vector in $\mathcal{W}_H$. The surviving good coordinates therefore retain a constant fraction of the static commutator, and the static commutator gap lifts uniformly to the dynamic commutator gap.

It remains to establish the static commutator gap. In the Pauli-injective regime, the commutator decomposes orthogonally over the edges of the anticommutation graph $\mathcal{G} = (\mathcal{V}, \mathsf{E})$,
\begin{align}
\frac{1}{d}\|[A,H]\|_F^2 = 4 \sum_{\{u,v\} \in \mathsf{E}} (a_u h_v - a_v h_u)^2,
\end{align}
and the change of variables $c_u = a_u/h_u$ recasts this sum as a weighted graph Dirichlet form $\sum_{\{u,v\} \in \mathsf{E}} h_u^2 h_v^2 (c_u - c_v)^2$. Orthogonality to $H$ rules out the constant mode, and connectivity of $\mathcal{G}$ supplies a Poincar\'{e} inequality, which yields an explicit inverse-polynomial lower bound whenever $|\mathcal V| \leq \operatorname{poly}(n)$, $h_{\max} \leq \operatorname{poly}(n)$, and $h_{\min} \geq 1/\operatorname{poly}(n)$. This argument first controls observables supported on the same Pauli family as $H$. The dominating-extension lemma then extends the bound to a larger search family $\mathcal{V}_{\max} \supseteq \mathcal{V}$, provided every added Pauli direction anticommutes with at least one Hamiltonian term and the resulting commutator products do not cancel.

In the geometrically local regime, distinct anticommuting pairs may produce the same output Pauli string, so the corresponding commutator terms can interfere rather than add independently. To keep track of these collisions, we use a commutator matrix $B(h)$ satisfying
\begin{align}
\frac{1}{d}\|[A,H]\|_F^2 = 4\|B(h)\,a\|_2^2 = 4\,a^\top \Gamma(h)\,a, \quad \Gamma(h) := B(h)^\top B(h).
\end{align}
Here $\Gamma(h)$ is the Gram matrix associated with the commutator map. The Pauli strings are supported on a bounded-degree interaction graph $G = (\Lambda,\mathsf E)$, whose vertices are the qubits, and this locality lets us decompose the Gram matrix into local positive semidefinite pieces, $\Gamma(h) = \sum_{c \in \Lambda} \Gamma_c(h)$. Each $\Gamma_c(h)$ only involves Pauli coefficients supported near the site $c$. For each local patch, a cofactor polynomial $S_c$ detects whether the only local null direction is the restriction of the Hamiltonian coefficient vector $h$. Anti-concentration gives a spectral gap for each local piece with high probability. A patching argument then promotes these local gaps to a global commutator gap.

Appendix~\ref{sec:commutator-gaps} also gives a smoothed version of the commutator-gap result, adopting the perspective of recent work~\cite{chen2025quantum}. Worst-case Hamiltonians may have many conserved quantities, so the theorem cannot hold uniformly over all local Hamiltonians. Instead, the result is generic: once the relevant support-family hypotheses are satisfied, a small continuous perturbation of the local coefficients typically removes accidental conserved directions in the low-weight sector and leaves the Hamiltonian direction isolated by an inverse-polynomial commutator gap.

The structural local nondegeneracy assumption underlying the geometrically local regime is a finite-dimensional algebraic condition on the local Pauli support family. It is certified patch by patch using exact arithmetic. For each rooted local type $c$, the certificate uses the actual coordinate set $\mathcal U_c$ and verifies the existence of an integer witness vector $h_0^{(\mathcal U_c)}$ for which the local commutator matrix has rank $d_c - 1$. Floating-point computations may be used to search for candidate witnesses, but the certificate itself is exact: it consists of the integer null-vector check $B_c(h_0)h_0^{(\mathcal U_c)}=0$ over $\mathbb Z$ and the modular-rank check $\operatorname{rank}_{\mathbb F_p} B_c(h_0) = d_c-1$ over the prime field used in the certificate. Details are in Appendix~\ref{subsubsec:deterministic-certificates}.\footnote{The code used to perform these checks is available at \texttt{witness\_hamiltonians/} in \url{https://github.com/ConstantinCed/Learning-Hamiltonians-using-Quantum-Equilibration}.}

\subsection{Weak equilibration}

The commutator gap has an immediate consequence for weak equilibration: it rules out the persistence of any local observable orthogonal to $H$ that retains almost all of its infinite-temperature autocorrelation at the sampled time.

\begin{theorem}[Weak equilibration, informal]
\label{thm:informal-equilibration}
Assume the dynamic commutator gap of Theorem~\ref{thm:informal-commutator-gap} at a sampled time~$t$. Then every normalized Hermitian $A \in \mathcal{O}_{\mathcal{V}}^{\mathrm{op}}$ with $\tr(AH) = 0$ satisfies
\begin{align}
\frac{1}{d}\,\tr\bigl(A(t)\, A\bigr) \leq 1 - \frac{1}{\poly(n)},
\end{align}
where $A(t) := U(t)^\dagger A\, U(t)$. Equivalently, every normalized local observable orthogonal to $H$ loses at least an inverse-polynomial fraction of its infinite-temperature autocorrelation.
\end{theorem}

\noindent The precise statement is Lemma~\ref{lem:weak-equilibration} and Corollary~\ref{cor:autocorrelation-decay-from-static-gap} in Appendix~\ref{sec:equilibration}. The proof is the identity
\begin{align}
\frac{1}{2d}\|[U(t),A]\|_F^2 = 1 - \frac{1}{d}\,\tr\bigl(A(t)\, A\bigr),
\end{align}
valid for normalized Hermitian $A$, which converts a dynamical commutator lower bound directly into an upper bound on the autocorrelation. The infinite-temperature autocorrelation $C_A(t) := \tfrac{1}{d}\tr(A(t)\, A)$ measures how close $A(t)$ remains to $A$. The commutator gap forces $C_A(t)$ to lie a uniform inverse-polynomial distance below $1$ for every normalized local $A \perp H$, so the local sector contains no almost-fixed direction other than $H$ itself.

This is genuinely weaker than conventional thermalization, in three distinct senses. First, we do not show that $C_A(t)$ is close to zero, only that it is bounded away from $1$. Second, we do not control cross-correlations $\tfrac{1}{d}\tr(A(t)\, B)$. Third, our bound does not rule out the possibility that $U(t)^\dagger A U(t)$ has large overlap with some other local observable $B$; it only rules out $A$ remaining close to itself.  What we do rule out is the obstruction that matters for the learning problem: a nontrivial local direction that remains nearly unchanged under the dynamics. This is precisely what the algorithm requires, since within the local search space the Hamiltonian direction is then the only surviving direction.

The strength of equilibration directly controls the strength of the learning guarantee. Our theorems prove only an inverse-polynomial loss of autocorrelation, which suffices for polynomial-time learning. The algorithm itself, however, depends on the \emph{actual} commutator gap, not on the lower bound we are able to prove. If the true autocorrelation loss is $\delta$ rather than $1/\poly(n)$, then the dynamic commutator gap is of order $\delta$, and the data matrix $\widehat X$ has a correspondingly larger separation between the Hamiltonian null direction and all orthogonal local directions. Stronger equilibration therefore translates immediately into better conditioning for the learning problem. This is consistent with our numerical experiments (Appendix~\ref{Additional Numerics}), which suggest faster convergence than the conservative inverse-polynomial bounds alone would predict.

\section{Discussion}
\label{sec:discuss}

We have argued that learning a local Hamiltonian from long-time dynamics is an equilibration problem in disguise. If local observables equilibrate aside from conserved quantities, and if $H$ is the only local conserved quantity, then $H$ is the unique local direction preserved by $U(t)$, and we can recover it. Our results make this strategy rigorous for broad local Hamiltonian ensembles.

The commutator gap ties the algorithmic and physical pictures together. Algorithmically, it is identifiability of $H$ within the search space spanned by local observables; physically, it is the statement that every normalized local observable orthogonal to $H$ loses at least an inverse-polynomial amount of its infinite-temperature autocorrelation under $U(t)$. Both are facts about the spectrum of $[U(t),\,\cdot\,]$ restricted to the space spanned by local observables. The same dynamic commutator gap controls both the amount of local equilibration and the conditioning of the learning problem. Our proof gives only a conservative inverse-polynomial lower bound on this gap. The algorithm, however, benefits from the actual gap, which may be much larger in typical instances.

Indeed, the inverse-polynomial gaps established here are conservative guarantees from general-purpose spectral and anti-concentration estimates, and the algorithm is sensitive to the \emph{actual} gap of $[U(t),\,\cdot\,]$ on local operators, not to the lower bound we prove. When equilibration is stronger, the approximate null direction of $H$ is more sharply isolated and the empirical data matrix is better conditioned. Our numerics suggest this stronger-equilibration regime occurs in practice, so any sharper understanding of long-time equilibration would translate directly into sharper sample-complexity and stability guarantees for Hamiltonian learning.

We work in average-case regimes since integrable and highly symmetric models carry many conserved charges, and no theorem of this shape can hold uniformly over all local Hamiltonians. Indeed the average-case picture is the natural fix since a small local perturbation breaks accidental conservation laws in the low-weight local sector, and our analysis turns this intuition into bounds via coefficient anti-concentration and local nondegeneracy. Our results are therefore best read as \emph{average-case identifiability} rather than worst-case uniqueness.

A few caveats deserve emphasis. The dynamic gap holds with high probability over $t$; special recurrence times can still produce accidental cancellations. The geometric theorem assumes a local nondegeneracy condition (Appendix~\ref{subsec:geometric-local-nondegeneracy}), which is checkable patch by patch but not automatic for sparse models. The equilibration statement controls only the autocorrelation of $A$, not arbitrary correlators $\tfrac{1}{d}\tr(A(t)\, B)$ or time-averaged projections onto a restricted commutant; extending it is a natural next step. The learning theorem is a promise result rather than a property tester: an isolated approximate null direction is a robust signature of $U = e^{-iHt}$ within our model classes, but \emph{soundness} against arbitrary unitaries far from any local Hamiltonian evolution remains open.

The broader picture is that conserved quantities and learnability are two faces of the same phenomenon: a local Hamiltonian is learnable from long-time dynamics precisely because its own direction persists while other local directions do not. Pushing this idea into systems with extensive conservation laws, hydrodynamic behavior, or weak ergodicity breaking may lead both to sharper learning algorithms and to a clearer operational interpretation of equilibration itself.

\subsection*{Acknowledgments}

CCVP thanks Abhiram Cherukupali for helpful discussions. JC is supported by an Alfred P.~Sloan Fellowship. H.H. acknowledges support from the Broadcom Innovation Fund and the U.S. Department of Energy, Office of Science, National Quantum Information Science Research Centers, Quantum Systems Accelerator. We acknowledge ChatGPT 5.5 Pro and Claude Opus 4.7 for supplying ideas for several proofs, checking typos, and helping write the Hamiltonian certification code. We have carefully written and verified the proofs, and responsibility for the correctness of the proofs is our own.

\newpage
\noindent {\LARGE \bfseries Appendices}

\appendix

\paragraph{Roadmap for Appendices}
\begin{itemize}[leftmargin=*]
    \item Appendix~\ref{Related work} reviews related work on equilibration, Hamiltonian learning, learning from real-time evolution, and conservation-law discovery.
    
    \item Appendix~\ref{section-sinc} proves the uniform sinc reduction, which converts a static commutator gap for $[H,A]$ into a dynamical commutator gap for $[U(t),A]$ at a single randomly sampled time, uniformly over the local observables in the chosen search space.
    
    \item Appendix~\ref{sec:commutator-gaps} proves the static commutator gaps used throughout the paper. It treats the Pauli-injective regime through the anticommutation graph and the edge-product injectivity condition, including the dominating-extension argument and the transverse-field Ising example. It then treats geometrically local Hamiltonians by constructing the collision-aware commutator matrix, decomposing its Gram form into local positive semidefinite pieces, proving local small-ball estimates from the cofactor polynomial, and patching local gaps into a global gap. We also specialize the bounds to independent uniform, independent Gaussian, and smoothed coefficient distributions, after stating the certified geometric support-families that are covered.
    
    \item Appendix~\ref{sec:learning-hamiltonian} gives the Hamiltonian-learning algorithm and its guarantees. The appendix proves the probe second-moment identity, converts the dynamical commutator gap into a population covariance gap, establishes the empirical singular-value gap for the learning matrix, and proves robustness to shadow-estimation noise. It also includes an auxiliary procedure for estimating a representative evolution time once the Hamiltonian direction is known.
    
    \item Appendix~\ref{sec:equilibration} explains the weak equilibration consequence of the commutator gap: normalized local observables orthogonal to $H$ cannot retain almost all of their infinite-temperature autocorrelation at the sampled time.
    
    \item Appendix~\ref{Additional Numerics} collects additional numerical experiments supporting the reconstruction and commutator-gap behavior.
\end{itemize}

\section{Related work}\label{Related work}

As mentioned in the main text, the physical starting point for our work is equilibration theory. The results of~\cite{Reimann2008,LindenPopescuShortWinter2009,Short2011Equilibration,ShortFarrelly2012,GogolinEisert2016} give broad conditions under which closed quantum systems equilibrate, controlling expectation values or reduced states in various senses. These results do not, however, provide the uniform inverse-polynomial commutator gaps for local observables that are needed for our learning argument.

A separate line of work concerns Hamiltonian learning, beginning with \cite{BaireyAradLindner2019}, who showed that a local Hamiltonian can be recovered from local measurements on a single eigenstate; subsequent results have established sample-efficient learning of quantum many-body systems~\cite{AnshuArunachalamKuwaharaSoleimanifar2021}, optimal sample complexity for learning from high-temperature Gibbs states and short-time evolutions~\cite{HaahKothariTang2024}, and the first polynomial-time algorithm at arbitrary temperature~\cite{BakshiLiuMoitraTang2024Temperature}. For real-time evolution specifically, \cite{HuangTongFangSu2023} achieved Heisenberg-limited scaling using quantum control, the role of control was analyzed more broadly in \cite{DutkiewiczOBrienSchuster2024}, and \cite{BakshiLiuMoitraTang2024Structure} solved a structure-learning problem in this setting. Related questions of testing and learning structured Hamiltonians, Hamiltonian symmetry testing, and conservation-law discovery were studied in \cite{ArunachalamDuttEscudero2025,LaBordeWilde2022,ZhanElbenHuangTong2024}. We work in a different access model, recovering the Hamiltonian direction from a single time evolution at an unknown, large time, rather than from the controlled short-time queries or multiple chosen evolution times assumed in prior real-time learning work.

\section{Relating $\lbrack U(t),A \rbrack$ to $\lbrack H, A \rbrack$}\label{section-sinc}

The main result of this section is as follows.
\begin{theorem}[Uniform reduction from $\lbrack U(t),A \rbrack$ to $\lbrack H,A \rbrack$]
\label{relation-commutator-A-U}
Let $\mathcal{V}$ be a finite traceless Pauli family, and let
\begin{align}
\mathcal{O}_{\mathcal{V}}^{\mathrm{op}} := \operatorname{span}_{\mathbb{R}}\{P_u : u \in \mathcal{V}\}.
\end{align}
Let $H \in \mathcal{O}_{\mathcal{V}}^{\mathrm{op}}$ be traceless. Throughout, we adopt the normalization convention $\tfrac{1}{d}\tr(H^2)=n$; equivalently $\|H\|_{HS}=\sqrt{n}$ under the normalized Hilbert--Schmidt norm. Let $U(t) = e^{-iHt}$, and define
\begin{align}
\mathcal{A}_{\perp} := \{A \in \mathcal{O}_{\mathcal{V}}^{\mathrm{op}} : \tr(AH) = 0\}.
\end{align}
Let
\begin{align}
\mathcal{W}_H := \{[A,H] : A \in \mathcal{A}_{\perp}\}, \quad r := \dim_{\mathbb{R}}(\mathcal{W}_H) \leq |\mathcal{V}| - 1.
\end{align}
Further, let $0 < \delta \leq 1$. Assume $H\neq 0$, $T>0$, and $r\ge 1$.  If $t \sim \operatorname{Unif}([0,T])$, then with probability at least $1-\delta$ over $t$, one has, simultaneously for every $A \in \mathcal{A}_{\perp}$,
\begin{align}
\|[U(t),A]\|_F^2 \geq \frac{\delta^2\,\min \{1/n, T^2\} }{2048 \, r^2 |\mathcal{V}|}\, \|[A,H]\|_F^2.
\label{eq:uniform-UH-reduction}
\end{align}
Consequently, if
\begin{align}
\min_{\substack{A \in \mathcal{A}_{\perp}\\ \tfrac{1}{d}\tr(A^2)=1}} \frac{1}{d}\|[A,H]\|_F^2 \geq \underline{\pi}_H,
\end{align}
then, with the same probability,
\begin{align}
\min_{\substack{A \in \mathcal{A}_{\perp}\\ \tfrac{1}{d}\tr(A^2)=1}} \frac{1}{d}\|[U(t),A]\|_F^2 \geq \frac{\delta^2 \,\min \{1/n, T^2 \}}{2048 \, r^2 |\mathcal{V}|}\,\underline{\pi}_H.
\label{eq:uniform-piU-from-piH-main}
\end{align}
\end{theorem}

\noindent We will prove the above through a series of lemmas. Write $H$ in its energy basis as
\begin{align}
H = \sum_{p=1}^{d} E_p \, |p\rangle \langle p|, \quad \Delta_{pq} = E_p - E_q.
\end{align}
In this basis,
\begin{align}
[A,H]_{pq} = A_{pq} (E_q - E_p) = - \Delta_{pq} A_{pq}.
\end{align}
Then we have the following, useful lemma.

\begin{lemma}[Sinc decomposition of $\|\lbrack U(t),A \rbrack\|_F^2$]
\label{sinc-bound-equality}
Define $\operatorname{sinc}(x) = \sin(x)/x$ for $x\neq0$, with $\operatorname{sinc}(0)=1$. Then for every $A \in \mathcal{O}_{\mathcal{V}}^{\mathrm{op}}$,
\begin{align}
\bigl\| [U(t),A] \bigr\|_{F}^{2} = t^{2} \sum_{p,q=1}^{d} \bigl| [A,H]_{pq} \bigr|^{2} \operatorname{sinc}^{2}\!\left( \frac{\Delta_{pq} t}{2} \right).
\end{align}
If $[A,H] \neq 0$, then
\begin{align}
\bigl\| [U(t),A] \bigr\|_{F}^{2} = t^{2} \bigl\| [A,H] \bigr\|_{F}^{2} \, \mathbb{E}_{w} \!\left[ \operatorname{sinc}^{2}\!\left( \frac{\Delta_{pq} t}{2} \right) \right],
\end{align}
where
\begin{align}
w_{pq} = \frac{\bigl| [A,H]_{pq} \bigr|^{2}}{\bigl\| [A,H] \bigr\|_{F}^{2}}.
\end{align}
\end{lemma}

\begin{proof}
Since $H = \sum_{p=1}^{d} E_p \, |p\rangle \langle p|$, for any scalar function $f$ one has
\begin{align*}
[f(H),A]_{pq} = \bigl( f(E_p) - f(E_q) \bigr) A_{pq}.
\end{align*}
Applying this with $f(x) = e^{- i t x}$, we obtain
\begin{align*}
[U(t),A]_{pq} &= \bigl( e^{- i E_p t} - e^{- i E_q t} \bigr) A_{pq} = - i t \, \Delta_{pq} \, e^{- \frac{i}{2} (E_p + E_q) t} \operatorname{sinc}\!\left( \frac{\Delta_{pq} t}{2} \right) A_{pq}.
\end{align*}
Using $[A,H]_{pq} = - \Delta_{pq} A_{pq}$, this becomes
\begin{align*}
[U(t),A]_{pq} = i t \, [A,H]_{pq} \, e^{- \frac{i}{2} (E_p + E_q) t} \operatorname{sinc}\!\left( \frac{\Delta_{pq} t}{2} \right).
\end{align*}
Taking absolute values and summing over $p,q$ gives
\begin{align*}
\bigl\| [U(t),A] \bigr\|_{F}^{2} = t^{2} \sum_{p,q=1}^{d} \bigl| [A,H]_{pq} \bigr|^{2} \operatorname{sinc}^{2}\!\left( \frac{\Delta_{pq} t}{2} \right).
\end{align*}
If $[A,H] \neq 0$, factoring out $\bigl\| [A,H] \bigr\|_{F}^{2}$ yields the second identity.
\end{proof}

This lemma isolates the oscillatory part of the problem. For a fixed observable $A$, it expresses $\|[U(t),A]\|_F^2$ as a sum over energy-gap coordinates, with each coordinate weighted by a sinc factor. To sample $t$ first and then optimize over $A$, however, we need a uniform version over the whole finite-dimensional space $\mathcal{W}_H=\{[A,H]:A \in \mathcal{A}_{\perp}\}$. The next lemma proves this uniform lower bound by showing that, with high probability over $t$, the bad energy-gap coordinates cannot capture more than half the Frobenius mass of any vector in $\mathcal{W}_H$.

\begin{lemma}[Uniform lower bound for the sinc factor]
\label{lem:uniform-sinc-factor}
In the setup of Theorem~\ref{relation-commutator-A-U}, write
\begin{align}
H = \sum_p E_p |p\rangle\langle p|, \quad \Delta_{pq} := E_p - E_q, \quad \Delta_{\max} := \max_{p,q}|\Delta_{pq}|.
\end{align}
Fix $0<\theta\leq1$, and sample $t\sim \operatorname{Unif}([0,T])$. Then, with probability at least
\begin{align}
1 - \frac{r\theta}{8} - \frac{r\pi\theta}{2\Delta_{\max}T},
\end{align}
one has, simultaneously for every $A \in \mathcal{A}_{\perp}$,
\begin{align}
\|[U(t),A]\|_F^2 \geq \frac{\theta^2}{128\,\Delta_{\max}^2}\,\|[A,H]\|_F^2.
\label{eq:uniform-sinc-factor-bound}
\end{align}
\end{lemma}

\begin{proof}
By Lemma~\ref{sinc-bound-equality}, for every $A$,
\begin{align}
\|[U(t),A]\|_F^2 = t^2\sum_{p,q}|[A,H]_{pq}|^2 \operatorname{sinc}^2\left(\frac{\Delta_{pq}t}{2}\right).
\label{eq:uniform-sinc-start}
\end{align}
Set
\begin{align}
\tau := \frac{\theta^2}{64\,\Delta_{\max}^2}.
\end{align}
For each pair $(p,q)$, define
\begin{align}
b_{pq}(t) := \mathbf{1}\left\{t^2 \operatorname{sinc}^2\left(\frac{\Delta_{pq}t}{2}\right) < \tau\right\}.
\end{align}
If $\Delta_{pq}=0$, then $B_{pq}=0$ for every $B \in \mathcal{W}_H$, since every $B \in \mathcal{W}_H$ is of the form $B=[A,H]$. Hence these coordinates do not contribute.

Now fix $(p,q)$ with $|\Delta_{pq}|>0$. If $b_{pq}(t)=1$, then
\begin{align}
\frac{4\sin^2(|\Delta_{pq}| t/2)}{|\Delta_{pq}|^2} = t^2 \operatorname{sinc}^2\left(\frac{|\Delta_{pq}| t}{2}\right) < \tau.
\end{align}
Therefore
\begin{align}
\left|\sin\!\left(\frac{|\Delta_{pq}| t}{2}\right)\right| < \frac{|\Delta_{pq}|\sqrt{\tau}}{2} = \frac{\theta|\Delta_{pq}|}{16\Delta_{\max}}.
\end{align}
Since $0<\theta\leq1$ and $|\Delta_{pq}|\leq\Delta_{\max}$, we have 
\begin{equation}
    \frac{\theta|\Delta_{pq}|}{16\Delta_{\max}}\le \frac{1}{16}.
\end{equation}Let $\alpha_{pq}:=\arcsin(\frac{\theta|\Delta_{pq}|}{16\Delta_{\max}})$. The bad set for $y=\Delta t/2$ is contained in
\begin{align}
\bigcup_{m\in\mathbb{Z}}(m\pi-\alpha_{pq},m\pi+\alpha_{pq}).
\end{align}
The variable $y = |\Delta_{pq}|t/2$ is uniform on $[0,L]$, where $L = |\Delta_{pq}|T/2$. At most $L/\pi+2$ of these intervals intersect $[0,L]$, and each has length $2\alpha_{pq}$. Hence
\begin{align}
\mathbb{P}_{t}\big[b_{pq}(t)=1\big] \leq \frac{(L/\pi+2)2\alpha_{pq}}{L} = \frac{2\alpha_{pq}}{\pi} + \frac{4\alpha_{pq}}{L}.
\end{align}
Using $\arcsin(u)\leq \frac{\pi}{2}u$ for $u\in[0,1]$, we obtain
\begin{align}
\mathbb{P}_{t}\big[b_{pq}(t)=1\big] \leq \frac{\theta|\Delta_{pq}|}{16\Delta_{\max}} + \frac{\pi\theta}{4\Delta_{\max}T}.
\label{eq:bad-pair-probability}
\end{align}

Let $\Pi_{\mathrm{bad}}(t)$ be the diagonal projector, in the energy-basis matrix-entry coordinates, that keeps precisely the entries $(p,q)$ with $b_{pq}(t) = 1$.  Let $B_1,\ldots,B_r$ be an orthonormal basis of the real vector space $\mathcal{W}_H$ with respect to the unnormalized Frobenius inner product $\langle X,Y\rangle_F := \operatorname{Re}\tr(X^\dagger Y)$.  Define
\begin{align}
Z(t) := \sum_{j=1}^r \|\Pi_{\mathrm{bad}}(t)B_j\|_F^2.
\end{align}
For each $j$, using \eqref{eq:bad-pair-probability} and $|\Delta_{pq}|\leq\Delta_{\max}$, we have
\begin{align}
\mathbb{E}_t\|\Pi_{\mathrm{bad}}(t)B_j\|_F^2 &= \sum_{p,q}|(B_j)_{pq}|^2\,\mathbb{P}_t[b_{pq}(t)=1] \\
&\leq \frac{\theta}{16} + \frac{\pi\theta}{4\Delta_{\max}T}.
\end{align}
Therefore
\begin{align}
\mathbb{E}_t Z(t) \leq r\left(\frac{\theta}{16} + \frac{\pi\theta}{4\Delta_{\max}T}\right).
\end{align}
By Markov's inequality,
\begin{align}
\mathbb{P}_t\left[Z(t)\geq\frac{1}{2}\right] \leq \frac{r\theta}{8} + \frac{r\pi\theta}{2\Delta_{\max}T}.
\label{eq:bad-trace-markov}
\end{align}
On the complementary event $Z(t)<1/2$, every $B\in\mathcal  W_H$ satisfies
\begin{align}
\|\Pi_{\rm bad}(t)B\|_F^2\le \frac12\|B\|_F^2 .
\end{align}
Indeed, let $P_{\mathcal W_H}$ be the Frobenius-orthogonal projector onto $\mathcal W_H$, and set
\begin{align}
K(t):=P_{\mathcal W_H}\Pi_{\rm bad}(t)P_{\mathcal W_H}:\mathcal W_H \longrightarrow \mathcal W_H .
\end{align}
Then $K(t)\succeq 0$, and, since $B_1,\ldots,B_r$ is Frobenius-orthonormal,
\begin{align}
\tr_{\mathcal W_H}K(t) =
\sum_{j=1}^r \langle B_j,\Pi_{\rm bad}(t) B_j \rangle_F = \sum_{j=1}^r \|\Pi_{\rm bad}(t)B_j\|_F^2 = Z(t).
\end{align}
Therefore, for any $B\in\mathcal  W_H$,
\begin{align}
\|\Pi_{\rm bad}(t)B\|_F^2 = \langle B,K(t)B\rangle_F \le \|K(t)\|_{\rm op}\|B\|_F^2 \le \tr_{\mathcal W_H}K(t)\|B\|_F^2 = Z(t)\|B\|_F^2.
\end{align}
Since $Z(t)<1/2$, the claim follows. Now fix $A\in\mathcal{A}_{\perp}$, and set $B=[A,H]$. If $B=0$, the claim is trivial. Otherwise, on the event $Z(t)<1/2$,
\begin{align}
\|(I-\Pi_{\mathrm{bad}}(t))B\|_F^2 \geq \frac{1}{2}\|B\|_F^2.
\end{align}
On every good coordinate, by definition,
\begin{align}
t^2 \operatorname{sinc}^2\left(\frac{\Delta_{pq}t}{2}\right) \geq \tau.
\end{align}
Using \eqref{eq:uniform-sinc-start}, we get
\begin{align}
\|[U(t),A]\|_F^2 &\geq \tau\,\|(I-\Pi_{\mathrm{bad}}(t))[A,H]\|_F^2 \geq \frac{\tau}{2}\|[A,H]\|_F^2 = \frac{\theta^2}{128\,\Delta_{\max}^2}\|[A,H]\|_F^2.
\end{align}
This proves \eqref{eq:uniform-sinc-factor-bound}.
\end{proof}

The above reduces the proof of the theorem to controlling the largest possible energy gap $\Delta_{\max}$. At this point the lower bound is already of the desired form, up to the appearance of $\Delta_{\max}^{-2}$. The final ingredient is therefore a simple structural estimate relating $\Delta_{\max}$ to the normalization and Pauli support of $H$.

\begin{lemma}[Bounds on $\Delta_{\max}$]
\label{Delta-max-bounds-lemma}
Let $H = \sum_{u \in \mathcal{V}} h_u P_u$ be a traceless Hermitian $n$-qubit Hamiltonian. Then
\begin{align}
2\sqrt{n} \leq \Delta_{\max} \leq 2\sqrt{n|\mathcal{V}|}.
\end{align}
\end{lemma}

\begin{proof}
Let $E_1,\dots,E_d$ be the eigenvalues of $H$. Since $H$ is traceless,
\begin{align}
\frac{1}{d} \sum_{i=1}^{d} E_i = 0.
\end{align}
Recalling the normalization convention $\tfrac{1}{d}\tr(H^2)=n$,
\begin{align}
\frac{1}{d} \tr(H^{2}) = \frac{1}{d} \sum_{i=1}^{d} E_i^{2} = \sum_{u \in \mathcal{V}} h_u^{2} = n.
\end{align}
Thus the eigenvalue distribution has mean $0$ and variance $n$. By Popoviciu's inequality,
\begin{align}
n = \operatorname{Var}(E) \leq \frac{(\max_i E_i - \min_i E_i)^{2}}{4} = \frac{\Delta_{\max}^{2}}{4},
\end{align}
so $\Delta_{\max} \geq 2\sqrt{n} $.
For the upper bound,
\begin{align}
\Delta_{\max} \leq 2 \, \|H\|_{\mathrm{op}} \leq 2 \sum_{u \in \mathcal{V}} |h_u| \, \|P_u\|_{\mathrm{op}} = 2 \sum_{u \in \mathcal{V}} |h_u| \leq 2 \sqrt{|\mathcal{V}|} \left( \sum_{u \in \mathcal{V}} h_u^{2} \right)^{1/2} = 2\sqrt{n|\mathcal{V}|}.
\end{align}
\end{proof}

\noindent We now combine the ingredients above. Lemma~\ref{sinc-bound-equality} gives the exact sinc decomposition, Lemma~\ref{lem:uniform-sinc-factor} gives a uniform lower bound over the whole commutator image $\mathcal{W}_H$, and Lemma~\ref{Delta-max-bounds-lemma} converts the remaining dependence on $\Delta_{\max}$ into a bound in terms of $|\mathcal{V}|$.

\begin{proof}[Proof of Theorem~\ref{relation-commutator-A-U}]
By Lemma~\ref{lem:uniform-sinc-factor}, the failure probability is at most
\begin{align}
\frac{r\theta}{8} + \frac{r\pi\theta}{2\Delta_{\max}T}.
\end{align}
Choose 
\begin{align}
\theta := \frac{\delta}{2r}\min\{1,\sqrt{n}\,T\}.
\end{align}
By Lemma~\ref{Delta-max-bounds-lemma}, $\Delta_{\max}\geq2\sqrt{n}$. The failure probability is at most
\begin{align}
\frac{\delta}{16}\min\{1,\sqrt{n}\,T\} + \frac{\delta \pi}{8\sqrt{n}\,T}\min \{1,\sqrt{n}\,T\}.
\end{align}
If $\sqrt{n}\,T\le 1$ this is at most
\begin{align} \frac{\delta}{16}\sqrt{n}\,T+\frac{\delta \pi}{8\sqrt{n}\,T}\sqrt{n}\,T\le \delta\Big(\frac{1}{16}+\frac{\pi}{8}\Big)< \delta.\end{align}
And if $\sqrt{n}\,T\ge 1$ the failure probability is at most
\begin{align} \frac{\delta}{16}+\frac{\pi \delta}{8\sqrt{n}\,T}< \delta.\end{align}
Thus, with probability at least $1-\delta$, Lemma~\ref{lem:uniform-sinc-factor} gives
\begin{align}
\|[U(t),A]\|_F^2 \geq \frac{\theta^2}{128\,\Delta_{\max}^2}\,\|[A,H]\|_F^2
\end{align}
simultaneously for every $A\in\mathcal{A}_{\perp}$. By Lemma~\ref{Delta-max-bounds-lemma}, $\Delta_{\max}\leq2\sqrt{n|\mathcal{V}|}$, so
\begin{align}
\|[U(t),A]\|_F^2 \geq \frac{\theta^2}{512|\mathcal{V}|n}\,\|[A,H]\|_F^2
\end{align}
simultaneously for every $A\in\mathcal{A}_{\perp}$. Plugging in the value of $\theta$ proves \eqref{eq:uniform-UH-reduction}. Dividing by $d$, and then applying the assumed static lower bound, proves \eqref{eq:uniform-piU-from-piH-main}.
\end{proof}

\section{Commutator gaps}
\label{sec:commutator-gaps}

In this section we prove lower bounds on $\frac{1}{d}\|[A,H]\|_F^2$ for local observables $A$ orthogonal to the Hamiltonian. The results are stated for real spans of Hermitian Pauli strings. Thus, whenever $\mathcal{V}$ is a Pauli family, we assume that $\mathcal{V}$ consists of distinct Pauli strings modulo phase, with Hermitian representatives chosen so that $\frac{1}{d}\tr(P_u P_v) =  \delta_{u,v}$. We will use the notation $\mathcal{O}_{\mathcal{V}}^{\mathrm{op}} := \operatorname{span}_{\mathbb{R}}\{P_u:u\in\mathcal{V}\}$ from Theorem~\ref{relation-commutator-A-U}.

There are two regimes that we consider. In the \emph{Pauli-injective regime}, the generic Pauli family $\mathcal V$ is specialized to a support $\mathcal V_{\mathrm{inj}}$ whose anticommutation graph is connected and whose anticommuting pairs produce distinct Pauli strings, up to phase. In this case, the commutator decomposes into orthogonal contributions indexed by the edges of the anticommutation graph, and the commutator gap follows from a graph Poincar\'{e} inequality. In the \emph{geometrically local regime}, the generic family is specialized to a finite-range local support $\mathcal V_{k,R}$. Here different anticommuting pairs may produce the same Pauli string, so the proof instead works with a collision-aware commutator matrix, decomposes its Gram form into local positive semidefinite pieces, and patches the resulting local gaps.

\subsection{Pauli-injective commutator gaps}
\label{subsec:pauli-injective-commutator-gaps}

Let $\mathcal{V}_{\mathrm{inj}}$ be a finite Pauli family. Its anticommutation graph is $\mathcal{G} = (\mathcal{V}_{\mathrm{inj}},\mathsf{E})$ where we take $\{u,v\}\in\mathsf{E}$ if and only if $P_u P_v = - P_v P_u$. The injectivity condition needed below is only an injectivity condition on anticommuting edges, not on the full product map $(u,v)\mapsto P_uP_v$.

\begin{definition}[Edge-product injectivity]
\label{def:edge-product-injectivity}
Choose one orientation $e=(u_e,v_e)$ for each edge $e\in\mathsf{E}$, and define
\begin{align}
Q_e := P_{u_e} P_{v_e}.
\end{align}
We say that $\mathcal{V}_{\mathrm{inj}}$ is edge-product injective if distinct anticommuting edges produce distinct Pauli products, up to phase. Equivalently, after choosing one orientation $e = (u_e,v_e)$ for each anticommuting edge and setting $Q_e := P_{u_e}P_{v_e}$, we require
\begin{align}
\frac{1}{d}\tr(Q_e^\dagger Q_{e'}) = \delta_{e,e'}
\end{align}
for all edges $e,e'$. Equivalently, the map from anticommuting unordered pairs $\{u,v\}$ to the Pauli string $P_u P_v$ modulo phase is injective.
\end{definition}

\begin{lemma}[Exact commutator quadratic form]
\label{lem:exact-commutator-quadratic-form}
Let $\mathcal{V}_{\mathrm{inj}}$ be edge-product injective. Let
\begin{align}
H = \sum_{u\in\mathcal{V}_{\mathrm{inj}}} h_u P_u, \quad A = \sum_{u\in\mathcal{V}_{\mathrm{inj}}}a_uP_u,
\end{align}
with $a_u, h_u \in \mathbb{R}$. Define the positive semidefinite quadratic form $\Gamma(h)$ by
\begin{align}
a^\top\Gamma(h)a := \sum_{e\in\mathsf{E}} (a_{u_e} h_{v_e} - a_{v_e} h_{u_e})^2.
\end{align}
Then
\begin{align}
\frac{1}{d}\|[A,H]\|_F^2 = 4 a^\top \Gamma(h) a.
\label{eq:pauli-inj-exact-form}
\end{align}
Moreover, if $\mathcal{G}$ is connected and $h_u\neq0$ for every $u\in\mathcal{V}_{\mathrm{inj}}$, then
\begin{align}
\ker\Gamma(h) = \operatorname{span}\{h\}.
\end{align}
\end{lemma}

\begin{proof}
Only anticommuting pairs contribute to the commutator. For an oriented edge $e=(u_e,v_e)$,
\begin{align}
a_{u_e} h_{v_e} [P_{u_e},P_{v_e}] + a_{v_e} h_{u_e}[P_{v_e},P_{u_e}] &= 2 a_{u_e} h_{v_e} P_{u_e} P_{v_e} + 2 a_{v_e} h_{u_e} P_{v_e} P_{u_e} \\
&= 2(a_{u_e} h_{v_e} - a_{v_e} h_{u_e}) Q_e.
\end{align}
Therefore
\begin{align}
[A,H] = 2\sum_{e \in \mathsf{E}} (a_{u_e} h_{v_e} - a_{v_e} h_{u_e}) Q_e.
\end{align}
By edge-product injectivity, the $Q_e$'s are orthonormal for the normalized Hilbert--Schmidt inner product. This proves \eqref{eq:pauli-inj-exact-form}.

If $a^\top\Gamma(h)a=0$, then for every edge $e=(u_e,v_e)$,
\begin{align}
a_{u_e} h_{v_e} - a_{v_e} h_{u_e} = 0.
\end{align}
Since all $h_u$'s are nonzero, this is equivalent to
\begin{align}
\frac{a_{u_e}}{h_{u_e}} = \frac{a_{v_e}}{h_{v_e}}.
\end{align}
Connectivity implies that $a_u/h_u$ is constant on all vertices of $\mathcal{G}$. Hence $a\in\operatorname{span}\{h\}$. Conversely, if $a\in\operatorname{span}\{h\}$, then each term
$a_{u_e} h_{v_e} - a_{v_e} h_{u_e}$ vanishes, so $a \in \ker\Gamma(h)$.
\end{proof}

Our goal is to get a lower bound on $\frac{1}{d}\|[A,H]\|_F^2$, and so in light of~\eqref{eq:pauli-inj-exact-form} it is prudent to study properties of $\Gamma(h)$.  To this end, a useful and standard result is the Poincar\'{e} inequality, which we state for connected graphs below.

\begin{lemma}[Connected graph Poincar\'{e} inequality]
\label{lem:connected-graph-poincare}
Let $\mathcal{G} = (\mathcal{V}_{\mathrm{inj}},\mathsf{E})$ be a connected unweighted graph with $N := |\mathcal{V}_{\mathrm{inj}}| \geq 2$, and let $L_{\mathcal{G}}$ be its graph Laplacian. Let $\lambda_2(L_{\mathcal G})$ denote the second-smallest eigenvalue of $L_{\mathcal G}$. Then
\begin{align}
\lambda_2(L_{\mathcal{G}}) \geq \frac{4}{N^2}.
\end{align}
\end{lemma}

\begin{proof}
Let $\mathbf{1} \in \mathbb{R}^{\mathcal{V}_{\mathrm{inj}}}$ denote the all-ones vector, and let $x \perp \mathbf{1}$. Write
\begin{align}
x_{\max} := \max_{u\in\mathcal{V}_{\mathrm{inj}}} x_u, \quad x_{\min} := \min_{u\in\mathcal{V}_{\mathrm{inj}}} x_u.
\end{align}
Since all coordinates lie in $[x_{\min},x_{\max}]$ and $\sum_u x_u=0$,
\begin{align}
\frac{1}{N}\|x\|_2^2 \leq \frac{(x_{\max}-x_{\min})^2}{4}.
\label{eq:range-variance-bound}
\end{align}
Choose vertices attaining $x_{\max}$ and $x_{\min}$, and choose a simple path between them. Its length is at most $N-1$. By Cauchy--Schwarz along the path,
\begin{align}
(x_{\max}-x_{\min})^2 \leq (N-1) x^\top L_{\mathcal{G}} x.
\label{eq:path-range-bound}
\end{align}
Combining \eqref{eq:range-variance-bound} and \eqref{eq:path-range-bound},
\begin{align}
x^\top L_{\mathcal{G}} x \geq \frac{4}{N(N-1)}\|x\|_2^2 \geq \frac{4}{N^2}\|x\|_2^2.
\end{align}
Taking the minimum over $x\perp\mathbf{1}$ proves the result.
\end{proof}

We now apply this graph inequality to the anticommutation graph of an edge-product-injective Pauli family.

\begin{theorem}[Pauli-injective commutator gap]
\label{thm:pauli-injective-gap}
Let $\mathcal{V}_{\mathrm{inj}}$ be a finite Pauli family with connected anticommutation graph $\mathcal{G}=(\mathcal{V}_{\mathrm{inj}},\mathsf{E})$. Assume $\mathcal{V}_{\mathrm{inj}}$ is edge-product injective. Let
\begin{align}
H = \sum_{u\in\mathcal{V}_{\mathrm{inj}}} h_u P_u \,\, \text{with } h_u \in \mathbb{R}, \quad h_{\min} := \min_{u\in\mathcal{V}_{\mathrm{inj}}}|h_u| > 0
, \quad h_{\max} := \max_{u\in\mathcal{V}_{\mathrm{inj}}}|h_u|,
\end{align}
and assume the normalization $\tfrac{1}{d}\tr(H^2)=\sum_{u\in\mathcal{V}_{\mathrm{inj}}} h_u^2 = n$.
Further let $N := |\mathcal{V}_{\mathrm{inj}}| \geq 2$. Then
\begin{align}
\min_{\substack{ A \in \mathcal{O}_{\mathcal{V}_{\mathrm{inj}}}^{\mathrm{op}}\\ \tfrac{1}{d}\tr(A^2)=1,\ \tfrac{1}{d}\tr(AH)=0 }} \frac{1}{d}\|[A,H]\|_F^2
 \ge \frac{16\,h_{\min}^4}{h_{\max}^2 N^2\bigl(1+\frac{Nh_{\max}^2}{n}\bigr)}.
\label{eq:pauli-inj-gap-bound}
\end{align}
\end{theorem}

\begin{proof}
Write $A = \sum_{u \in \mathcal{V}_{\mathrm{inj}}} a_u P_u$. Since the Pauli strings are orthonormal,
\begin{align}
\frac{1}{d}\tr(A^2) = \|a\|_2^2\,, \quad \frac{1}{d}\tr(AH) =  a\cdot h.
\end{align}
Thus the constraints are $\|a\|_2=1$ and $a\perp h$.

Define $c_u := \frac{a_u}{h_u}$. Then we have $\sum_{u \in\mathcal{V}_{\mathrm{inj}}} c_u h_u^2 = 0$ and $\sum_{u \in\mathcal{V}_{\mathrm{inj}}} c_u^2 h_u^2 = 1$. We further let $\bar{c} := \frac{1}{N}\sum_{u\in\mathcal{V}_{\mathrm{inj}}} c_u$ and $c^\perp := c - \bar{c}\,\mathbf{1}$. The weighted mean-zero condition gives
\begin{align}
\bar{c}\,\sum_{u \in\mathcal{V}_{\mathrm{inj}}} h_u^2 =  - \sum_{u \in\mathcal{V}_{\mathrm{inj}}} c_u^\perp h_u^2.
\end{align}
By Cauchy--Schwarz, we find
\begin{align}\label{cauchy-c-bar}
|\bar{c}| \,\sum_{u\in\mathcal{V}_{\mathrm{inj}}} h_u^2 \leq \|c^\perp\|_2 \left(\sum_{u \in\mathcal{V}_{\mathrm{inj}}} h_u^4 \right)^{1/2}.
\end{align}
Let $S_4=\sum_u h_u^4$ and recall from normalization $\sum_u h_u^2=n$. Squaring and multiplying by $N$, \eqref{cauchy-c-bar} gives $N \bar{c}^2 \leq \frac{NS_4}{n^2}\|c^\perp\|_2^2$, and therefore\label{c-norm-lower-bound-c-bar}
\begin{align}
\|c\|_2^2 = \|c^\perp\|_2^2 + N \bar{c}^2 \leq \left(1 + \frac{NS_4}{n^2}\right) \|c^\perp\|_2^2.
\end{align}
On the other hand,
\begin{align}\label{c-norm-lower-bound}
1 = \sum_{u\in\mathcal{V}_{\mathrm{inj}}} c_u^2 h_u^2 \leq h_{\max}^2\|c\|_2^2\,,
\end{align}
Combining \eqref{c-norm-lower-bound-c-bar} and \eqref{c-norm-lower-bound} gives
\begin{align}
\|c^\perp\|_2^2 \geq \frac{1}{h_{\max}^2(1+\frac{NS_4}{n^2})}.
\label{eq:pauli-cperp-lower}
\end{align}
By Lemma~\ref{lem:exact-commutator-quadratic-form}, we find
\begin{align}
a^\top\Gamma(h)a &= \sum_{\{u,v\}\in\mathsf{E}} h_u^2 h_v^2 (c_u-c_v)^2 \geq h_{\min}^4 \sum_{\{u,v\}\in\mathsf{E}}(c_u-c_v)^2 = h_{\min}^4(c^\perp)^\top L_{\mathcal{G}} c^\perp\,.
\end{align}
Since $c^\perp \perp \mathbf{1}$, Lemma~\ref{lem:connected-graph-poincare} and \eqref{eq:pauli-cperp-lower} imply
\begin{align}
a^\top\Gamma(h)a \geq \frac{4h_{\min}^4}{h_{\max}^2(1+\frac{NS_4}{n^2})N^2}\,.
\end{align}
Multiplying by $4$ via \eqref{eq:pauli-inj-exact-form} and using that $S_4=\sum_u h_u^4\le h_{\max}^2 \sum_u h_u^2=h_{\max}^2 n$, we obtain $1+\frac{NS_4}{n^2}\le 1+\frac{Nh_{\max}^2}{n}$ (by the normalization), which proves \eqref{eq:pauli-inj-gap-bound}.
\end{proof}

The preceding theorem controls observables supported on the Hamiltonian support $\mathcal{V}_{\mathrm{inj}}$. In applications, one may want to search over a larger Pauli family $\mathcal{V}_{\max} \supseteq \mathcal{V}_{\mathrm{inj}}$. The following extension gives a sufficient condition under which the same lower bound holds on the larger span. The two required hypotheses are that every added Pauli direction anticommutes with at least one Hamiltonian term, and that the corresponding commutator products do not cancel.

\begin{definition}[Relative edge-product injectivity]
\label{def:relative-edge-product-injectivity}
Let $\mathcal{V}_{\mathrm{inj}} \subseteq \mathcal{V}_{\max}$, and set $\mathcal{F} := \mathcal{V}_{\max} \setminus \mathcal{V}_{\mathrm{inj}}$. Choose one orientation for every anticommuting edge internal to $\mathcal{V}_{\mathrm{inj}}$, and denote the resulting oriented edge set by $\mathsf{E}_{\mathcal{V}_{\mathrm{inj}}}^{\rightarrow}$. Define the labelled index set
\begin{align}
\mathcal{I} := \mathsf{E}_{\mathcal{V}_{\mathrm{inj}}}^{\rightarrow} \sqcup \{ (f, u) \in \mathcal{F} \times \mathcal{V}_{\mathrm{inj}} : P_f P_u = - P_u P_f \}.
\end{align}
For $e = (u_e, v_e) \in \mathsf{E}_{\mathcal{V}_{\mathrm{inj}}}^{\rightarrow}$, set $Q_e := P_{u_e}P_{v_e}$. For $(f, u)$ in the second component, set $Q_{(f, u)} := P_fP_u$. We say that $\mathcal{V}_{\mathrm{inj}} \subseteq \mathcal{V}_{\max}$ is relatively edge-product injective if
\begin{align}
\frac{1}{d}\tr(Q_\alpha^\dagger Q_\beta) = \delta_{\alpha,\beta}
\end{align}
for all $\alpha, \beta \in \mathcal{I}$.
\end{definition}

\begin{lemma}[Commutator gap extension from a dominating subfamily]
\label{lem:dominating-extension}
Let $\mathcal{V}_{\mathrm{inj}} \subseteq \mathcal{V}_{\max}$, and let
\begin{align}
H = \sum_{u \in \mathcal{V}_{\mathrm{inj}}} h_u P_u
\end{align}
(normalized as $\tfrac{1}{d}\tr(H^2) = n$ as usual), satisfy the assumptions of Theorem~\ref{thm:pauli-injective-gap} on $\mathcal{V}_{\mathrm{inj}}$. Assume that $\mathcal{V}_{\mathrm{inj}}$ dominates $\mathcal{V}_{\max}$, meaning that for every $f \in \mathcal{V}_{\max} \setminus \mathcal{V}_{\mathrm{inj}}$, there exists $u \in \mathcal{V}_{\mathrm{inj}}$ such that
\begin{align}
P_f P_u = -P_u P_f.
\end{align}
Assume also that $\mathcal{V}_{\mathrm{inj}} \subseteq \mathcal{V}_{\max}$ is relatively edge-product injective. Then
\begin{align}
\min_{\substack{A \,\in\, \operatorname{span}_{\mathbb{R}}\{ P_x : x \in \mathcal{V}_{\max} \}\\ \frac{1}{d}\tr(A^2) = 1,\ \tr(AH) = 0}} \frac{1}{d}\|[A,H]\|_F^2 \geq \frac{16\,h_{\min}^4}{h_{\max}^2 N^2\bigl(1+\frac{Nh_{\max}^2}{n}\bigr)}.
\end{align}
\end{lemma}

\begin{proof}
Set $\mathcal{F} := \mathcal{V}_{\max} \setminus \mathcal{V}_{\mathrm{inj}}$, and decompose
\begin{align}
A = A_{\mathcal{V}_{\mathrm{inj}}} + A_{\mathcal{F}}, \quad A_{\mathcal{V}_{\mathrm{inj}}} = \sum_{u \in \mathcal{V}_{\mathrm{inj}}} a_u P_u, \quad A_{\mathcal{F}} = \sum_{f \in \mathcal{F}} a_f P_f.
\end{align}
Let
\begin{align}
\lambda := \sum_{u \in \mathcal{V}_{\mathrm{inj}}} a_u^2, \quad 1 - \lambda := \sum_{f \in \mathcal{F}} a_f^2.
\end{align}
Since $H$ is supported on $\mathcal{V}_{\mathrm{inj}}$, the condition $\tr(AH) = 0$ implies $\tr(A_{\mathcal{V}_{\mathrm{inj}}}H) = 0$. Relative edge-product injectivity gives orthogonality between the internal and cross commutator products. Therefore
\begin{align}
\frac{1}{d}\|[A,H]\|_F^2 = \frac{1}{d}\|[A_{\mathcal{V}_{\mathrm{inj}}},H]\|_F^2 + \frac{1}{d}\|[A_{\mathcal{F}},H]\|_F^2.
\label{eq:dominating-extension-split}
\end{align}
If $\lambda > 0$, then $A_{\mathcal{V}_{\mathrm{inj}}}/\sqrt{\lambda}$ is normalized and orthogonal to $H$, so Theorem~\ref{thm:pauli-injective-gap} gives
\begin{align}
\frac{1}{d}\|[A_{\mathcal{V}_{\mathrm{inj}}},H]\|_F^2 \geq \lambda \pi_0, \quad \pi_0 := \frac{16\,h_{\min}^4}{h_{\max}^2 N^2\bigl(1+\frac{Nh_{\max}^2}{n}\bigr)}.
\end{align}
The same inequality is trivial when $\lambda = 0$.

For the cross contribution, only anticommuting pairs contribute, so
\begin{align}
[A_{\mathcal{F}},H] = 2 \sum_{\substack{f \in \mathcal{F},\, u \in \mathcal{V}_{\mathrm{inj}}\\ P_fP_u = -P_uP_f}} a_f h_u P_fP_u.
\end{align}
By relative edge-product injectivity, the products $P_fP_u$ in this sum are mutually orthonormal. Hence
\begin{align}
\frac{1}{d}\|[A_{\mathcal{F}},H]\|_F^2 = 4 \sum_{f \in \mathcal{F}} a_f^2 \sum_{\substack{u \in \mathcal{V}_{\mathrm{inj}}\\ P_fP_u = -P_uP_f}} h_u^2.
\end{align}
By the dominance assumption, for each $f \in \mathcal F$ there is at least one $u \in \mathcal V_{\mathrm{inj}}$ that anticommutes with $f$. Hence each inner sum is nonempty, and since every nonzero Hamiltonian coefficient on $\mathcal V_{\mathrm{inj}}$ has magnitude at least $h_{\min}$, each inner sum is at least $h_{\min}^2$. Thus
\begin{align}
\frac{1}{d}\|[A_{\mathcal{F}},H]\|_F^2 \geq 4 h_{\min}^2 (1 - \lambda).
\end{align}
Combining this with \eqref{eq:dominating-extension-split}, we get
\begin{align}
\frac{1}{d}\|[A,H]\|_F^2 \geq \lambda \pi_0 + 4 h_{\min}^2 (1 - \lambda).
\end{align}
Finally, since $|\mathcal{V}_{\mathrm{inj}}| \geq 2$ and $h_{\min} \leq h_{\max}$,
\begin{align}
\pi_0 &= \frac{16\,h_{\min}^4}{h_{\max}^2 N^2\bigl(1+\frac{Nh_{\max}^2}{n}\bigr)}\leq \frac{4 h_{\min}^4}{h_{\max}^2}=4h_{\min}^2\Big(\frac{h_{\min}^2}{h_{\max}^2}\Big)  \leq 4 h_{\min}^2\,.
\end{align}
Therefore $\lambda \pi_0 + 4 h_{\min}^2 (1 - \lambda) \geq \pi_0$, which proves the claim.
\end{proof}

\subsection{Geometrically local commutator gaps}
\label{subsec:geometric-local-nondegeneracy}

We now turn to more general geometrically local Hamiltonians. The main difference from the Pauli-injective setting is that distinct anticommuting pairs of Pauli strings can produce the same output Pauli string. Thus the commutator need not decompose as an orthogonal sum over anticommuting edges. Instead, we group together all contributions that produce the same output Pauli string and work with the resulting commutator matrix.

Let $G = (\Lambda, \textsf{E})$ be a connected simple interaction graph with $|\Lambda| = n$ and maximum degree $\Delta$. We assume $\Delta = O(1)$. Let $\operatorname{dist}_G$ denote graph distance, and define
\begin{align}
B(q, r) := \{ p \in \Lambda : \operatorname{dist}_G(p, q) \leq r \}, \quad b_r := \max_{q \in \Lambda} |B(q, r)|.
\end{align}
For $S \subseteq \Lambda$, write
\begin{align}
\operatorname{diam}_G(S) := \max_{p, q \in S} \operatorname{dist}_G(p, q).
\end{align}
Now for fixed $k, R = O(1)$, with $R \geq 1$, let $\mathcal{V}_{k,R}$ be a finite family of distinct Hermitian Pauli strings modulo phase satisfying
\begin{align}
1 \leq |\operatorname{supp}(P_u)| \leq k, \quad \operatorname{diam}_G(\operatorname{supp}(P_u)) \leq R.
\end{align}
Set
\begin{align}
\mathcal{O}_{k,R}^{\mathrm{op}} := \mathcal{O}^{\mathrm{op}}_{\mathcal{V}_{k,R}} =\operatorname{span}_{\mathbb{R}}\{ P_u : u \in \mathcal{V}_{k,R} \}\,,
\end{align}
and let
\begin{align}
H = \sum_{u \in \mathcal{V}_{k,R}} h_u P_u, \quad A = \sum_{u \in \mathcal{V}_{k,R}} a_u P_u.
\end{align}
As before, we define
\begin{align}
h_{\min} := \min_{u \in \mathcal{V}_{k,R}} |h_u|, \quad h_{\max} := \max_{u \in \mathcal{V}_{k,R}} |h_u|.
\end{align}

\subsubsection{Global commutator matrix}

For $u, v \in \mathcal{V}_{k,R}$, write
\begin{align}
P_u P_v = (-1)^{\chi(u, v)} P_v P_u, \quad \chi(u, v) \in \{0,1\}.
\end{align}
If $\chi(u, v) = 1$, define $u \oplus v$ and $\sigma_{uv} \in \{\pm 1\}$ by
\begin{align}
[P_u, P_v] = 2\,i\,\sigma_{uv} P_{u \oplus v}\,,
\end{align}
where $P_{u \oplus v}$ is the Hermitian Pauli string obtained after removing the phase from the product. Let
\begin{align}
\mathcal{W}_{\oplus} := \{ u \oplus v : u, v \in \mathcal{V}_{k,R},\ \chi(u, v) = 1 \}.
\end{align}
Thus $\mathcal{W}_{\oplus}$ indexes the Pauli strings that can appear in the commutator. For each $w \in \mathcal{W}_{\oplus}$, define
\begin{align}
b_w := \sum_{\substack{u, v \in \mathcal{V}_{k,R}\\ u \oplus v = w\\ \chi(u, v) = 1}} \sigma_{uv} h_v a_u.
\end{align}
Equivalently, $b = B(h)a$, where $B(h)$ is the real matrix with rows indexed by $\mathcal{W}_{\oplus}$ and columns indexed by $\mathcal{V}_{k,R}$, defined by
\begin{align}\label{eq:def-B-matrix}
B(h)_{w,u} := \sum_{\substack{v \in \mathcal{V}_{k,R}\\ u \oplus v = w\\ \chi(u, v) = 1}} \sigma_{uv} h_v.
\end{align}
With this notation, all contributions producing the same output Pauli string have been combined, and the commutator is
\begin{align}
[A,H] = 2\,i\sum_{w \in \mathcal{W}_{\oplus}} b_w P_w.
\end{align}
Since the $P_w$'s are distinct Pauli strings modulo phase, they are orthonormal for the normalized Hilbert--Schmidt inner product. Therefore
\begin{align}
\frac{1}{d}\|[A,H]\|_F^2 = 4\|B(h)a\|_2^2 = 4 a^\top \Gamma(h) a\,, \quad \Gamma(h) := B(h)^\top B(h).
\label{eq:geo-global-commutator-identity}
\end{align}
Finally, since $[H,H] = 0$, substituting $a = h$ into the preceding commutator representation gives $B(h)h = 0$, and hence $\Gamma(h)h = 0$. It follows that
\begin{align}
\min_{\substack{A \in \mathcal{O}_{k,R}^{\mathrm{op}}\\ \tfrac{1}{d}\tr(A^2) = 1,\ \tr(AH) = 0}} \frac{1}{d}\|[A,H]\|_F^2 = 4 \min_{\substack{a \perp h\\ \|a\|_2 = 1}} a^\top \Gamma(h) a\,.
\label{eq:geo-gap-rayleigh}
\end{align}

\subsubsection{Local decomposition}

We next decompose the global commutator Gram matrix $\Gamma(h)$ into local positive semidefinite pieces. The purpose of this decomposition is to reduce the global gap problem to bounded-size local gap estimates, which can later be patched together using connectivity of the interaction graph.

Fix a support-respecting anchoring map $\alpha: \mathcal{V}_{k,R} \to \Lambda$, where $\alpha(u) \in \operatorname{supp}(P_u)$ for every $u \in \mathcal{V}_{k,R}$. We write
\begin{equation}
\operatorname{anc}(u) := \alpha(u).
\end{equation}
This gives a disjoint partition
\begin{equation}
\mathcal{V}_{k,R} = \bigsqcup_{q \in \Lambda} \mathcal{V}_q, \quad \mathcal{V}_q := \{ u \in \mathcal{V}_{k,R} : \operatorname{anc}(u) = q \}.
\end{equation}
We assume that the anchoring map has been chosen so that $\mathcal{V}_q \neq \varnothing$ for every $q \in \Lambda$. This nonempty-block condition is used only in the patching argument. The only other property of the anchoring map used below is that the anchor of each operator lies in the support of that operator.

\begin{lemma}[Nonempty anchoring for exact two-body finite-range families]
\label{lem:exact-two-body-nonempty-anchoring}
Let $G = (\Lambda,\mathsf E)$ be connected with $|\Lambda| \geq 2$, and let
\begin{equation}
\mathcal V^{(2)}_{R} := \{ \sigma_i^a \sigma_j^b : 1 \leq \operatorname{dist}_G(i,j) \leq R,\ a,b \in \{X,Y,Z\} \}
\end{equation}
with $R \geq 1$. Then there exists a support-respecting anchoring map $\alpha: \mathcal V^{(2)}_{R} \to \Lambda$ such that every anchor block $\mathcal V_q = \{ u : \alpha(u) = q \}$ is nonempty.
\end{lemma}

\begin{proof}
For each vertex $q \in \Lambda$, choose one neighbor $p(q)$ of $q$, which exists because $G$ is connected and $|\Lambda| \geq 2$. The edge $\{q,p(q)\}$ has graph distance one, hence is allowed because $R \geq 1$. For every unordered edge $\{p,q\}$ that is selected by both of its endpoints, assign two distinct Pauli labels, for example $X_p X_q$ to one endpoint and $Y_p Y_q$ to the other. If the edge is selected by only one endpoint, assign $X_p X_q$ to that endpoint. These selected Pauli strings are distinct and each contains the vertex to which it is assigned. Anchor every remaining Pauli string in $\mathcal V^{(2)}_R$ arbitrarily at one of the sites in its support. This defines a support-respecting anchoring map, and by construction each vertex anchors at least one selected Pauli string.
\end{proof}

The block sizes are uniformly bounded. Indeed, if $\operatorname{anc}(u) = q$ and $|\operatorname{supp}(P_u)| = s$, then $q \in \operatorname{supp}(P_u)$ by the support-respecting property of the anchoring map, and the range condition implies $\operatorname{supp}(P_u) \subseteq B(q, R)$. Therefore
\begin{align}
|\mathcal{V}_q| \leq d_{\mathrm{blk}} := \sum_{s = 1}^{k} \binom{b_R - 1}{s - 1} 3^s.
\label{eq:geo-dblk}
\end{align}
For $c \in \Lambda$, define the local coordinate neighborhood
\begin{align}
\mathcal{U}_c := \{ u \in \mathcal{V}_{k,R} : \operatorname{supp}(P_u) \cap B(c, R) \neq \varnothing \}\,,
\end{align}
and let
\begin{align}
d_c := |\,\mathcal{U}_c|, \quad d_{\mathrm{loc}} := \max_{c \in \Lambda} d_c\,.
\end{align}
Then we have
\begin{align}
d_{\mathrm{loc}} \leq b_{2R} \, d_{\mathrm{blk}} = O_{\Delta,k,R}(1).
\label{eq:geo-dloc}
\end{align}
To see this, suppose $u \in \mathcal{U}_c$. Then some $s \in \operatorname{supp}(P_u)$ lies in $B(c, R)$. If $q = \operatorname{anc}(u)$, then $\operatorname{dist}_G(q, s) \leq R$, and hence $\operatorname{dist}_G(q, c) \leq 2R$. Thus there are at most $b_{2R}$ possible anchors $q$, and for each such anchor there are at most $d_{\mathrm{blk}}$ possible Pauli strings.

We assume
\begin{align}
d_c \geq 2 \quad \textnormal{for every } c \in \Lambda.
\label{eq:geo-dc-at-least-two}
\end{align}
For $w \in \mathcal{W}_{\oplus}$, let $r_w(h) \in \mathbb{R}^{\mathcal{V}_{k,R}}$ denote the transpose of the $w$-th row of $B(h)$, so that
\begin{align}
\Gamma(h) = \sum_{w \in \mathcal{W}_{\oplus}} r_w(h) r_w(h)^\top\,,
\end{align}
and for $c \in \Lambda$ define
\begin{align}
\mathcal{W}_c := \{ w \in \mathcal{W}_{\oplus} : c \in \operatorname{supp}(P_w) \}.
\end{align}
We also define
\begin{align}
w_{\mathrm{loc}} := \max_{c \in \Lambda} |\mathcal{W}_c|.
\end{align}
For fixed $(\Delta,k,R)$, the number $w_{\mathrm{loc}}$ is $O_{\Delta,k,R}(1)$. Indeed, if $w \in \mathcal{W}_c$, then $c \in \operatorname{supp}(P_w)$. Moreover, $w = u \oplus v$ for anticommuting $u, v \in \mathcal{V}_{k,R}$, so $\operatorname{supp}(P_u)\cap\operatorname{supp}(P_v)\neq\varnothing$. Since both supports have graph diameter at most $R$, the support of $P_w$ is contained in $B(c, 2R)$, and $|\operatorname{supp}(P_w)| \leq 2k - 1$. Thus $w_{\mathrm{loc}}$ is bounded by a constant depending only on $(\Delta,k,R)$.

The local contribution to $\Gamma(h)$ at $c$ is then
\begin{align}
\Gamma_c(h) := \sum_{w \in \mathcal{W}_c} \frac{1}{|\operatorname{supp}(P_w)|} \, r_w(h) r_w(h)^\top.
\label{eq:geo-Gammac}
\end{align}
At this point $\Gamma_c(h)$ is a matrix on the full coordinate space $\mathbb{R}^{\mathcal{V}_{k,R}}$. The next lemma shows that it is supported only on $\mathcal{U}_c$, so it may be identified with a $d_c \times d_c$ matrix.

\begin{lemma}[Local decomposition and row localization]
\label{lem:geo-local-decomposition}
The matrices $\Gamma_c(h)$ satisfy
\begin{align}
\Gamma(h) = \sum_{c \in \Lambda} \Gamma_c(h).
\end{align}
Moreover, every row vector entering $\Gamma_c(h)$ is supported on the coordinate set $\mathcal{U}_c$. Thus $\Gamma_c(h)$ may be identified with its restriction to $\mathbb{R}^{\mathcal{U}_c}$. Under this identification,
\begin{align}
\Gamma_c(h) h^{(\mathcal{U}_c)} = 0,
\end{align}
where $h^{(\mathcal{U}_c)}$ denotes the restriction of the coefficient vector $h$ to the coordinates in $\mathcal{U}_c$.
\end{lemma}

\begin{proof}
The identity $\Gamma(h) = \sum_{c \in \Lambda} \Gamma_c(h)$ follows from
\begin{align}
\sum_{c \in \Lambda} \mathbf{1}_{\{c \,\in\, \operatorname{supp}(P_w)\}} \frac{1}{|\operatorname{supp}(P_w)|} = 1
\end{align}
for every $w \in \mathcal{W}_{\oplus}$.

It remains to show that $\Gamma_c(h)$ only involves coordinates in $\mathcal U_c$. Suppose $c \in \operatorname{supp}(P_w)$ and $(r_w(h))_u \neq 0$. By the definition of $B(h)$, there exists $v \in \mathcal{V}_{k,R}$ such that $u \oplus v = w$ and $\chi(u, v) = 1$. Since $P_u$ and $P_v$ anticommute, their supports intersect. Choose $s \in \operatorname{supp}(P_u) \cap \operatorname{supp}(P_v)$. Both supports have diameter at most $R$, so
\begin{align}
\operatorname{supp}(P_u) \cup \operatorname{supp}(P_v) \subseteq B(s, R).
\end{align}
Because $c \in \operatorname{supp}(P_w)$ and $\operatorname{supp}(P_w) \subseteq \operatorname{supp}(P_u) \cup \operatorname{supp}(P_v)$, we have $c \in B(s, R)$. Thus $\operatorname{dist}_G(c, s) \leq R$. Since $s \in \operatorname{supp}(P_u)$, it follows that $\operatorname{supp}(P_u) \cap B(c, R) \neq \varnothing$, so $u \in \mathcal{U}_c$. The same argument gives $v \in \mathcal{U}_c$. Hence every vector $r_w(h)$ entering $\Gamma_c(h)$ is supported on $\mathcal{U}_c$.

Finally, each such vector has zero inner product with $h$, because $B(h)h = 0$. Since the vector is supported on $\mathcal{U}_c$, its inner product with $h^{(\mathcal{U}_c)}$ is also zero. Therefore $\Gamma_c(h) h^{(\mathcal{U}_c)} = 0$ after restricting to $\mathcal{U}_c$.
\end{proof}

\subsubsection{Local cofactor polynomial}

We now introduce a local polynomial that detects whether the local commutator form has exactly the expected one-dimensional nullspace. Throughout this discussion, $\Gamma_c(h)$ is viewed as the $d_c \times d_c$ matrix obtained by restricting to the coordinates in $\mathcal{U}_c$. Define
\begin{align}
S_c(h^{(\mathcal{U}_c)}) := \sum_{j = 1}^{d_c} \det\!\big(\Gamma_c(h)^{(j)}\big),
\label{eq:geo-Sc-def}
\end{align}
where $\Gamma_c(h)^{(j)}$ is the $(d_c - 1) \times (d_c - 1)$ principal minor obtained by deleting the $j$-th row and column. Since $\Gamma_c(h) \succeq 0$, all these principal minors are nonnegative.

We always have $\Gamma_c(h) h^{(\mathcal{U}_c)} = 0$. If $h^{(\mathcal{U}_c)} \neq 0$, this immediately implies $\operatorname{rank}(\Gamma_c(h)) \leq d_c - 1$. If $h^{(\mathcal{U}_c)} = 0$, then every row entering $\Gamma_c(h)$ vanishes by row localization, so $\Gamma_c(h) = 0$, and the same rank bound holds. Hence
\begin{align}
\operatorname{rank}(\Gamma_c(h)) \leq d_c - 1
\end{align}
for all $h$. Therefore
\begin{align}
S_c(h^{(\mathcal{U}_c)}) > 0 \quad \quad \text{if and only if } \quad \operatorname{rank}(\Gamma_c(h)) = d_c - 1.
\label{eq:geo-Sc-rank}
\end{align}
In the rank-$(d_c - 1)$ case, the eigenvalues satisfy $\lambda_1(\Gamma_c(h)) = 0$, and
\begin{align}
S_c(h^{(\mathcal{U}_c)}) = \prod_{j = 2}^{d_c} \lambda_j(\Gamma_c(h)).
\label{eq:geo-Sc-product}
\end{align}
Indeed, $S_c$ is the elementary symmetric polynomial of degree $d_c - 1$ in the eigenvalues of $\Gamma_c(h)$, and only the product of the nonzero eigenvalues survives when the rank is $d_c - 1$.

\begin{assumption}[Structural local nondegeneracy]
\label{ass:geo-structural-nondeg}
For every $c \in \Lambda$, we have $S_c \not\equiv 0$.
\end{assumption}
\noindent This is a structural condition on the local Pauli support family. It says that the local commutator constraints are rich enough that, for at least one choice of local coefficients, the only local null direction is the Hamiltonian direction. Equivalently, for each local type there exists $h_0^{(\mathcal{U}_c)}$ such that $\operatorname{rank}(\Gamma_c(h_0)) = d_c - 1$. This condition can be checked by exact finite-dimensional linear algebra on local patches. It fails for commuting support families and may fail for overly constrained sparse families, but it holds for sufficiently rich finite-range Pauli families.

Let $\mu_c$ be a local coefficient law on $\mathbb R^{\mathcal U_c}$ with a density with respect to Lebesgue measure. Under Assumption~\ref{ass:geo-structural-nondeg}, we have
\begin{align}
\ker \Gamma_c(h) = \operatorname{span}\{h^{(\mathcal{U}_c)}\}
\quad \mu_c\textnormal{-almost surely}.
\label{eq:geo-local-kernel}
\end{align}
Indeed, since $S_c$ is a nonzero polynomial, its zero set has Lebesgue measure zero. Therefore, for $\mu_c$-almost every local coefficient vector, $S_c(h^{(\mathcal U_c)}) \neq 0$, and hence $\operatorname{rank}(\Gamma_c(h)) = d_c - 1$. Also, $h^{(\mathcal U_c)} \neq 0$ $\mu_c$-almost surely, and $\Gamma_c(h)h^{(\mathcal U_c)} = 0$. Thus the kernel is exactly the one-dimensional span of $h^{(\mathcal U_c)}$.

\subsubsection{Local moment lower bounds}

The local small-ball estimate below requires a quantitative lower bound on $\|S_c\|_{L^1}$. We give two ways to obtain such a bound. The first applies to coefficient distributions whose local densities are bounded below on a fixed full-dimensional box, such as independent uniform variables and independent Gaussians at a fixed scale. The second applies to Gaussian smoothing around an arbitrary center.

First we clear the fixed rational denominators in $\Gamma_c(h)$. If $w = u \oplus v$ with $u, v \in \mathcal{V}_{k,R}$ anticommuting, then $P_u$ and $P_v$ overlap on at least one site, and hence
\begin{align}
1 \leq |\operatorname{supp}(P_w)| \leq 2k - 1.
\end{align}
Let us define
\begin{align}
L_\star := \operatorname{lcm}\{1, 2, \ldots, 2k - 1\}
\end{align}
as well as
\begin{align}
\widetilde{\Gamma}_c(h) := L_\star \Gamma_c(h), \quad \widetilde{S}_c(h) := \sum_{j = 1}^{d_c} \det\big(\widetilde{\Gamma}_c(h)^{(j)}\big) = L_\star^{d_c - 1} S_c(h).
\end{align}
Then $\widetilde{S}_c$ has integer coefficients and is homogeneous of degree $2(d_c - 1)$. Equivalently,
\begin{equation}\label{eq:Sc-homogeneity}
S_c = L_\star^{-(d_c - 1)} \widetilde{S}_c    
\end{equation}
is homogeneous of degree $2(d_c - 1)$.

For fixed $(\Delta,k,R)$, only finitely many rooted local support types can occur. Hence, for any fixed $r > 0$, Assumption~\ref{ass:geo-structural-nondeg} gives a uniform constant $\gamma_{\mathrm{loc}}(r) > 0$ such that
\begin{align}
\int_{[-r,r]^{d_c}} \widetilde{S}_c(x)\,dx \geq \gamma_{\mathrm{loc}}(r) \quad \textnormal{for every } c.
\label{eq:finite-local-mass-bound}
\end{align}
Indeed, each local type determines a nonzero polynomial $\widetilde S_c$ with integer coefficients, and $\widetilde S_c \geq 0$. Therefore the integral is positive for each type. Since the number of possible types is finite, the minimum of these positive integrals is positive. Assumption~\ref{ass:geo-structural-nondeg} makes each corresponding $\widetilde{S}_c$ a nonzero polynomial. Since $\widetilde{S}_c \geq 0$, each integral above is positive, and the minimum over the finite list is positive.

\begin{assumption}[Uniform local density lower bound]
\label{ass:geo-flat-core}
For every $c$, let $\mu_c$ denote the marginal law of the local coefficient vector $h^{(\mathcal U_c)} \in \mathbb R^{d_c}$, and assume that $\mu_c$ has a log-concave density $f_c$ with respect to Lebesgue measure. There exist constants $\kappa_0,r_0 > 0$, independent of $n$ and $c$, and a scale $s_n \in (0,1]$ such that
\begin{align}
f_c(x) \geq \kappa_0 s_n^{-d_c}
\quad
\textnormal{for all } x \in [-r_0 s_n,r_0 s_n]^{d_c}.
\label{eq:geo-flat-core}
\end{align}
\end{assumption}

This assumption says that every local coefficient law has a full-dimensional core, of side length proportional to $s_n$, on which its density is bounded below at the natural scale $s_n^{-d_c}$. The log-concavity assumption is used later in the Carbery--Wright small-ball estimate, while the lower bound \eqref{eq:geo-flat-core} is used in Lemma~\ref{lem:geo-flat-core-L1} to obtain a uniform lower bound on $\|S_c\|_{L^1(\mu_c)}$.

\begin{lemma}[Local moment bound from a local density lower bound]
\label{lem:geo-flat-core-L1}
Assume Assumptions~\ref{ass:geo-structural-nondeg} and~\ref{ass:geo-flat-core}. Then there exists $\eta_{\mathrm{dens}} > 0$, depending only on the finite list of local support types and on $\kappa_0, r_0$, such that
\begin{align}
\|S_c\|_{L^1(\mu_c)} \geq \eta_{\mathrm{dens}} s_n^{D_\star}, \quad D_\star := 2(d_{\mathrm{loc}} - 1),
\label{eq:density-lower-bound-L1}
\end{align}
for every $c$.
\end{lemma}

\begin{proof}
Since $S_c \geq 0$, the $L^1$-norm is just its expectation. By the local density lower bound in Assumption~\ref{ass:geo-flat-core} and the degree-$2(d_c - 1)$ homogeneity of $S_c$ (\eqref{eq:Sc-homogeneity}),
\begin{align}
\|S_c\|_{L^1(\mu_c)} &\geq \kappa_0 s_n^{-d_c} \int_{[-r_0 s_n, r_0 s_n]^{d_c}} S_c(x)\,dx \\
&= \kappa_0 s_n^{2(d_c - 1)} \int_{[-r_0, r_0]^{d_c}} S_c(z)\,dz.
\end{align}
Since $S_c = L_\star^{-(d_c - 1)} \widetilde{S}_c$, and using \eqref{eq:finite-local-mass-bound} with $r = r_0$ gives a positive lower bound for the last integral uniformly in $c$, we obtain
\begin{align}
\|S_c\|_{L^1(\mu_c)} \geq \eta_{\mathrm{dens}} s_n^{2(d_c - 1)}.
\end{align}
Because $s_n \in (0,1]$ and $2(d_c - 1) \leq D_\star$, this implies \eqref{eq:density-lower-bound-L1}.
\end{proof}

The following variant is the form used for smoothed analysis.

\begin{lemma}[Local moment bound under Gaussian smoothing]
\label{lem:geo-smoothed-L1}
Assume Assumption~\ref{ass:geo-structural-nondeg}. Let $x_0^{(c)} \in \mathbb{R}^{d_c}$ be arbitrary, and let
\begin{align}
X_c = x_0^{(c)} + \sigma Z_c, \quad Z_c \sim \mathcal{N}(0, I_{d_c}), \quad 0 < \sigma \leq 1.
\end{align}
Then there exists $\eta_{\mathrm{sm}} > 0$, depending only on the finite list of local support types, such that
\begin{align}
\mathbb{E}\,S_c(X_c) \geq \eta_{\mathrm{sm}} \sigma^{D_\star}
\label{eq:smoothed-L1-bound}
\end{align}
for every $c$ and every center $x_0^{(c)}$.
\end{lemma}

\begin{proof}
It suffices to prove the claim for $\widetilde{S}_c$, and then divide by $L_\star^{d_c - 1}$. Let $D_c = 2(d_c - 1)$. For fixed $x_0$ and $\sigma$, define
\begin{align}
P_{x_0,\sigma}(z) := \sigma^{-D_c} \widetilde{S}_c(x_0 + \sigma z).
\end{align}
The degree-$D_c$ homogeneous part of $P_{x_0,\sigma}$ is $\widetilde{S}_c(z)$; the remaining lower-degree terms may depend on $x_0$ and $\sigma$. Therefore $P_{x_0,\sigma}$ lies in the affine space
\begin{align}
\widetilde{S}_c + \mathcal{P}_{<D_c},
\end{align}
where $\mathcal{P}_{<D_c}$ is the finite-dimensional space of polynomials of degree at most $D_c - 1$.

Since $\widetilde{S}_c \notin \mathcal{P}_{<D_c}$, its distance from $\mathcal{P}_{<D_c}$ in the $L^1([-1,1]^{d_c})$ norm is positive. Taking the minimum over the finite list of possible local polynomials gives a constant $\theta > 0$ such that
\begin{align}
\int_{[-1,1]^{d_c}} |P_{x_0,\sigma}(z)|\,dz \geq \theta
\end{align}
for every $c, x_0, \sigma$.

Since $\widetilde{S}_c \geq 0$, the polynomial $P_{x_0,\sigma}$ is nonnegative. The standard Gaussian density is bounded below by a positive constant on $[-1,1]^{d_c}$, uniformly because $d_c \leq d_{\mathrm{loc}}$. Hence
\begin{align}
\mathbb{E}\,\widetilde{S}_c(x_0 + \sigma Z_c) = \sigma^{D_c}\mathbb{E}\,P_{x_0,\sigma}(Z_c) \geq \theta' \sigma^{D_c}
\end{align}
for some $\theta' > 0$. Since $D_c \leq D_\star$ and $\sigma \in (0,1]$, this is at least $\theta' \sigma^{D_\star}$. Dividing by $L_\star^{d_c - 1}$ and taking the minimum over $d_c \leq d_{\mathrm{loc}}$ proves \eqref{eq:smoothed-L1-bound}.
\end{proof}

\subsubsection{Local small-ball bound}

We now turn the lower moment bound for $S_c$ into a lower-tail bound for the local spectral gap. The point is simple: $S_c$ is the product of the nonzero local eigenvalues. Thus, if the second eigenvalue of $\Gamma_c(h)$ is small and the operator norm of $\Gamma_c(h)$ is bounded, then $S_c$ must also be small. Carbery--Wright then controls the probability of this small-value event.

We use the Carbery--Wright inequality in the following form \cite[Theorem 8]{CW01}. If $\mu$ is log-concave and $Q$ is a nonzero polynomial of degree at most $D$, then
\begin{align}
\mu(\{ |Q| \leq \varepsilon \}) \leq C_{\mathrm{CW}} D \left(\frac{\varepsilon}{\|Q\|_{L^1(\mu)}}\right)^{1/D}.
\label{eq:carbery-wright}
\end{align}

For $H_{\max} > 0$, define the coefficient upper-bound event
\begin{align}
\mathcal{E}_{\max}(H_{\max}) := \left\{ \max_{u \in \mathcal{V}_{k,R}} |h_u| \leq H_{\max} \right\}.
\end{align}
Also define
\begin{align}
M_\star := \max\{1, d_{\mathrm{loc}} w_{\mathrm{loc}} H_{\max}^2\}.
\end{align}

\begin{lemma}[Local small-ball estimate]
\label{lem:geo-local-smallball}
Assume Assumption~\ref{ass:geo-structural-nondeg}. Suppose that the local coefficient laws have log-concave densities and satisfy the uniform moment lower bound
\begin{align}
\|S_c\|_{L^1(\mu_c)} \geq \eta_0 \quad \textnormal{for all } c.
\label{eq:local-L1-eta0}
\end{align}
Then, for every $\xi > 0$,
\begin{align}
\mathbb{P}\!\left[\mathcal{E}_{\max}(H_{\max}) \cap \{ \exists c \in \Lambda : \lambda_2(\Gamma_c(h)) \leq \xi \}\right] \leq n C_{\mathrm{CW}} D_\star \min\,\left\{1, \left(\frac{\xi M_\star^{d_{\mathrm{loc}} - 2}}{\eta_0}\right)^{1/D_\star}\right\}.
\label{eq:geo-local-smallball}
\end{align}
Here $\lambda_2(\Gamma_c(h))$ denotes the second eigenvalue of the restricted $d_c \times d_c$ matrix on $\mathcal{U}_c$, and $D_\star = 2(d_{\mathrm{loc}} - 1)$.
\end{lemma}

\begin{proof}
First we prove a deterministic norm bound on $\mathcal{E}_{\max}(H_{\max})$. Fix a site $c \in \Lambda$. Let $B_c(h)$ be the submatrix of $B(h)$ with rows indexed by $\mathcal{W}_c$ and columns indexed by $\mathcal{U}_c$. Then
\begin{align}
\Gamma_c(h) = B_c(h)^\top D_c^{\mathcal W} B_c(h),
\end{align}
where $D_c^{\mathcal W}$ is the diagonal matrix indexed by $w \in \mathcal W_c$ with entries
\begin{equation}
(D_c^{\mathcal W})_{w,w} := \frac{1}{|\operatorname{supp}(P_w)|}.
\end{equation}
Since every $w \in \mathcal W_c$ is a nonidentity Pauli string, $|\operatorname{supp}(P_w)| \geq 1$, and hence $0 \preceq D_c^{\mathcal W} \preceq I$. Therefore
\begin{align}
\|\Gamma_c(h)\|_{\mathrm{op}} \leq \|B_c(h)\|_{\mathrm{op}}^2 \leq \|B_c(h)\|_F^2.
\end{align}
For fixed $w$ and $u$, the equation $u \oplus v = w$ determines $v$ uniquely. Thus, from \eqref{eq:def-B-matrix}, each entry of $B_c(h)$ is either $0$ or $\pm h_v$. On $\mathcal{E}_{\max}(H_{\max})$, each entry therefore has magnitude at most $H_{\max}$. Since $B_c(h)$ has at most $d_{\mathrm{loc}} w_{\mathrm{loc}}$ entries, we obtain
\begin{align}
\|\Gamma_c(h)\|_{\mathrm{op}} \leq d_{\mathrm{loc}} w_{\mathrm{loc}} H_{\max}^2 \leq M_\star.
\label{eq:local-opnorm-bound}
\end{align}

Since $\Gamma_c(h) h^{(\mathcal{U}_c)} = 0$, the restricted matrix $\Gamma_c(h)$ has at least one zero eigenvalue. Write its eigenvalues as
\begin{align}
0 = \lambda_1(\Gamma_c(h)) \leq \lambda_2(\Gamma_c(h)) \leq \cdots \leq \lambda_{d_c}(\Gamma_c(h)).
\end{align}
As $S_c$ is the sum of the $(d_c - 1)$-principal minors, it is the elementary symmetric polynomial of degree $d_c - 1$ in these eigenvalues. Since $\lambda_1 = 0$, this gives
\begin{align}
S_c(h^{(\mathcal{U}_c)}) = \prod_{j = 2}^{d_c} \lambda_j(\Gamma_c(h)).
\label{eq:Sc-product-all-cases}
\end{align}
This identity also covers the rank-deficient case, in which both sides vanish.

Therefore, on $\mathcal{E}_{\max}(H_{\max})$, the condition $\lambda_2(\Gamma_c(h)) \leq \xi$ implies
\begin{align}
S_c(h^{(\mathcal{U}_c)}) \leq \xi M_\star^{d_c - 2} \leq \xi M_\star^{d_{\mathrm{loc}} - 2},
\end{align}
where we used $M_\star \geq 1$ and $d_c \leq d_{\mathrm{loc}}$. Thus
\begin{align}
\mathcal{E}_{\max}(H_{\max}) \cap \{ \lambda_2(\Gamma_c(h)) \leq \xi \} \subseteq \{ S_c \leq \xi M_\star^{d_{\mathrm{loc}} - 2} \}.
\label{eq:lambda-small-implies-Sc-small}
\end{align}

By Assumption~\ref{ass:geo-structural-nondeg}, $S_c$ is a nonzero polynomial, and its degree is $2(d_c - 1) \leq D_\star$. Applying Carbery--Wright to $S_c$, using \eqref{eq:local-L1-eta0}, gives
\begin{align}
\mathbb{P}\left(S_c \leq \varepsilon\right) \leq C_{\mathrm{CW}} D_\star \min\left\{1, \left(\frac{\varepsilon}{\eta_0}\right)^{1/D_\star}\right\}.
\end{align}
Here the minimum with $1$ makes the bound trivial when $\varepsilon \geq \eta_0$, and when $\varepsilon < \eta_0$, replacing the local degree $2(d_c - 1)$ by $D_\star$ only weakens the estimate.

Finally, taking $\varepsilon = \xi M_\star^{d_{\mathrm{loc}} - 2}$ and applying a union bound over $c \in \Lambda$ proves \eqref{eq:geo-local-smallball}.
\end{proof}

\subsubsection{Patching local gaps to a global gap}

We now show how local spectral gaps imply a global commutator gap. The intuition is as follows. If $x$ has small local error $y_c$ on every neighborhood $\mathcal{U}_c$\,, then locally $x$ is close to a scalar multiple $\beta_c h$. Neighboring neighborhoods overlap on anchor blocks, so the scalars $\beta_c$ must vary slowly across the interaction graph. The graph Poincar\'e inequality then forces the $\beta_c$'s to be nearly constant. Since $x \perp h$, this is incompatible with $\|x\|_2 = 1$, unless the local errors $y_c$ have nontrivial total mass.

Let $L_G$ be the unnormalized graph Laplacian of $G$. Since $G$ is connected, for every $\beta \in \mathbb{R}^{\Lambda}$,
\begin{align}
\sum_{q \in \Lambda}(\beta_q - \bar{\beta})^2 \leq \frac{1}{\lambda_2(L_G)} \sum_{\{p, q\} \in \textsf{E}}(\beta_p - \beta_q)^2, \quad \bar{\beta} := \frac{1}{n}\sum_{q \in \Lambda}\beta_q.
\label{eq:geo-poincare}
\end{align}
Moreover, for every connected graph on $n \geq 2$ vertices, $\lambda_2(L_G) \geq \frac{4}{n^2}$.  Then we have the following.

\begin{lemma}[Patching inequality]
\label{lem:geo-patching}
Let $x \in \mathbb{R}^{|\mathcal{V}_{k,R}|}$ satisfy $x \perp h$ and $\|x\|_2 = 1$. Assume $h_{\min} > 0$. For each $c \in \Lambda$, decompose
\begin{align}
x^{(\mathcal{U}_c)} = \beta_c h^{(\mathcal{U}_c)} + y_c\,, \quad y_c \perp h^{(\mathcal{U}_c)}.
\label{eq:geo-local-decomp}
\end{align}
Then
\begin{align}
\sum_{c \in \Lambda}\|y_c\|_2^2 \geq \frac{\lambda_2(L_G) h_{\min}^2}{2 \lambda_2(L_G) h_{\min}^2 + 4 \Delta d_{\mathrm{blk}} h_{\max}^2}\,,
\label{eq:geo-patching-strong}
\end{align}
and in particular,
\begin{align}
\sum_{c \in \Lambda}\|y_c\|_2^2 \geq \frac{h_{\min}^2}{2 h_{\min}^2 + n^2 \Delta d_{\mathrm{blk}} h_{\max}^2}.
\label{eq:geo-patching-simple}
\end{align}
\end{lemma}

\begin{proof}
The decomposition in \eqref{eq:geo-local-decomp} is well-defined because $h_{\min} > 0$, so $h^{(\mathcal{U}_c)} \neq 0$. Let $\bar{\beta} \in \mathbb{R}$. Since $x \perp h$,
\begin{align}
\|x - \bar{\beta}h\|_2^2 = \|x\|_2^2 + \bar{\beta}^2\|h\|_2^2 \geq 1.
\end{align}
Using the anchor partition,
\begin{align}
\|x - \bar{\beta}h\|_2^2 = \sum_{q \in \Lambda}\|x^{(\mathcal{V}_q)} - \bar{\beta}h^{(\mathcal{V}_q)}\|_2^2.
\end{align}
Restricting \eqref{eq:geo-local-decomp} with $c = q$ to the block $\mathcal{V}_q \subseteq \mathcal{U}_q$, we have
\begin{align}
x^{(\mathcal{V}_q)} - \bar{\beta}h^{(\mathcal{V}_q)} = (\beta_q - \bar{\beta})h^{(\mathcal{V}_q)} + y_q^{(\mathcal{V}_q)}.
\end{align}
Using $\|a + b\|_2^2 \leq 2\|a\|_2^2 + 2\|b\|_2^2$, $\|y_q^{(\mathcal{V}_q)}\|_2 \leq \|y_q\|_2$, and $\|h^{(\mathcal{V}_q)}\|_2^2 \leq d_{\mathrm{blk}} h_{\max}^2$, we get
\begin{align}
1 \leq 2 d_{\mathrm{blk}} h_{\max}^2 \sum_{q \in \Lambda}(\beta_q - \bar{\beta})^2 + 2 \sum_{q \in \Lambda}\|y_q\|_2^2.
\label{eq:geo-patch-step-one}
\end{align}
Choose $\bar{\beta} = \tfrac{1}{n}\sum_q \beta_q$, and apply \eqref{eq:geo-poincare}.

We next control the edge differences $\beta_p - \beta_q$ by the local errors. Fix an edge $\{p, q\} \in \textsf{E}$. Since $R \geq 1$, every Pauli string anchored at $p$ lies in $\mathcal{U}_q$, so $\mathcal{V}_p \subseteq \mathcal{U}_q$. Restricting the decompositions at $p$ and $q$ to $\mathcal{V}_p$ gives
\begin{align}
(\beta_p - \beta_q)h^{(\mathcal{V}_p)} = y_q^{(\mathcal{V}_p)} - y_p^{(\mathcal{V}_p)}.
\end{align}
Because $\mathcal{V}_p \neq \varnothing$ and $h_{\min} > 0$,
\begin{align}
\|h^{(\mathcal{V}_p)}\|_2^2 \geq h_{\min}^2.
\end{align}
Therefore
\begin{align}
(\beta_p - \beta_q)^2 \leq \frac{2}{h_{\min}^2}\left(\|y_p\|_2^2 + \|y_q\|_2^2\right).
\end{align}
Summing over edges and using that every site has degree at most $\Delta$,
\begin{align}
\sum_{\{p, q\} \in \textsf{E}}(\beta_p - \beta_q)^2 \leq \frac{2\Delta}{h_{\min}^2}\sum_{q \in \Lambda}\|y_q\|_2^2.
\label{eq:geo-edge-difference-control}
\end{align}
Combining \eqref{eq:geo-patch-step-one}, \eqref{eq:geo-poincare}, and \eqref{eq:geo-edge-difference-control}, we obtain
\begin{align}
1 \leq \left(2 + \frac{4\Delta d_{\mathrm{blk}} h_{\max}^2}{\lambda_2(L_G) h_{\min}^2}\right)\sum_{q \in \Lambda}\|y_q\|_2^2.
\end{align}
Rearranging proves \eqref{eq:geo-patching-strong}. 

It remains to derive the simplified form. Dividing the right-hand side of \eqref{eq:geo-patching-strong} by $\lambda_2(L_G)$ in numerator and denominator gives
\begin{align}
\frac{\lambda_2(L_G) h_{\min}^2}{2 \lambda_2(L_G) h_{\min}^2 + 4 \Delta d_{\mathrm{blk}} h_{\max}^2} = \frac{h_{\min}^2}{2 h_{\min}^2 + \frac{4 \Delta d_{\mathrm{blk}} h_{\max}^2}{\lambda_2(L_G)}}.
\end{align}
Using $\lambda_2(L_G) \geq 4/n^2$, we have
\begin{align}
\frac{4 \Delta d_{\mathrm{blk}} h_{\max}^2}{\lambda_2(L_G)} \leq n^2 \Delta d_{\mathrm{blk}} h_{\max}^2.
\end{align}
This proves \eqref{eq:geo-patching-simple}.
\end{proof}

\begin{proposition}[Deterministic geometric gap from local gaps]
\label{prop:deterministic-geometric-gap-from-local-gaps}
Assume the geometric setup above, including the existence of a support-respecting anchoring map with $\mathcal V_q \neq \varnothing$ for every $q \in \Lambda$, the local dimension condition \eqref{eq:geo-dc-at-least-two}, and $h_{\min} > 0$. Suppose that for some $\xi > 0$,
\begin{align}
\lambda_2(\Gamma_c(h)) \geq \xi \quad \textnormal{for every } c \in \Lambda,
\end{align}
where $\Gamma_c(h)$ is viewed as a $d_c \times d_c$ matrix on $\mathcal{U}_c$. Then
\begin{align}
\min_{\substack{A \in \mathcal{O}_{k,R}^{\mathrm{op}}\\ \tfrac{1}{d}\tr(A^2) = 1,\ \tr(AH) = 0}} \frac{1}{d}\|[A,H]\|_F^2 \geq 4\xi \,\frac{h_{\min}^2}{2 h_{\min}^2 + n^2 \Delta d_{\mathrm{blk}} h_{\max}^2}.
\label{eq:deterministic-geo-gap-local}
\end{align}
\end{proposition}

\begin{proof}
Let $x \perp h$ and $\|x\|_2 = 1$. For each $c$, decompose
\begin{align}
x^{(\mathcal{U}_c)} = \beta_c h^{(\mathcal{U}_c)} + y_c, \quad y_c \perp h^{(\mathcal{U}_c)}.
\end{align}
Since $h^{(\mathcal{U}_c)} \in \ker \Gamma_c(h)$ and $\lambda_2(\Gamma_c(h)) \geq \xi$,
\begin{align}
(x^{(\mathcal{U}_c)})^\top\Gamma_c(h)x^{(\mathcal{U}_c)} \geq \xi\|y_c\|_2^2.
\end{align}
Summing over $c$ and using $\Gamma(h) = \sum_c \Gamma_c(h)$,
\begin{align}
x^\top\Gamma(h)x \geq \xi\sum_{c \in \Lambda}\|y_c\|_2^2.
\end{align}
The patching lemma gives
\begin{align}
x^\top\Gamma(h)x \geq \xi \frac{h_{\min}^2}{2 h_{\min}^2 + n^2 \Delta d_{\mathrm{blk}} h_{\max}^2}.
\end{align}
Taking the infimum over $x \perp h$, $\|x\|_2 = 1$, and using \eqref{eq:geo-global-commutator-identity}, proves the claim.
\end{proof}

\begin{theorem}[Probabilistic geometric commutator gap]
\label{thm:geo-global-gap}
Assume the geometric setup above, including:
\begin{enumerate}
\item $G$ is connected, $\deg(G) \leq \Delta = O(1)$, and $k, R = O(1)$ with $R \geq 1$;
\item there exists a support-respecting anchoring map $\alpha: \mathcal V_{k,R} \to \Lambda$ such that $\mathcal V_q \neq \varnothing$ for every $q$, and $d_c \geq 2$ for every $c$;
\item Assumption~\ref{ass:geo-structural-nondeg} holds;
\item the local coefficient laws have log-concave densities and satisfy $\|S_c\|_{L^1(\mu_c)} \geq \eta_0$ for every $c$;
\item $h_{\min} > 0$ almost surely.
\end{enumerate}
Fix $H_{\max} > 0$, and let $M_\star = \max\{1, d_{\mathrm{loc}} w_{\mathrm{loc}} H_{\max}^2\}$. Then, for every $\xi > 0$, with probability at least
\begin{align}
1 - \mathbb{P}(\mathcal{E}_{\max}(H_{\max})^c) - n C_{\mathrm{CW}} D_\star \min\left\{1, \left(\frac{\xi M_\star^{d_{\mathrm{loc}} - 2}}{\eta_0}\right)^{1/D_\star}\right\},
\label{eq:geo-global-gap-probability}
\end{align}
we have
\begin{align}
\min_{\substack{A \in \mathcal{O}_{k,R}^{\mathrm{op}}\\ \tfrac{1}{d}\tr(A^2) = 1,\ \tr(AH) = 0}} \frac{1}{d}\|[A,H]\|_F^2 \geq 4\xi \frac{h_{\min}^2}{2 h_{\min}^2 + n^2 \Delta d_{\mathrm{blk}} h_{\max}^2}.
\label{eq:geo-global-gap}
\end{align}
\end{theorem}

\begin{proof}
By Lemma~\ref{lem:geo-local-smallball}, outside an event of probability at most
\begin{align}
\mathbb{P}\left[\mathcal{E}_{\max}(H_{\max})^c\right] + n C_{\mathrm{CW}} D_\star \min\!\left\{1, \left(\frac{\xi M_\star^{d_{\mathrm{loc}} - 2}}{\eta_0}\right)^{1/D_\star}\right\},
\end{align}
the event $\mathcal{E}_{\max}(H_{\max})$ holds and $\lambda_2(\Gamma_c(h)) \geq \xi$ for every $c \in \Lambda$. The claim then follows from Proposition~\ref{prop:deterministic-geometric-gap-from-local-gaps}.
\end{proof}

The hypotheses of Theorem~\ref{thm:geo-global-gap} have separate structural and probabilistic roles. Hypothesis~(1) is the geometric locality input: bounded degree and finite range imply that all local neighborhoods have size bounded only in terms of $(\Delta,k,R)$. Hypothesis~(2) is used in the patching argument. The condition $\mathcal V_q \neq \varnothing$ ensures that each site has at least one anchored Pauli coordinate, so neighboring scalar factors $\beta_p$ and $\beta_q$ can be compared on a shared local block. The condition $d_c \geq 2$ only excludes degenerate local neighborhoods for which a second local eigenvalue is not defined.

The main structural condition is hypothesis~(3), namely Assumption~\ref{ass:geo-structural-nondeg} that $S_c \not\equiv 0$ for every $c$. This is a finite-dimensional algebraic condition on the local Pauli support family. It says that the local commutator constraints can isolate the Hamiltonian direction on each rooted local type. For fixed $(\Delta,k,R)$, it can be checked on a finite list of local neighborhoods by exact arithmetic. It fails for commuting families and may fail for overly constrained sparse families.

Hypothesis~(4) supplies the quantitative anti-concentration input. Log-concavity is used for the Carbery--Wright inequality, while the lower bound
\begin{equation}
\|S_c\|_{L^1(\mu_c)} \geq \eta_0
\end{equation}
sets the scale in the local small-ball estimate. In applications, this moment lower bound follows either from the local density lower-bound condition in Lemma~\ref{lem:geo-flat-core-L1} or from Gaussian smoothing in Lemma~\ref{lem:geo-smoothed-L1}. Finally, $h_{\min} > 0$ is automatic for the continuous coefficient distributions considered here, but is kept explicit because the deterministic patching bound depends on it.

\begin{remark}[From the theorem to an inverse-polynomial gap]
The lower bound in \eqref{eq:geo-global-gap} is stated in terms of the random quantities $h_{\min}$ and $h_{\max}$. To obtain a fully explicit high-probability inverse-polynomial bound, one combines Theorem~\ref{thm:geo-global-gap} with coefficient estimates of the form
\begin{align}
h_{\min} \geq n^{-C}, \quad h_{\max} \leq n^C, \quad H_{\max} \leq n^C, \quad \eta_0 \geq n^{-C}\,,
\end{align}
with high probability. Then choosing $\xi = n^{-C'}$ for $C'$ sufficiently large gives
\begin{align}
\min_{\substack{A \in \mathcal{O}_{k,R}^{\mathrm{op}}\\ \tfrac{1}{d}\tr(A^2) = 1,\ \tr(AH) = 0}} \frac{1}{d}\|[A,H]\|_F^2 \geq \frac{1}{\operatorname{poly}(n)}
\end{align}
with probability at least $1 - 1/\operatorname{poly}(n)$.
\end{remark}

\subsubsection{Deterministic certificates for the commutator gap }
\label{subsubsec:deterministic-certificates}

We briefly explain how the deterministic quantities appearing above can be checked in concrete instances. This subsection does not verify the probabilistic hypotheses in Theorem~\ref{thm:geo-global-gap}; those distributional assumptions, such as log-concavity, local density lower bounds, coefficient lower tails, and coefficient upper tails, are addressed separately in the ensemble-specific results below.

For a fixed Hamiltonian, the strongest direct diagnostic is the projected commutator gap. Construct the collision-aware commutator matrix $B(h)$, whose rows are indexed by Pauli strings appearing in commutators and whose columns are indexed by the chosen local support family $\mathcal V_{k,R}$. Recall that $\Gamma(h) := B(h)^\top B(h)$. The projected static commutator gap is
\begin{equation}
\pi_H := 4 \min_{\substack{a \perp h\\ \|a\|_2 = 1}} a^\top \Gamma(h) a.
\end{equation}
This is exactly the optimal static commutator gap in the chosen local span. Numerically, this quantity can be estimated with sparse eigensolvers. A rigorous fixed-Hamiltonian certificate requires an exact computation or a certified eigenvalue lower bound, for example by interval arithmetic or another validated spectral method.

The geometric theorem also gives a local sufficient certificate. For each $c$, compute the local matrix $\Gamma_c(h)$ restricted to the coordinate set
\begin{equation}
\mathcal U_c = \{u \in \mathcal V_{k,R} : \operatorname{supp}(P_u) \cap B(c,R) \neq \varnothing\}.
\end{equation}
If $\lambda_2(\Gamma_c(h)) \geq \xi$ for every $c$, then Proposition~\ref{prop:deterministic-geometric-gap-from-local-gaps} gives the global bound \eqref{eq:deterministic-geo-gap-local}.

The structural condition $S_c \not\equiv 0$ is a property of the support family, not of a particular Hamiltonian. It is certified by finding, for each rooted local type $c$, an integer witness vector $h_0^{(\mathcal U_c)} \neq 0$ such that the local commutator constraints have rank $d_c - 1$. Equivalently, let $B_c(h_0)$ be the local row matrix whose rows are indexed by $\mathcal W_c$ and whose columns are indexed by $\mathcal U_c$. Since $\Gamma_c(h_0) = B_c(h_0)^\top D_c^{\mathcal W} B_c(h_0)$ with $D_c^{\mathcal W}$ positive diagonal, $\operatorname{rank}\Gamma_c(h_0) = \operatorname{rank}B_c(h_0)$. Thus it suffices to verify
\begin{equation}
\operatorname{rank}_{\mathbb Q} B_c(h_0) = d_c - 1.
\end{equation}

In the implementation, this rank statement is certified by modular exact arithmetic. The matrix $B_c(h_0)$ has integer entries. The supplied certificates use the prime $p = 2147483647$. If
\begin{equation}
\operatorname{rank}_{\mathbb F_p} B_c(h_0) = d_c - 1,
\end{equation}
then $\operatorname{rank}_{\mathbb Q} B_c(h_0) \geq d_c - 1$. On the other hand, the exact integer identity
\begin{equation}
B_c(h_0)h_0^{(\mathcal U_c)} = 0
\end{equation}
and the condition $h_0^{(\mathcal U_c)} \neq 0$ imply $\operatorname{rank}_{\mathbb Q} B_c(h_0) \leq d_c - 1$. Hence $\operatorname{rank}_{\mathbb Q} B_c(h_0) = d_c - 1$, and therefore $S_c \not\equiv 0$. Floating-point computations may be used to search for candidate witnesses, but the certificate itself is exact: it consists of verifying $B_c(h_0)h_0^{(\mathcal U_c)} = 0$ over the integers and verifying $\operatorname{rank}_{\mathbb F_p} B_c(h_0) = d_c - 1$ over the prime field used in the certificate.

For reproducibility, each certificate file records:
\begin{enumerate}
\item the support family, lattice, locality parameters $k,R$, and boundary or rooted local type;
\item the precise definition of $\mathcal U_c$ used in the check;
\item the ordered local coordinate list $\mathcal U_c$;
\item the ordered row list $\mathcal W_c$, or an equivalent hash of the row construction;
\item the integer witness vector $h_0^{(\mathcal U_c)}$;
\item the prime $p$ used for the modular rank computation;
\item the exact check that $B_c(h_0) h_0^{(\mathcal U_c)} = 0$ over the integers;
\item the modular rank value $\operatorname{rank}_{\mathbb F_p}B_c(h_0)$ and the target value $d_c - 1$;
\item pivot-column data, or an equivalent machine-checkable modular-rank transcript.
\end{enumerate}

Periodic lattice certificates are performed on rooted balls of radius $2R$ in a periodic lattice whose side length is large enough that this ball does not wrap around; equivalently, one may work directly in the corresponding infinite lattice. Such a certificate applies to all sufficiently large periodic systems with the same local type. Open-boundary certificates require enumerating all rooted radius-$2R$ local environments that can occur near a boundary, edge, corner, or bulk region. A sparse physical subfamily is not certified by embedding it into a larger dense family; the certificate must be run for the actual support family and the actual local coordinate sets $\mathcal U_c$.

The current exact-arithmetic certification run certifies the following support families. The certified cases include the periodic XYZ-chain support family; all dense finite-range cases on the $1$D chain, square, triangular, honeycomb, and cubic lattices for $k \in \{2,3,4\}$ and $R \in \{1,2\}$, with open or periodic boundary conditions; and the exact two-body no-field Pauli families for $R \in \{1,2\}$ on the chain, square, triangular, honeycomb, and cubic lattices, with open or periodic boundary conditions.

\subsection{Explicit ensembles and smoothed analysis}
\label{subsec:explicit-ensembles-and-smoothed-analysis}

We now specialize the commutator-gap bounds to several standard coefficient distributions. Throughout this subsection, $N := |\mathcal{V}|$, unless a specific support family is being discussed.

\subsubsection{Elementary coefficient bounds}

We will repeatedly use the following elementary estimates.

\begin{lemma}[Uniform coefficient bounds]
\label{lem:uniform-coefficient-bounds}
Let $h_1,\ldots,h_N$ be independent random variables with
\begin{align}
h_j \stackrel{\,\textnormal{i.i.d.}\,}{\sim} \operatorname{Unif}[-\sigma,\sigma], \quad 0 < \sigma \leq 1.
\end{align}
Then for every $\delta \in (0,1)$, with probability at least $1-\delta$,
\begin{align}
\min_{1 \leq j \leq N}|h_j| \geq \frac{\sigma\delta}{N}, \quad \max_{1 \leq j \leq N}|h_j| \leq \sigma.
\end{align}
\end{lemma}

\begin{proof}
The upper bound is deterministic. For the lower bound, $\mathbb{P}\left(|h_j| \leq \frac{\sigma\delta}{N}\right) = \frac{\delta}{N}$. A union bound over $j = 1,\ldots,N$ proves the claim.
\end{proof}

\begin{lemma}[Gaussian coefficient bounds]
\label{lem:gaussian-coefficient-bounds}
Let $h_1,\ldots,h_N$ be independent random variables with
\begin{align}
h_j \stackrel{\,\textnormal{i.i.d.}\,}{\sim} \mathcal{N}(0,\sigma^2), \quad 0 < \sigma \leq 1.
\end{align}
Fix $\delta_{\min}, \delta_{\max} \in (0,1)$, and define
\begin{align}
L_{\max} := 2\log\left(\frac{2N}{\delta_{\max}}\right).
\end{align}
Then with probability at least $1-\delta_{\min}-\delta_{\max}$,
\begin{align}
\min_{1 \leq j \leq N}|h_j| \geq \frac{\sigma\delta_{\min}}{N}, \quad \max_{1 \leq j \leq N}|h_j| \leq \sigma\sqrt{L_{\max}}.
\end{align}
\end{lemma}

\begin{proof}
For a standard Gaussian $G$,
\begin{align}
\mathbb{P}(|G| \leq t) \leq \sqrt{\frac{2}{\pi}}\,t \leq t \quad \textnormal{for } t \geq 0.
\end{align}
Thus $\mathbb{P}\left(|h_j| \leq \frac{\sigma\delta_{\min}}{N}\right) \leq \frac{\delta_{\min}}{N}$. A union bound proves the lower bound with probability at least $1-\delta_{\min}$.

For the upper bound, using $\mathbb{P}(|G| > t) \leq 2 e^{-t^2/2}$, we get
\begin{align}
\mathbb{P}\left(|h_j| > \sigma\sqrt{L_{\max}}\right) \leq 2 e^{-L_{\max}/2} = \frac{\delta_{\max}}{N}.
\end{align}
A second union bound proves the upper bound with probability at least $1-\delta_{\max}$.
\end{proof}

For shifted Gaussian variables, we will use the following elementary bound. For every $m \in \mathbb{R}$, every $\varepsilon > 0$, and $G \sim \mathcal{N}(0,1)$,
\begin{align}
\mathbb{P}(|m+\sigma G| \leq \varepsilon) \leq \frac{\varepsilon}{\sigma}.
\label{eq:shifted-gaussian-small-interval}
\end{align}
Indeed, the density of $m+\sigma G$ is bounded by $1/\sqrt{2\pi \sigma^2}$, and the relevant interval has length $2\varepsilon$.

\subsubsection{Transverse-field Ising support}

Let $G=(\Lambda,\mathsf E)$ be a connected simple graph with $|\Lambda|=n\ge 2$, and consider
\begin{equation}
\mathcal V_{\mathrm{TFIM}}:=\{X_i:i\in\Lambda\}\cup\{Z_iZ_j:\{i,j\}\in E\}.
\end{equation}
For now, set
\begin{equation}
N:=|\mathcal V_{\mathrm{TFIM}}|=|\Lambda|+|E|.
\end{equation}
The anticommutation graph is the incidence graph of $G$, hence is connected. The family is
edge-product injective because
\begin{equation}
X_i(Z_iZ_j)\propto Y_iZ_j,\quad X_j(Z_iZ_j)\propto Z_iY_j,
\end{equation}
and these products determine the incidence pair.

Therefore Theorem \ref{thm:pauli-injective-gap} gives the following deterministic statement. If
\begin{equation}
H = \sum_{i\in\Lambda} g_i X_i+\sum_{\{i,j\}\in E} J_{ij} Z_i Z_j,
\end{equation}
all coefficients are nonzero, and
\begin{equation}
\frac1d\tr(H^2) = \sum_{i\in\Lambda}g_i^2+\sum_{\{i,j\}\in E}J_{ij}^2 = n,
\end{equation}
then
\begin{equation}
\min_{\substack{
A\in\operatorname{span}_{\mathbb R}\mathcal V_{\mathrm{TFIM}}\\
\tfrac{1}{d}\tr(A^2)=1,\ \tr(AH)=0}} \frac1d\|[A,H]\|_F^2 \ge \frac{16 h_{\min}^4}{h_{\max}^2N^2\left(1+Nh_{\max}^2/n\right)}\,,
\end{equation}
where $h_{\min}$ and $h_{\max}$ are the minimum and maximum absolute values of the
normalized TFIM coefficients.

The deterministic bound above assumes the normalization $\tfrac{1}{d}\tr(H^2) = n$. To apply it to iid uniform or Gaussian coefficients, we first sample an unnormalized TFIM Hamiltonian and then rescale it to this convention. Thus, let
\begin{equation}
\widetilde H = \sum_{i\in\Lambda}\widetilde g_i X_i+ \sum_{\{i,j\}\in E}\widetilde J_{ij} Z_i Z_j,
\end{equation}
write $\widetilde h\in\mathbb R^N$ for its coefficient vector, and set
\begin{equation}
S:=\|\widetilde h\|_2^2 = \sum_{i\in\Lambda}\widetilde g_i^2 + \sum_{\{i,j\}\in E}\widetilde J_{ij}^2.
\end{equation}
The Hamiltonian used in the theorem is the normalized Hamiltonian
\begin{equation}
H := \sqrt{\frac nS}\,\widetilde H,
\end{equation}
which satisfies $\tfrac{1}{d}\tr(H^2)=n$ whenever $S>0$. Let
\begin{equation}
\widetilde h_{\min}:=\min_{u\in\mathcal V_{\mathrm{TFIM}}}|\widetilde h_u|, \quad \widetilde h_{\max}:=\max_{u\in\mathcal V_{\mathrm{TFIM}}}|\widetilde h_u|.
\end{equation}
Then the normalized coefficients satisfy
\begin{equation}
h_{\min} = \sqrt n\,\frac{\widetilde h_{\min}}{\sqrt S}, \quad h_{\max} = \sqrt n\,\frac{\widetilde h_{\max}}{\sqrt S}.
\end{equation}
Substituting these identities into the deterministic TFIM bound gives
\begin{equation}
\min_{\substack{A\in\operatorname{span}_{\mathbb R}\mathcal V_{\mathrm{TFIM}}\\
\tfrac{1}{d}\tr(A^2)=1,\ \tr(AH)=0}} \frac1d\|[A,H]\|_F^2 \ge \frac{16n\,\widetilde h_{\min}^4}{N^2\widetilde h_{\max}^2(S+N\widetilde h_{\max}^2)}.
\end{equation}
Since $S\le N\widetilde h_{\max}^2$, we also have the simpler bound
\begin{equation}
\min_{\substack{
A\in\operatorname{span}_{\mathbb R}\mathcal V_{\mathrm{TFIM}}\\
\tfrac{1}{d}\tr(A^2)=1,\ \tr(AH)=0}} \frac1d\|[A,H]\|_F^2 \ge \frac{8n}{N^3}\left(\frac{\widetilde h_{\min}}{\widetilde h_{\max}}\right)^4.
\end{equation}

If the unnormalized TFIM coefficients are independent uniform random variables on $[-\sigma,\sigma]$,
with $0<\sigma\le 1$, then for every $\delta\in (0,1)$, with probability at least $1-\delta$,
\begin{equation}
\widetilde h_{\min}\ge \frac{\sigma\delta}{N}, \quad \widetilde h_{\max}\le \sigma.
\end{equation}
Therefore, with probability at least $1-\delta$,
\begin{equation}
\min_{\substack{A\in\operatorname{span}_{\mathbb R}\mathcal V_{\mathrm{TFIM}}\\
\tfrac{1}{d}\tr(A^2)=1,\ \tr(AH)=0}} \frac1d\|[A,H]\|_F^2 \ge \frac{8n\delta^4}{N^7}.
\end{equation}
Combining this with Theorem \ref{relation-commutator-A-U}, and using the same failure parameter $\delta$ for the coefficient event and the time-sampling event, gives the following. For $t\sim\operatorname{Unif}([0,T])$,
with probability at least $1-2\delta$,
\begin{equation}
\min_{\substack{
A\in\operatorname{span}_{\mathbb R}\mathcal V_{\mathrm{TFIM}}\\
\tfrac{1}{d}\tr(A^2)=1,\ \tr(AH)=0}} \frac1d\|[A,U(t)]\|_F^2 \ge \frac{\delta^6}{256N^{10}}\min\{nT^2,1\}.
\end{equation}

If the unnormalized TFIM coefficients are independent Gaussian random variables with law $\mathcal{N}(0,\sigma^2)$,
with $0<\sigma\le 1$, then for $\delta_{\min},\delta_{\max}\in (0,1)$, with probability at least
$1-\delta_{\min}-\delta_{\max}$,
\begin{equation}
\widetilde h_{\min}\ge \frac{\sigma\delta_{\min}}{N}, \quad \widetilde h_{\max}\le \sigma\sqrt{L_{\max}}, \quad L_{\max}:=2\log\left(\frac{2N}{\delta_{\max}}\right).
\end{equation}
Therefore, with probability at least $1-\delta_{\min}-\delta_{\max}$,
\begin{equation}
\min_{\substack{
A\in\operatorname{span}_{\mathbb R}\mathcal V_{\mathrm{TFIM}}\\
\tfrac{1}{d}\tr(A^2)=1,\ \tr(AH)=0}} \frac1d\|[A,H]\|_F^2 \ge \frac{8n\delta_{\min}^4}{N^7L_{\max}^2}.
\end{equation}
Combining this with Theorem \ref{relation-commutator-A-U}, and setting $\delta_{\min}=\delta_{\max}=\delta$ while also using $\delta$ as the time-sampling failure parameter, gives the following. For $t\sim\operatorname{Unif}([0,T])$, with probability at least
$1-3\delta$,
\begin{equation}
\min_{\substack{
A\in\operatorname{span}_{\mathbb R}\mathcal V_{\mathrm{TFIM}}\\
\tfrac{1}{d}\tr(A^2)=1,\ \tr(AH)=0}} \frac1d\|[A,U(t)]\|_F^2 \ge \frac{\delta^6}{256N^{10}L_{\max}^2}\min\{nT^2,1\}, \quad L_{\max} := 2\log\left(\frac{2N}{\delta}\right).
\end{equation}

\subsubsection{Locally nondegenerate geometric families}
\label{locally-nondegenerate-geo-family-subsec}

We call a geometrically local support family $\mathcal V_{k,R}$ locally nondegenerate if the following three conditions hold. First, there exists a support-respecting anchoring map $\alpha : \mathcal V_{k,R} \to \Lambda$ such that, writing $\operatorname{anc}(u) := \alpha(u)$, the anchor blocks $\mathcal V_q := \{ u \in \mathcal V_{k,R} : \operatorname{anc}(u) = q \}$ are nonempty for every $q \in \Lambda$. Second, the local coordinate neighborhoods satisfy $d_c \geq 2$ for every $c \in \Lambda$, as in \eqref{eq:geo-dc-at-least-two}. Third, Assumption~\ref{ass:geo-structural-nondeg} holds, i.e.~the local cofactor polynomial $S_c$ is not identically zero for every $c$.

The nonempty-anchor condition is used only in the patching argument: it ensures that each site has at least one anchored Pauli coordinate through which neighboring local scalar factors can be compared. Dense finite-range Pauli families containing all one-site Pauli terms satisfy the nonempty-anchor condition directly. Exact two-body finite-range Pauli families without on-site fields satisfy it by Lemma~\ref{lem:exact-two-body-nonempty-anchoring}. In all cases, however, the structural condition $S_c \not\equiv 0$ is a separate local rank condition and is claimed only for the support families certified in Appendix~\ref{subsubsec:deterministic-certificates}.

\begin{remark}[Sampling before normalization]
\label{rem:sampling-before-normalization}
The anti-concentration assumptions in Theorem~\ref{thm:geo-global-gap} are assumptions on the coefficient vector to which the local small-ball estimate is applied. In normalized ensemble statements, we first sample an unnormalized coefficient vector $\widetilde h$, apply the local small-ball and coefficient-tail estimates to $\widetilde h$, and only then set
\begin{equation}
H = \sqrt{\frac{n}{\widetilde S}} \sum_{u \in \mathcal V} \widetilde h_u P_u\,, \quad \widetilde S = \sum_{u \in \mathcal V} \widetilde h_u^2.
\end{equation}
This distinction matters because the globally normalized vector need not have independent or log-concave local marginals. The passage from $\widetilde h$ to the normalized $h = \sqrt{n/\widetilde S}\,\widetilde h$ is handled by the homogeneity relations
\begin{equation}
B(h) = \sqrt{\frac{n}{\widetilde S}}\, B(\widetilde h),
\quad
\Gamma_c(h) = \frac{n}{\widetilde S}\, \Gamma_c(\widetilde h), \quad \lambda_2(\Gamma_c(h)) = \frac{n}{\widetilde S}\, \lambda_2(\Gamma_c(\widetilde h)).
\end{equation}
Similarly, if
\begin{equation}
\widetilde h_{\min} := \min_{u \in \mathcal V} |\widetilde h_u|, \quad \widetilde h_{\max} := \max_{u \in \mathcal V} |\widetilde h_u|,
\end{equation}
then
\begin{equation}
h_{\min} = \sqrt{\frac{n}{\widetilde S}}\, \widetilde h_{\min}, \quad h_{\max} = \sqrt{\frac{n}{\widetilde S}}\, \widetilde h_{\max}.
\end{equation}
Consequently, the coefficient ratio appearing in Proposition~\ref{prop:deterministic-geometric-gap-from-local-gaps} is invariant under this normalization:
\begin{equation}
\label{coefficient-ratio-propc1}
\frac{h_{\min}^2}{2h_{\min}^2 + n^2\Delta d_{\rm blk}h_{\max}^2} = \frac{\widetilde h_{\min}^2}
{2\widetilde h_{\min}^2 + n^2\Delta d_{\rm blk}\widetilde h_{\max}^2}.
\end{equation}
\end{remark}

We now apply this scaling reduction to the standard coefficient ensembles. Since the final normalized Hamiltonian is unchanged if all unnormalized coefficients are multiplied by the same positive constant, the common scale $\sigma$ drops out. Equivalently, for unnormalized coefficients sampled from $\operatorname{Unif}[-\sigma,\sigma]$ or $\mathcal{N}(0,\sigma^2)$, divide the unnormalized coefficient vector by $\sigma$ before applying the bounds below. In both applications below, Lemma~\ref{lem:geo-local-smallball} is applied to the unnormalized coefficient vector $\widetilde h$, whose local marginals remain independent and log-concave. Proposition~\ref{prop:deterministic-geometric-gap-from-local-gaps} is then applied to the normalized Hamiltonian, using Remark~\ref{rem:sampling-before-normalization} and the coefficient-ratio identity \eqref{coefficient-ratio-propc1}. First consider independent uniform unnormalized coefficients, after removing the common scale:
\begin{equation}
\widetilde h_u\overset{\rm i.i.d.}{\sim}\operatorname{Unif}[-1,1].
\end{equation}
The local coefficient marginals satisfy Assumption~\ref{ass:geo-flat-core} with
\begin{equation}
s_n=1,\quad \kappa_0=2^{-d_{\rm loc}},\quad r_0=1.
\end{equation}
Let $\eta_{\rm unif}>0$ be the corresponding constant from Lemma \ref{lem:geo-flat-core-L1}, so that
\begin{equation}
\|S_c\|_{L^1(\mu_c)}\ge \eta_{\rm unif}
\end{equation}
for every $c$. Define
\begin{equation}
M_{\rm unif}:=\max\{1,d_{\rm loc}w_{\rm loc}\},
\end{equation}
and, for failure parameters $\delta_{\rm loc},\delta_{\min}\in (0,1)$, set
\begin{equation}
\widetilde \xi_{\rm unif}:=
\frac{\eta_{\rm unif}}{M_{\rm unif}^{d_{\rm loc}-2}}
\left(
\frac{\delta_{\rm loc}}{nC_{\rm CW}D_\star}
\right)^{D_\star}.
\end{equation}
The choice of $\widetilde \xi_{\rm unif}$ makes the local small-ball failure probability at most
$\delta_{\rm loc}$. Also, Lemma~\ref{lem:uniform-coefficient-bounds} applied at unit scale gives
\begin{equation}
\widetilde h_{\min}\ge \frac{\delta_{\min}}{N},
\quad
\widetilde h_{\max}\le 1\,,
\end{equation}
with probability at least $1-\delta_{\min}$, while $\widetilde S\le N$ deterministically. Therefore,
with probability at least $1-\delta_{\rm loc}-\delta_{\min}$,
\begin{equation}\label{independent-uniform-dense-2local}
\min_{\substack{
A\in \mathcal O^{\rm op}_{k,R}\\
\tfrac{1}{d}\tr(A^2)=1,\ \tr(AH)=0}}
\frac1d\|[A,H]\|_F^2
\ge
4\frac{n}{N}\widetilde \xi_{\rm unif}
\frac{(\delta_{\min}/N)^2}
{2(\delta_{\min}/N)^2+n^2\Delta d_{\rm blk}}.
\end{equation}

\noindent Next consider independent Gaussian unnormalized coefficients, after removing the common scale:
\begin{equation}
\widetilde h_u \stackrel{\,\textnormal{i.i.d.}\,}{\sim} \mathcal{N}(0,1).
\end{equation}
Fix $\delta_{\rm loc},\delta_{\min},\delta_{\max} \in (0,1)$, and set
\begin{equation}
L_{\max} := 2\log\left(\frac{2N}{\delta_{\max}}\right), \quad H_{\max} := \sqrt{L_{\max}}.
\end{equation}
The Gaussian local marginals satisfy Assumption~\ref{ass:geo-flat-core} with $s_n = 1$, $\kappa_0 = (2\pi e)^{-d_{\rm loc}/2}$, and $r_0 = 1$. Let $\eta_{\rm Gauss} > 0$ be the corresponding constant from Lemma~\ref{lem:geo-flat-core-L1}, so that
\begin{equation}
\|S_c\|_{L^1(\mu_c)} \geq \eta_{\rm Gauss}
\end{equation}
for every $c$. Since $H_{\max}^2 = L_{\max}$, the parameter $M_\star$ in Lemma~\ref{lem:geo-local-smallball} becomes
\begin{equation}
M_{\rm Gauss} := \max\{1,d_{\rm loc}w_{\rm loc}L_{\max}\}.
\end{equation}
We then set
\begin{equation}
\widetilde \xi_{\rm Gauss} := \frac{\eta_{\rm Gauss}}{M_{\rm Gauss}^{d_{\rm loc} - 2}} \left(\frac{\delta_{\rm loc}}{n C_{\rm CW} D_\star} \right)^{D_\star}.
\end{equation}
The local small-ball event fails with probability at most $\delta_{\rm loc}$. By Lemma \ref{lem:gaussian-coefficient-bounds},
\begin{equation}
\widetilde h_{\min}\ge \frac{\delta_{\min}}{N}, \quad \widetilde h_{\max}\le \sqrt{L_{\max}}
\end{equation}
with probability at least $1-\delta_{\min}-\delta_{\max}$. Finally, the standard $\chi$-square tail
bound gives
\begin{equation}
\widetilde S = \sum_{u\in \mathcal V}\widetilde h_u^2\le 2N
\end{equation}
with probability at least $1-e^{-N/8}$. Therefore, with probability at least
$1-\delta_{\rm loc}-\delta_{\min}-\delta_{\max}-e^{-N/8}$,
\begin{equation}\label{independent-gaussian-dense-2local}
\min_{\substack{
A\in \mathcal O^{\rm op}_{k,R}\\
\tfrac{1}{d}\tr(A^2)=1,\ \tr(AH)=0}} \frac1d\|[A,H]\|_F^2 \ge 4\frac{n}{2N}\widetilde \xi_{\rm Gauss} \frac{(\delta_{\min}/N)^2}{2(\delta_{\min}/N)^2+n^2\Delta d_{\rm blk}L_{\max}}.
\end{equation}

For fixed $k,R,\Delta$, if $N = O(n)$, and the failure parameters are chosen inverse-polynomially, both \eqref{independent-uniform-dense-2local} and \eqref{independent-gaussian-dense-2local} give inverse-polynomial commutator gaps with probability at least $1 - 1/\operatorname{poly}(n)$.

The same scaling argument gives a smoothed-analysis version. Here the unnormalized coefficient vector is an arbitrary deterministic center plus an independent Gaussian perturbation. Thus a generic local perturbation removes accidental low-weight conserved directions, provided the same geometric support-family hypotheses hold.

\begin{theorem}[Smoothed geometric commutator gap]
\label{thm:smoothed-geometric-gap}
Assume the geometric hypotheses of Theorem~\ref{thm:geo-global-gap}, including a support-respecting nonempty anchoring map, $d_c \geq 2$, and $S_c \not\equiv 0$ for every rooted local type. Let $N = |\mathcal V|$. Fix an arbitrary center $\mu \in \mathbb R^N$, and sample unnormalized coefficients
\begin{equation}
\widetilde h = \mu + \sigma g, \quad g_u \stackrel{\rm i.i.d.}{\sim} \mathcal{N}(0,1), \quad 0 < \sigma \leq 1.
\end{equation}
Let
\begin{equation}
\widetilde H = \sum_{u \in \mathcal V} \widetilde h_u P_u, \quad \widetilde S = \|\widetilde h\|_2^2, \quad H = \sqrt{\frac n{\widetilde S}}\, \widetilde H.
\end{equation}
Let $M_\infty \geq \|\mu\|_\infty$ and $M_2 \geq \|\mu\|_2$. For failure parameters $\delta_{\rm loc},\delta_{\min},\delta_{\max},\delta_S \in (0,1)$, define
\begin{align}
L_{\max} &:= 2 \log\!\left(\frac{2N}{\delta_{\max}}\right) \\
\widetilde{h}_{\max}^{\rm ub} &:= M_\infty + \sigma \sqrt{L_{\max}} \\
S_{\rm ub} &:= \left(M_2 + \sigma\left(\sqrt N + \sqrt{2\log(1/\delta_S)}\right)\right)^2 \\
\varepsilon_{\min} &:= \frac{\sigma \delta_{\min}}{N} \\
M_{\rm sm} &:= \max\{1,d_{\rm loc} w_{\rm loc} (\widetilde{h}_{\max}^{\rm ub})^2\} \\
\widetilde\xi_{\rm sm} &:= \frac{\eta_{\rm sm} \sigma^{D_\star}}{M_{\rm sm}^{d_{\rm loc} - 2}} \left(\frac{\delta_{\rm loc}}{n C_{\rm CW} D_\star}\right)^{D_\star}\,,
\end{align}
where $\eta_{\rm sm}>0$ is the constant from
Lemma~\ref{lem:geo-smoothed-L1}. Then, with probability at least
\begin{equation}
1 - \delta_{\rm loc} - \delta_{\min} - \delta_{\max} - \delta_S,
\end{equation}
the normalized Hamiltonian $H$ satisfies
\begin{equation}
\min_{\substack{A \in \mathcal O_{\mathcal V}^{\rm op}\\ \tfrac{1}{d}\tr(A^2) = 1,\ \tr(AH) = 0}} \frac{1}{d}\|[A,H]\|_F^2 \geq 4\, \frac{n}{S_{\rm ub}}\, \widetilde\xi_{\rm sm}\, \frac{\varepsilon_{\min}^2}{2 \varepsilon_{\min}^2 + n^2 \Delta d_{\rm blk} (\widetilde{h}_{\max}^{\rm ub})^2}.
\end{equation}
In particular, for fixed $(\Delta,k,R)$, if $N = O(n)$, $M_\infty$, $M_2$, and $\sigma^{-1}$ are polynomially bounded in $n$, then the smoothed ensemble has an inverse-polynomial static commutator gap with probability at least $1 - 1/\operatorname{poly}(n)$.
\end{theorem}

\begin{proof}
The shifted-Gaussian small-interval bound, the Gaussian maximum bound, and Gaussian concentration for the Euclidean norm give, with failure probability at most $\delta_{\min} + \delta_{\max} + \delta_S$,
\begin{equation}
\widetilde h_{\min} \geq \varepsilon_{\min}, \quad
\widetilde h_{\max} \leq \widetilde h_{\max}^{\rm ub}, \quad
\widetilde S \leq S_{\rm ub}.
\end{equation}
Lemma~\ref{lem:geo-smoothed-L1} gives $\|S_c\|_{L^1} \geq \eta_{\rm sm}\sigma^{D_\star}$. With the chosen value of $\widetilde\xi_{\rm sm}$, Lemma~\ref{lem:geo-local-smallball} then gives
\begin{equation}
\lambda_2(\Gamma_c(\widetilde h)) \geq \widetilde\xi_{\rm sm}
\end{equation}
for every $c$, with failure probability at most $\delta_{\rm loc}$. By homogeneity,
\begin{equation}
\lambda_2(\Gamma_c(h)) = \frac n{\widetilde S}\, \lambda_2(\Gamma_c(\widetilde h)) \geq \frac n{S_{\rm ub}}\, \widetilde\xi_{\rm sm}.
\end{equation}
Finally apply Proposition~\ref{prop:deterministic-geometric-gap-from-local-gaps} and use the scale-invariant coefficient-ratio identity
\begin{equation}
\frac{h_{\min}^2}{2 h_{\min}^2 + n^2 \Delta d_{\rm blk} h_{\max}^2} = \frac{\widetilde h_{\min}^2}{2 \widetilde h_{\min}^2 + n^2 \Delta d_{\rm blk} \widetilde h_{\max}^2}.
\end{equation}
The inverse-polynomial consequence follows by substituting the polynomial bounds.
\end{proof}

\subsection{Summary of the commutator gap results}
\label{subsec:commutator-gap-summary}

We collect the static commutator-gap bounds proved in this appendix. These are the $\underline{\pi}_H$ inputs that the sinc reduction in Appendix~\ref{section-sinc} converts into dynamical $\underline{\pi}_U$ bounds.

For reference, we use the notation from the preceding sections.  In the Pauli-injective regime, let $\mathcal{V}$ be a finite Pauli family whose anticommutation graph is connected and edge-product injective. For
\begin{align}
H = \sum_{u \in \mathcal{V}} h_u P_u, \quad h_{\min} := \min_{u \in \mathcal{V}} |h_u| > 0, \quad h_{\max} := \max_{u \in \mathcal{V}} |h_u|,
\end{align}
Theorem~\ref{thm:pauli-injective-gap} gives
\begin{align}
\min_{\substack{A \in \mathcal{O}_{\mathcal{V}}^{\mathrm{op}}\\ \tfrac{1}{d}\tr(A^2) = 1,\ \tr(AH) = 0}} \frac{1}{d}\|[A,H]\|_F^2 \geq \underline{\pi}_H^{\mathrm{inj}},
\label{eq:summary-piH-inj}
\end{align}
where
\begin{align}
\underline{\pi}_H^{\mathrm{inj}} := \frac{16\,h_{\min}^4}{h_{\max}^2 |\mathcal V|^2\bigl(1+\frac{|\mathcal V|h_{\max}^2}{n}\bigr)}.
\end{align}
Moreover, Lemma~\ref{lem:dominating-extension} shows that the same lower bound holds on a larger search family $\mathcal{V}_{\max}$ when $\mathcal{V}$ dominates $\mathcal{V}_{\max}$ and the pair $\mathcal{V} \subseteq \mathcal{V}_{\max}$ is relatively edge-product injective.

In the geometric regime, suppose the local decomposition is defined on a bounded-degree interaction graph $G = (\Lambda,\textsf{E})$, and suppose that the deterministic local gap condition
\begin{align}
\lambda_2(\Gamma_c(h)) \geq \xi \quad \textnormal{for every } c \in \Lambda
\end{align}
holds for some $\xi > 0$. Then Proposition~\ref{prop:deterministic-geometric-gap-from-local-gaps} gives
\begin{align}
\min_{\substack{A \in \mathcal{O}_{k,R}^{\mathrm{op}}\\ \tfrac{1}{d}\tr(A^2) = 1,\ \tr(AH) = 0}} \frac{1}{d}\|[A,H]\|_F^2 \geq \underline{\pi}_H^{\mathrm{geo}}(\xi,H),
\label{eq:summary-piH-geo-det}
\end{align}
where
\begin{align}
\underline{\pi}_H^{\mathrm{geo}}(\xi,H) := 4\xi\,\frac{h_{\min}^2}{2h_{\min}^2 + n^2\Delta d_{\mathrm{blk}}h_{\max}^2}.
\end{align}
This is a deterministic certificate: once the local eigenvalue bounds are checked, the global static commutator gap follows.

Theorem~\ref{thm:geo-global-gap} gives a probabilistic way to certify the local gaps. Under structural local nondegeneracy, log-concavity of the local coefficient laws, and the moment lower bound
\begin{align}
\|S_c\|_{L^1(\mu_c)} \geq \eta_0 \quad \textnormal{for every } c,
\end{align}
we have, for every $\xi > 0$ and $H_{\max} > 0$,
\begin{align}
\min_{\substack{A \in \mathcal{O}_{k,R}^{\mathrm{op}}\\ \tfrac{1}{d}\tr(A^2) = 1,\ \tr(AH) = 0}} \frac{1}{d}\|[A,H]\|_F^2 \geq \underline{\pi}_H^{\mathrm{geo}}(\xi,H)
\end{align}
with probability at least
\begin{align}
1 - \mathbb{P}(\mathcal{E}_{\max}(H_{\max})^c) - n C_{\mathrm{CW}}D_\star \min\!\left\{1,\left(\frac{\xi M_\star^{d_{\mathrm{loc}} - 2}}{\eta_0}\right)^{1/D_\star}\right\}.
\end{align}
The moment lower bound is verified from local density lower bounds by Lemma~\ref{lem:geo-flat-core-L1}, and from Gaussian smoothing around an arbitrary center by Lemma~\ref{lem:geo-smoothed-L1}.

For the normalized explicit ensembles in Section \ref{locally-nondegenerate-geo-family-subsec}, the probabilistic local-gap
estimate is applied to the unnormalized coefficient vector $\widetilde h$ before normalization,
and the resulting local spectral gap is then rescaled by the factor $n/\widetilde S$.

Consequently, for fixed $(\Delta,k,R)$, suppose $N = |\mathcal{V}| = O(n)$, and suppose that, with high probability,
\begin{align}
h_{\min} \geq n^{-C}, \quad h_{\max} \leq n^C, \quad H_{\max} \leq n^C, \quad \eta_0 \geq n^{-C},
\end{align}
for some constant $C > 0$. Then choosing $\xi = n^{-C'}$ for $C'$ sufficiently large gives
\begin{align}
\min_{\substack{A \in \mathcal{O}_{k,R}^{\mathrm{op}}\\ \tfrac{1}{d}\tr(A^2) = 1,\ \tr(AH) = 0}} \frac{1}{d}\|[A,H]\|_F^2 \geq \frac{1}{\operatorname{poly}(n)}
\end{align}
with probability at least $1 - 1/\operatorname{poly}(n)$.

Finally, suppose that a static lower bound
\begin{align}
\frac{1}{d}\|[A,H]\|_F^2 \geq \underline{\pi}_H
\end{align}
holds for every admissible normalized $A\perp H$. Since $r=|\mathcal{V}|-1$ when $H$ is the only conserved quantity, the uniform dynamical reduction in Theorem~\ref{relation-commutator-A-U} implies that, with probability at least $1-\delta$ over a single sampled time $t\sim\operatorname{Unif}([0,T])$, one has simultaneously for every admissible normalized $A\perp H$,
\begin{align}
\frac{1}{d}\|[U(t),A]\|_F^2 \geq \frac{\delta^2\min\{T^2,1/n\}}{2048|\mathcal{V}|^3}\,\underline{\pi}_H.
\end{align}
under the normalization convention $\tfrac{1}{d}\tr(H^2)=n$. Thus, whenever $\underline{\pi}_H$ is inverse-polynomial, $|\mathcal{V}|=\operatorname{poly}(n)$, $\delta\geq1/\operatorname{poly}(n)$, and $\min\{T^2,1/n\}\ge 1/\mathrm{poly}(n)$, the dynamical commutator gap remains inverse-polynomial, even when the observable $A$ is chosen after $t$ is sampled.

\section{Hamiltonian Learning}\label{sec:learning-hamiltonian}
Let $\mathcal V$ be a Pauli family of size $N:=|\mathcal{V}|$ and maximal weight $k$ and let $\mathcal{O}_{\mathcal{V}}^{\mathrm{op}}=\mathrm{span}_{\mathbb R} \{P_v: v\in \mathcal{V} \}$. Let $H$ be some traceless Hamiltonian in the real span of the Pauli family $\mathcal{V}$ with coefficient vector $h$ with the normalization convention $\tfrac{1}{d}\tr(H^2)=n$, equivalently $\sum_{v\in\mathcal V}h_v^2=n$. We denote the unit Hamiltonian direction by
\begin{equation}
h_{\mathrm{unit}}:=\frac{h}{\sqrt{n}}\in\mathbb R^N,\quad \|h_{\mathrm{unit}}\|_2=1,
\end{equation}
so that $H=\sqrt{n}\sum_{v\in\mathcal V}(h_{\mathrm{unit}})_vP_v$. The unit-direction $h_{\mathrm{unit}}$ is what the singular-vector argument below estimates. Let $U=e^{-iHt}$. Define the dynamical quantity condition
\begin{equation}\label{dynamical condition}
\pi_U:=\min_{\substack{A\in \mathcal{O}_{\mathcal{V}}^{\mathrm{op}},\\\tr(AH)=0,\\ \tfrac{1}{d}\tr(A^2)=1}}\frac{1}{d}\Big\|[U,A]\Big\|_F^2\ge \underline{\pi}_U.
\end{equation}
Appendices~\ref{section-sinc} and~\ref{sec:commutator-gaps} establish this condition, with $\underline{\pi}_U \geq 1/\operatorname{poly}(n)$, for the model classes considered there, including random ensembles and deterministic certified instances. This section assumes such a bound as a promise and focuses on learning $H$ from $U = e^{-iHt}$.
\subsection{Algorithm}
\begin{algorithm}[H]
\caption{Hamiltonian learning from conservation laws}
\label{alg:learn-H-general}
\begin{algorithmic}[1]
\Statex \textbf{Input.}
Target unitary $U$ on $n$ qubits, with Hilbert-space dimension $d = 2^n$; Pauli family $\mathcal V$; a certified lower bound $\underline\pi_U \in (0,4]$ on the dynamical commutator gap; number of probes $M_{\mathrm{probe}}$; shots per probe $M_{\mathrm{sh}}$; accuracy and confidence parameters $\varepsilon,\varepsilon_{\mathrm{cs}},\eta \in (0,1)$.
\Statex \textbf{Promise.}
$U = e^{-iHt}$ for some unknown traceless Hamiltonian $H = \sum_{v \in \mathcal V} h_v P_v$, and
\begin{equation}
\pi_U := \min_{\substack{A \in \mathcal O_{\mathcal V}^{\rm op}\\ \tfrac{1}{d}\tr(A^2) = 1,\ \tr(AH) = 0}} \frac{1}{d}\|[U,A]\|_F^2 \geq \underline\pi_U.
\end{equation}
\Statex \textbf{Output.}
An estimate $\widehat H=\sum_{v\in\mathcal V}\widehat h_vP_v$, normalized so that $\tr(\widehat H^2)=dn$, such that
\begin{equation}\label{reconstruction-error}
        \min_{s\in\{\pm1\}}\|\widehat H-sH\|_{\mathrm{op}}\le\varepsilon
\end{equation}
with probability at least $1-\eta$.
\State Fix an ordering of $\mathcal V$ and identify it with $\{1,\dots,N\}$.
\For{$\ell=1,\dots,M_{\mathrm{probe}}$}
    \State Sample a product-state probe $\rho^{(\ell)}=\bigotimes_{j=1}^n\rho_j^{(\ell)}$, where each $\rho_j^{(\ell)}$ is drawn independently and uniformly from the six Pauli eigenstates
    $\{\ket{0}\!\bra{0},\ket{1}\!\bra{1},\ket{+}\!\bra{+},\ket{-}\!\bra{-},\ket{y^+}\!\bra{y^+},\ket{y^-}\!\bra{y^-}\}$.
    \State Prepare $\rho^{(\ell)}$, evolve under $U$, and perform classical shadows~\cite{HuangKuengPreskill2020} with $M_{\mathrm{sh}}$ shots to obtain estimates
    $\widehat o_v^{(\ell)}\approx \tr(P_vU\rho^{(\ell)}U^\dagger)$ simultaneously for all $v\in\mathcal V$.
    \State Assume the shadow estimates satisfy
    $|\widehat o_v^{(\ell)}-\tr(P_vU\rho^{(\ell)}U^\dagger)|\le\varepsilon_{\mathrm{cs}}$
    for all $(\ell,v)$ with probability at least $1-\eta/2$.
\EndFor
\State Form the data matrix $\widehat X\in\mathbb R^{M_{\mathrm{probe}}\times N}$ with entries
$\widehat X_{\ell v}:=\tr(P_v\rho^{(\ell)})-\widehat o_v^{(\ell)}$.
\State Compute the unit right singular vector $\widehat h_{\mathrm{unit}}\in\mathbb R^N$ of $\widehat X$ associated with its smallest singular value.
\State Return $\widehat H=\sqrt{n}\sum_{v\in\mathcal V}(\widehat h_{\mathrm{unit}})_vP_v$, which then satisfies $\tr(\widehat H^2)=dn$.
\end{algorithmic}
\end{algorithm}

\begin{lemma}[Uniform local-Pauli shadow accuracy]
\label{lem:uniform-shadow-accuracy}
Let $\mathcal V$ be a Pauli family of size $N$ and weight at most $k$. Fix product probes
$\rho^{(1)},\ldots,\rho^{(M_{\rm probe})}$. 
For each $\ell$, perform $M_{\rm sh}$ independent local-Pauli classical-shadow measurements on $U\rho^{(\ell)} U^\dagger$, and estimate each
\begin{equation}
o_v^{(\ell)}:=\operatorname{tr}(P_vU\rho^{(\ell)} U^\dagger)
\end{equation}
using the standard median-of-means classical-shadow estimator of \cite{HuangKuengPreskill2020}. There is a universal constant $C_{\rm sh}>0$ such that if
\begin{equation}
M_{\rm sh}\ge C_{\rm sh}\,\frac{3^k}{\epsilon_{\rm cs}^2}
\log\!\left(\frac{2NM_{\rm probe}}{\eta}\right),
\end{equation}
then, with probability at least $1-\eta/2$,
\begin{equation}
\left|\widehat o_v^{(\ell)} -\operatorname{tr}(P_vU\rho^{(\ell)}U^\dagger)\right| \le \epsilon_{\rm cs}
\end{equation}
simultaneously for all $\ell\in\{1,\ldots,M_{\rm probe}\}$ and all $v\in \mathcal V$.
\end{lemma}

\begin{proof}
For local-Pauli classical shadows, the shadow norm of a Pauli observable $P_v$ is bounded by
\begin{equation}
\|P_v\|_{\rm sh}^2 \le 3^{\operatorname{wt}(v)} \le 3^k .
\end{equation}
The median-of-means guarantee for classical shadows \cite{HuangKuengPreskill2020} therefore implies that, for any fixed pair $(\ell,v)$,
\begin{equation}
\left|\widehat o_v^{(\ell)}-\operatorname{tr}(P_v\sigma^{(\ell)})\right|
\le \epsilon_{\rm cs}
\end{equation}
with failure probability at most $\eta/(2NM_{\rm probe})$, provided
\begin{equation}
M_{\rm sh}\ge C_{\rm sh}\,\frac{3^k}{\epsilon_{\rm cs}^{2}}
\log\!\left(\frac{2NM_{\rm probe}}{\eta}\right).
\end{equation}
A union bound over the $NM_{\rm probe}$ pairs $(\ell,v)$ gives the simultaneous claim.
\end{proof}

The following parameter choices are sufficient to obtain the reconstruction guarantee~\eqref{reconstruction-error} with probability at least $1 - \eta$. We allocate the failure budget as $\eta/2$ for probe-side concentration and $\eta/2$ for shadow estimation. Taking
\begin{equation}
M_{\mathrm{probe}}\,\ge\, 128\cdot \frac{3^k\,N}{\underline{\pi}_U^{2}}\log(2N/\eta)
\end{equation}
ensures the probe matrix is well-conditioned with probability $1-\eta/2$. Let $X$ denote the noiseless counterpart of $\widehat X$, with entries
$X_{\ell v}:=\tr(P_v\rho^{(\ell)})-\tr(P_vU\rho^{(\ell)}U^\dagger)$. For the estimation error $E:=\widehat X-X$, Lemma~\ref{lem:uniform-shadow-accuracy}, the simultaneous shadow-estimation event in Algorithm~\ref{alg:learn-H-general} holds with probability at least $1-\eta/2$ whenever
\begin{equation}
M_{\rm sh}\ge C_{\rm sh}\,\frac{3^k}{\epsilon_{\rm cs}^{2}} \log\!\left(\frac{2NM_{\rm probe}}{\eta}\right).
\end{equation}
Taking
\begin{equation}
\epsilon_{\rm cs}\le
\frac{\underline{\pi}_U\epsilon}{8\sqrt n\,N\,3^{k/2}}
\end{equation}
therefore gives the sufficient shot bound
\begin{equation}
M_{\rm sh}\ge C\,
\frac{9^k n N^2}{\underline{\pi}_U^{2}\epsilon^{2}}
\log\!\left(\frac{2NM_{\rm probe}}{\eta}\right),
\end{equation}
after increasing the universal constant $C$. Thus the total number of preparations and uses of $U$ is
\begin{equation}
M_{\mathrm{probe}} M_{\mathrm{sh}} = \widetilde O\!\left(27^k\, n\, N^3\, \underline\pi_U^{-4} \varepsilon^{-2}\right).
\end{equation}
\subsection{Rigorous guarantees}\label{rigorous-guarantees}

\begin{theorem}[Robust recovery from the data matrix]\label{thm:robust-recovery-data-matrix}
    Let $\mathcal{V}$ be a traceless Pauli family of size $N$ with maximum weight $k$. Let $H$ be an unknown Hamiltonian in the real span of $\mathcal V$ with corresponding vector $h$ and let $U=e^{-iHt}$. Fix a target accuracy $\varepsilon\in (0,1)$. Assume the dynamical identifiability condition $\pi_U \geq \underline\pi_U$. Run Algorithm~\ref{alg:learn-H-general} with
    \begin{equation}
    M_{\mathrm{probe}} \geq 128 \cdot 3^k N \underline\pi_U^{-2} \log(2N/\eta), \quad
    M_{\mathrm{sh}} \geq C\, \frac{9^k\, n\, N^2}{\underline\pi_U^2\, \varepsilon^2} \log\Bigl(\frac{2 N M_{\mathrm{probe}}}{\eta}\Bigr),
    \end{equation}
    so that the shadow accuracy satisfies
    \begin{equation}
    \varepsilon_{\mathrm{cs}} \leq \frac{\underline\pi_U \varepsilon}{8 \sqrt{n}\, N\, 3^{k/2}}.
    \end{equation}
    Then, with probability at least $1-\eta$ over the probes and shadow measurements,
    \begin{equation}\label{robust-recovery}
        \min_{s\in\{ \pm 1\}}\|\widehat H-sH\|_{\mathrm{op}}\le \varepsilon.
    \end{equation}
    If \eqref{dynamical condition} is supplied by Appendices \ref{section-sinc} and \ref{sec:commutator-gaps} with probability at least $1-\delta_H-\delta_t$ then \eqref{robust-recovery} holds with probability at least $1-\eta-\delta_H-\delta_t$. Moreover, on the same event we have the singular gap,
    \begin{equation}
    \label{singular-gap-thm-claim}
    \sigma_2(\widehat X)-\sigma_1(\widehat X) \geq \left(1-\frac{\varepsilon}{\sqrt{2nN}}\right) \sqrt{\frac{M_{\mathrm{probe}}3^{-k}\underline\pi_U^2}{8}} \geq \Omega\!\left(\sqrt{N\log(N/\eta)}\right).
    \end{equation}
\end{theorem}

We prove the theorem by first analyzing the ideal data matrix $X$ and then controlling the perturbation $\widehat X - X$. Throughout, singular values are listed in nondecreasing order,
\begin{equation}
\sigma_1 \leq \sigma_2 \leq \cdots \leq \sigma_N.
\end{equation}
The six-state product ensemble gives the following second-moment identity, which controls the projection of any traceless Hermitian operator onto a low-weight Pauli subspace.

\begin{lemma}[Second moment of six-state product probes]
\label{lem:second-moment-six-state}
Let $\rho = \bigotimes_{j = 1}^n \rho_j$, where the $\rho_j$'s are sampled independently and uniformly from the six single-qubit Pauli eigenstates. Let $\mathcal P_n^\circ := \mathcal P_n \setminus \{I\}$ denote the set of nonidentity $n$-qubit Pauli strings, modulo phase. If
\begin{equation}
B = \sum_{u \in \mathcal P_n^\circ} b_u P_u
\end{equation}
is any traceless Hermitian operator, then
\begin{equation}
\mathbb E_\rho\!\left[\tr(B\rho)^2\right] = \sum_{u \in \mathcal P_n^\circ} 3^{-\wt(u)}\, b_u^2.
\end{equation}
Furthermore, if $W$ is any Pauli subspace spanned by strings of weight at most $k$, and $\Pi_W$ denotes the orthogonal projector onto $W$ with respect to $\langle\cdot,\cdot\rangle_{HS}$, then
\begin{equation}
\mathbb E_\rho\left[\tr(B\rho)^2\right] \geq 3^{-k}\,\|\Pi_W B\|_{HS}^2.
\end{equation}
If $B \in W$, then
\begin{equation}
3^{-k}\,\|B\|_{HS}^2 \leq \mathbb E_\rho\left[\tr(B\rho)^2\right] \leq \|B\|_{HS}^2.
\end{equation}
\end{lemma}

\begin{proof}
Write each single-qubit Pauli eigenstate probe as $\rho_j=\frac{I+s_j \sigma_{a_j}}{2}$ with $a_j \in \{x,y,z\}$ and $s_j \in \{\pm 1\}$, where $a_j$ and $s_j$ are independent and uniform. For $b\in\{x,y,z\}$,
\begin{equation}
\tr(\sigma_b \rho_j) = \tr\!\left(\sigma_b\,\frac{I+s_j\sigma_{a_j}}{2}\right) = s_j\,\mathbf 1_{\{a_j=b\}},
\end{equation}
and $\tr(I\rho_j)=1$.

Now let $u=(u_1,\dots,u_n)\in\mathcal P_n^\circ$, so 
$P_u=\bigotimes_{j=1}^n \sigma_{u_j},$ with $
\sigma_{u_j}\in\{I,X,Y,Z\}.$
Then
\begin{equation}
\tr(P_u\rho)
=
\prod_{j:u_j\neq I} s_j\,\mathbf 1_{\{a_j=u_j\}}.
\end{equation}
Therefore
\begin{equation}
\mathbb E_\rho\left[\tr(P_u\rho)^2\right] = \prod_{j:u_j\neq I} \mathbb P(a_j=u_j) = \left(\frac{1}{3}\right)^{\wt(u)}.
\end{equation}

Next, if $u\neq v$, choose a site $j$ with $u_j\neq v_j$. Then with the convention $\sigma_I := I$,
\begin{equation}
\mathbb E\left[\tr(\sigma_{u_j}\rho_j)\tr(\sigma_{v_j}\rho_j)\right]=0.
\end{equation}
and so by independence across sites, for $u\neq v$,
\begin{equation}
\mathbb E_\rho\left[\tr(P_u\rho)\tr(P_v\rho)\right]=0.
\end{equation}
Thus,
\begin{equation}
\mathbb E_\rho\left[\tr(B\rho)^2\right] = \sum_{u,v} b_u b_v\, \mathbb E_\rho\left[\tr(P_u\rho)\tr(P_v\rho)\right] = \sum_{u} 3^{-\wt(u)} b_u^2.
\end{equation}
This proves the first identity.

For the projector statement, restrict the sum to the Pauli strings spanning $W$,
\begin{equation}
\mathbb E_\rho\left[\tr(B\rho)^2\right] \ge \sum_{u:\,P_u\in W} 3^{-\wt(u)} b_u^2 \ge 3^{-k}\sum_{u:\,P_u\in W} b_u^2 = 3^{-k}\,\|\Pi_W B\|_{HS}^2.
\end{equation}
If $B\in W$, this gives the lower bound. The upper bound follows from $3^{-\wt(u)}\le 1$,
\begin{equation}
\mathbb E_\rho\left[\tr(B\rho)^2\right] = \sum_u 3^{-\wt(u)} b_u^2 \le \sum_u b_u^2 = \|B\|_{HS}^2.
\end{equation}
\end{proof}

The next step is to specialize this general moment identity to the Heisenberg difference $U^\dagger A U-A.$
This operator need not itself lie in the original model span
$\mathcal{O}_{\mathcal{V}}^{\mathrm{op}}$, so what matters is the size of its $\mathcal V$-projection.
The following lemma shows that this projection is already large enough whenever the commutator
$[U,A]$ is large.

\begin{lemma}[The $\mathcal V$-projection of the Heisenberg difference is large]
\label{lem:projection-heisenberg-difference}
Let $\Pi_{\mathcal V}$ denote the orthogonal projector onto
$\mathcal{O}_{\mathcal{V}}^{\mathrm{op}}$ with respect to $\langle\cdot,\cdot\rangle_{HS}$. For any $A\in \mathcal{O}_{\mathrm{\mathcal{V}}}^{\mathrm{op}}$ such that $\|A\|_{\mathrm{HS}}=1$, 
\begin{equation}
\|\Pi_{\mathcal V} (U^\dagger AU-A)\|_{HS}^2 \ge \frac{1}{4}\left(\frac{1}{d}\|[U,A]\|_F^2\right)^2.
\end{equation}
\end{lemma}

\begin{proof}
Since $A\in \mathcal{O}_{\mathcal{V}}^{\mathrm{op}}$ and $\|A\|_{HS}=1$,
\begin{equation}\label{projection-HS}
\|\Pi_{\mathcal V}( U^\dagger A U-A)\|_{HS} \ge |\langle A,\Pi_{\mathcal V} (U^\dagger A U-A)\rangle_{HS}| = |\langle A,U^\dagger A U-A\rangle_{HS}|.
\end{equation}
Squaring both sides yields
\begin{equation}
\|\Pi_{\mathcal V} (U^\dagger A U-A)\|_{HS}^2 \ge |\langle A,U^\dagger A U-A\rangle_{HS}|^2.
\end{equation}
So it suffices to compute $\langle A,U^\dagger A U-A\rangle_{HS}$. Because $A$ and $U^\dagger A U$ are Hermitian,
\begin{align}
\langle A,U^\dagger A U-A\rangle_{HS} &= \frac{1}{d} \,\tr\bigl(A(U^\dagger A U-A)\bigr) \\
&= -\frac{1}{2}\frac{1}{d}\tr(2A^2-2AU^\dagger AU) \\
&=-\frac{1}{2}\frac{1}{d}\tr((U^\dagger AU-A)^2) \\
&= -\frac{1}{2}\,\|U^\dagger A U-A\|_{HS}^2.
\end{align}
Hence 
\begin{equation}
|\langle A,U^\dagger A U-A\rangle_{HS}|^2 = \frac{1}{4}\,\|U^\dagger A U-A\|_{HS}^4= \frac{1}{4}\left(\frac{1}{d}\|[U,A]\|_F^2\right)^2.
\end{equation}
Combining this with \eqref{projection-HS} proves the claim.
\end{proof}

We now translate the commutator lower bound into a lower bound on the ideal row second-moment matrix of the data matrix. Since the random row vector has mean zero over the six-state product ensemble, this second-moment matrix is also the row covariance. The commutator gap forces every direction orthogonal to the true Hamiltonian to have nontrivial row variance.

\begin{lemma}[Row covariance from the commutator gap]
\label{lem:population-covariance-gap}
For a probe $\rho$, define the random row vector
\begin{equation}
x(\rho) \in \mathbb R^{N}, \quad x_v(\rho) := \tr(P_v\rho) - \tr(P_vU\rho U^\dagger).
\end{equation}
Let
\begin{equation}
\Sigma := \mathbb E_\rho\bigl[x(\rho)x(\rho)^\top\bigr].
\end{equation}
Since $\mathbb E_\rho[x(\rho)] = 0$, the matrix $\Sigma$ is the covariance matrix of one ideal row. For any normalized coefficient vector $a = (a_v)_{v \in \mathcal V} \in \mathbb R^N$, with associated observable $A = \sum_{v \in \mathcal V} a_v P_v$,
\begin{equation}
a^\top \Sigma a = \mathbb E_\rho\left[\bigl(a^\top x(\rho)\bigr)^2\right] \geq \frac{3^{-k}}{4}\left(\frac{1}{d}\|[U,A]\|_F^2\right)^2.
\end{equation}
Moreover, let $P_\perp := I - h_{\mathrm{unit}}h_{\mathrm{unit}}^\top$ be the orthogonal projector onto $h_{\mathrm{unit}}^\perp$. If the dynamical identifiability condition \eqref{dynamical condition} holds, then
\begin{equation}
P_\perp \Sigma P_\perp \succeq \frac{3^{-k}\pi_U^2}{4} P_\perp.
\end{equation}
\end{lemma}

\begin{proof}
For the operator $A=\sum_{v\in\mathcal V} a_v P_v$,
\begin{equation}
a^\top x(\rho) = \sum_{v\in\mathcal V} a_v\Bigl(\tr(P_v\rho)-\tr(P_vU\rho U^\dagger)\Bigr) = \tr(A\rho)-\tr(AU\rho U^\dagger)=\tr\bigl((A-U^\dagger A U)\rho\bigr).
\end{equation}
Therefore
\begin{equation}
a^\top \Sigma a = \mathbb E_\rho\left[\bigl(a^\top x(\rho)\bigr)^2\right] = \mathbb E_\rho\left[\tr((U^\dagger A U-A)\rho)^2\right].
\end{equation}

Apply Lemma~\ref{lem:second-moment-six-state} with
$W=\mathcal{O}_{\mathcal{V}}^{\mathrm{op}}$ and use Lemma ~\ref{lem:projection-heisenberg-difference},
\begin{equation}\label{first-statement}
\mathbb E_\rho\left[\tr((U^\dagger A U-A)\rho)^2\right] \ge
3^{-k}\,\|\Pi_{\mathcal V} (U^\dagger A U-A)\|_{HS}^2 \ge  \frac{3^{-k}}{4}\left(\frac{1}{d}\|[U,A]\|_F^2\right)^2,
\end{equation}
proving the first claim.

For the second statement, take any $a\in h_{\mathrm{unit}}^\perp$ with $\|a\|_2=1$, and let
$A=\sum_{v\in\mathcal V} a_v P_v.$ Then $\|A\|_{HS}=1$, and since $\langle a,h_{\mathrm{unit}}\rangle=0$ is equivalent to $\langle a,h\rangle=0$, we have $\frac{1}{d}\tr(AH)=\langle a,h\rangle =0 .$ Hence \eqref{first-statement} becomes
\begin{equation}
a^\top \Sigma a \ge \frac{3^{-k}\pi_{U}^2}{4}.
\end{equation}
Since this holds for every unit vector $a \in h_{\mathrm{unit}}^\perp$, it implies the corresponding Loewner-order lower bound on that subspace. Indeed, for any $z \in \mathbb R^N$, if $P_\perp z \neq 0$, apply the preceding inequality to
\begin{equation}
a = \frac{P_\perp z}{\|P_\perp z\|_2}.
\end{equation}
Then
\begin{equation}
z^\top P_\perp \Sigma P_\perp z = (P_\perp z)^\top \Sigma (P_\perp z) \geq \frac{3^{-k}\pi_U^2}{4}\|P_\perp z\|_2^2 = z^\top\left(\frac{3^{-k}\pi_U^2}{4}P_\perp\right)z.
\end{equation}
The case $P_\perp z = 0$ is trivial. Therefore
\begin{equation}
P_\perp \Sigma P_\perp \succeq \frac{3^{-k}\pi_U^2}{4}P_\perp.
\end{equation}
\end{proof}

The empirical covariance is controlled by the following matrix Chernoff bound.

\begin{lemma}[Matrix Chernoff lower tail \cite{Tropp2012UserFriendly}]\label{lem:matrix-chernoff}
Let $Y_1,\dots,Y_M$ be independent positive semidefinite matrices on an $r$-dimensional Hilbert
space. Suppose that
\begin{equation}
0\preceq Y_\ell \preceq R I \quad\text{for all }\ell,
\end{equation}
and
\begin{equation}
\lambda_{\min}\bigl(\mathbb E[Y_\ell]\bigr)\ge \mu.
\end{equation}
Then, for every $\delta\in (0,1)$,
\begin{equation}
\mathrm{Pr}\Bigg[
\lambda_{\min}\left(\frac{1}{M}\sum_{\ell=1}^M Y_\ell\right)\le (1-\delta)\mu \Bigg] \le r\,\exp\left(-\frac{M\mu\delta^2}{2R}\right).
\end{equation}
\end{lemma}

We now apply Lemma~\ref{lem:matrix-chernoff} to the covariance of the rows of the data matrix. This
yields the desired lower bound on the second-smallest singular value of $X$.

\begin{lemma}[Empirical covariance gap and singular-value lower bound]
\label{lem:empirical-gap-singular-value}
Assume the dynamical identifiability condition $\pi_U \geq \underline\pi_U$. If
\begin{equation}
M_{\mathrm{probe}} \geq 128\cdot \frac{3^k N}{\underline\pi_U^{2}}\log(2N/\eta),
\end{equation}
then with probability at least $1-\eta/2$ over the probes,
\begin{equation}
\sigma_2(X) \geq \sqrt{\frac{M_{\mathrm{probe}} 3^{-k} \underline\pi_U^2}{8}}.
\end{equation}
\end{lemma}

\begin{proof}
Let $\rho^{(1)},\dots,\rho^{(M_{\mathrm{probe}})}$ be i.i.d.~six-state product probes, and let
\begin{equation}
x_\ell := x(\rho^{(\ell)})\in\mathbb R^{N}, \quad 1\le \ell\le M_{\mathrm{probe}}.
\end{equation}
Then the $\ell$-th row of $X$ is $x_\ell^\top$, and
\begin{equation}
\frac{1}{M_{\mathrm{probe}}}X^\top X = \frac{1}{M_{\mathrm{probe}}}\sum_{\ell=1}^{M_{\mathrm{probe}}} x_\ell x_\ell^\top.
\end{equation}

Let $P_\perp $ be as in Lemma \ref{lem:population-covariance-gap}. Since $U=e^{-iHt}$, we have $U^\dagger H U=H$, so for every probe $\rho$,
\begin{equation}
h_{\mathrm{unit}}^\top x(\rho) = \tfrac{1}{\sqrt{n}}\bigl(\tr(H\rho)-\tr(HU\rho U^\dagger)\bigr) = 0.
\end{equation}
Therefore $x_\ell\in h_{\mathrm{unit}}^\perp$ for every $\ell$, and hence $Xh_{\mathrm{unit}}=0.$

Define $Y_\ell:=P_\perp x_\ell x_\ell^\top P_\perp.$ Each $Y_\ell$ is positive semidefinite. Moreover, for every $v\in\mathcal V$,
\begin{equation}
|x_{\ell,v}| = \bigl|\tr(P_v\rho^{(\ell)})-\tr(P_vU\rho^{(\ell)}U^\dagger)\bigr| \le 2,
\end{equation}
because $\|P_v\|_{\mathrm{op}}=1$ and both $\rho^{(\ell)}$ and $U\rho^{(\ell)}U^\dagger$ are states. Hence $\|x_\ell\|_2^2\le 4N$ and so $0\preceq Y_\ell \preceq 4N\,P_\perp.$

Restricting to the $(N - 1)$-dimensional subspace $h_{\mathrm{unit}}^\perp$, applying Lemma~\ref{lem:matrix-chernoff} with $R = 4N$, $r = N - 1$, $\delta = 1/2$, $M = M_{\mathrm{probe}}$, and
\begin{equation}
\mu = \frac{3^{-k}\underline\pi_U^2}{4}
\end{equation}
gives
\begin{equation}
\mathrm{Pr}\Bigg[ \lambda_{\min}\Bigg(
\frac{1}{M_{\mathrm{probe}}}\sum_{\ell = 1}^{M_{\mathrm{probe}}}Y_\ell\Big|_{h_{\mathrm{unit}}^\perp} \Bigg) \leq \frac{3^{-k}\underline\pi_U^2}{8} \Bigg] \leq (N - 1)\exp\left(-\frac{M_{\mathrm{probe}}3^{-k}\underline\pi_U^2}{128N}\right).
\end{equation}
The assumed lower bound on $M_{\mathrm{probe}}$ ensures that this failure probability is at most $\eta/2$.
Moreover,
\begin{equation}
\frac{1}{M_{\mathrm{probe}}}X^\top X\Big|_{h_{\mathrm{unit}}^\perp} = \frac{1}{M_{\mathrm{probe}}}\sum_{\ell=1}^{M_{\mathrm{probe}}} Y_\ell, \quad Y_\ell = x_\ell x_\ell^\top \ \text{on } h_{\mathrm{unit}}^\perp.
\end{equation}
So on the good event, since $Xh_{\mathrm{unit}}=0$,
\begin{equation}
\sigma_2(X)^2 \geq \frac{M_{\mathrm{probe}}3^{-k}\underline\pi_U^2}{8}.
\end{equation}
Taking square roots proves the claim.
\end{proof}

It remains to pass from the ideal matrix $X$ to the estimated matrix $\widehat X = X + E$. We use two standard perturbation estimates for singular values and singular vectors \cite{StewartSun1990}.

\begin{lemma}[Lower bound on singular values under perturbation]
\label{lem:error-matrix-bound}
Let $A,E\in\mathbb{C}^{M_{\mathrm{probe}}\times N}$, and set
$\widehat A = A+E.$
For every $j\in\{1,\dots,N\}$,
\begin{equation}
\sigma_j(\widehat A)\ \ge\ \sigma_j(A)\,-\, \|E\|_{\mathrm{op}}, \quad \text{and}\quad \sigma_j(\widehat A)\le \sigma_j(A) + \|E\|_{\mathrm{op}}\,.
\end{equation}
\end{lemma}

\begin{proof}
With singular values listed in nondecreasing order, $\sigma_1 \leq \sigma_2 \leq \cdots \leq \sigma_N$, the Courant--Fischer characterization gives
\begin{equation}
\sigma_j(M) = \max_{\substack{W\le \mathbb C^N\\ \dim W = N-j+1}}  \min_{\substack{x\in W \\ \|x\|_2=1}}\,\|Mx\|_2.
\end{equation}
Fix any subspace $W\subset\mathbb C^N$ with $\dim W=N-j+1$. For any unit vector $x\in W$,
\begin{equation}
\|\widehat A x\|_2 = \|(A+E)x\|_2 \ge \|Ax\|_2-\|Ex\|_2 \ge \|Ax\|_2-\|E\|_{\mathrm{op}}.
\end{equation}
Taking the minimum over unit $x\in W$ gives
\begin{equation}
\min_{\substack{x\in W\\ \|x\|_2=1}} \|\widehat A x\|_2 \ge \min_{\substack{x\in W\\ \|x\|_2=1}} \|Ax\|_2 - \|E\|_{\mathrm{op}}.
\end{equation}
Now maximizing over all such subspaces $W$,
\begin{equation}
\sigma_j(\widehat A) = \max_W \min_{\substack{x\in W\\ \|x\|_2=1}} \|\widehat A x\|_2 \ge \max_W \min_{\substack{x\in W\\ \|x\|_2=1}} \|Ax\|_2 - \|E\|_{\mathrm{op}} = \sigma_j(A)-\|E\|_{\mathrm{op}}.
\end{equation}
 The upper bound follows by replacing the reverse triangle inequality by the triangle inequality.
\end{proof}

\begin{lemma}[Principal-angle bound for the perturbed null vector]
\label{lem:sin-theta-bound-error-matrix}
Let $h_{\mathrm{unit}}\in\mathbb R^N$ be a unit vector satisfying $Xh_{\mathrm{unit}}=0$, and let $P_\perp := I-h_{\mathrm{unit}}h_{\mathrm{unit}}^\top$.  Further let $\widehat X=X+E$, and let $\widehat h_{\mathrm{unit}}$ be a unit right singular vector of $\widehat X$ associated with its smallest singular value. Define the principal angle $\theta\in[0,\pi/2]$ by
\begin{equation}
\cos\theta=|\langle \widehat h_{\mathrm{unit}}, h_{\mathrm{unit}}\rangle|, \quad \sin\theta=\|P_{\perp}\widehat h_{\mathrm{unit}}\|_2.
\end{equation}
Then
\begin{equation}
 \sin\theta \,\le\, \frac{2\,\|E\|_{\mathrm{op}}}{\sigma_2(X)}.
\end{equation}
\end{lemma}

\begin{proof}
For any $z\in\mathbb R^N$, since $Xh_{\mathrm{unit}}=0$,
\begin{equation}
Xz = X z- Xh_{\mathrm{unit}}h_{\mathrm{unit}}^\top z=X(I-h_{\mathrm{unit}}h_{\mathrm{unit}}^\top)z=XP_{\perp}z.
\end{equation}
By definition of $\sigma_2(X)$, 
\begin{equation}\label{eq:gap-on-orthogonal-complement}
\|Xz\|_2=\|XP_{\perp}z\|_2 \,\ge\, \sigma_2(X)\,\|P_{\perp }z\|_2.
\end{equation}

Because $\widehat h_{\mathrm{unit}}$ is a right singular vector associated with the smallest singular value of
$\widehat X$, it minimizes $\|\widehat Xx\|_2$ over all unit vectors $x$. Hence
\begin{equation}
\|\widehat X \widehat h_{\mathrm{unit}}\|_2 \,\le\, \|\widehat X h_{\mathrm{unit}}\|_2
\,=\, \|Eh_{\mathrm{unit}}\|_2
\,\le\, \|E\|_{\mathrm{op}}.
\end{equation}
By the triangle inequality,
\begin{equation}
\|X\widehat h_{\mathrm{unit}}\|_2 \le \|\widehat X\widehat h_{\mathrm{unit}}\|_2+\|E\widehat h_{\mathrm{unit}}\|_2 \le 2\,\|E\|_{\mathrm{op}}.
\end{equation}
Applying \eqref{eq:gap-on-orthogonal-complement} with $z = \widehat h_{\mathrm{unit}}$, we obtain
\begin{equation}
\sin\theta=\|P_{\perp}\widehat h_{\mathrm{unit}}\|_2 \le \frac{\|X\widehat h_{\mathrm{unit}}\|_2}{\sigma_2(X)} \le \frac{2\,\|E\|_{\mathrm{op}}}{\sigma_2(X)}.
\end{equation}
\end{proof}

We can now combine the singular-value lower bound for the ideal matrix with the perturbation bounds
for $\widehat X$. The next theorem gives a clean guarantee in terms of the
commutator gap and the entrywise estimation error.

\begin{proof}[Proof of Theorem \ref{thm:robust-recovery-data-matrix}]
By Lemma~\ref{lem:empirical-gap-singular-value}, with probability at least $1-\eta/2$,
\begin{equation}
\label{sigma-2-lower-bound}
\sigma_2(X) \geq \sqrt{\frac{M_{\mathrm{probe}}3^{-k}\underline\pi_U^2}{8}}.
\end{equation}

Also, by the shadow-estimation assumption in the theorem, with probability at least $1-\eta/2$ all entrywise bounds $|E_{\ell v}| \leq \varepsilon_{\mathrm{cs}}$ hold simultaneously. Using
\begin{equation}
\varepsilon_{\mathrm{cs}} \leq \frac{\underline\pi_U\varepsilon}{8\sqrt{n}\,N\,3^{k/2}},
\end{equation}
the entrywise error bound implies
\begin{equation}
\label{eq:upp-bound-E_op}
\|E\|_{\mathrm{op}} \leq \|E\|_F \leq \varepsilon_{\mathrm{cs}}\sqrt{M_{\mathrm{probe}}N} \leq \frac{\varepsilon}{2}\sqrt{\frac{M_{\mathrm{probe}}3^{-k}\underline\pi_U^2}{16 n\,N}} \leq \frac{\varepsilon \sigma_2(X)}{2\sqrt{2nN}}.
\end{equation}
Let $\widehat h_{\mathrm{unit}}$ be the unit right singular vector of $\widehat X$ associated with its smallest singular value, and $\widehat H=\sqrt{n}\sum_{v\in\mathcal V}(\widehat h_{\mathrm{unit}})_v P_v$. Let $\theta \in [0,\pi/2]$ be the principal angle between $\widehat h_{\mathrm{unit}}$ and $h_{\mathrm{unit}}$, defined by
\begin{equation}
\cos\theta = |\langle \widehat h_{\mathrm{unit}},h_{\mathrm{unit}}\rangle|.
\end{equation}
Applying Lemma~\ref{lem:sin-theta-bound-error-matrix} and using \eqref{eq:upp-bound-E_op} gives
\begin{equation}
\sin\theta \leq \frac{2\|E\|_{\mathrm{op}}}{\sigma_2(X)} \leq \frac{\varepsilon}{\sqrt{2nN}}.
\end{equation}

\noindent Finally, choose $s \in \{\pm 1\}$ so that $\langle \widehat h_{\mathrm{unit}}, s\,h_{\mathrm{unit}}\rangle = |\langle \widehat h_{\mathrm{unit}},h_{\mathrm{unit}}\rangle|$. Then
\begin{equation}
\label{2norm}
\|\widehat h_{\mathrm{unit}} - s\,h_{\mathrm{unit}}\|_2^2 = 2 - 2|\langle \widehat h_{\mathrm{unit}},h_{\mathrm{unit}}\rangle| \leq 2\bigl(1 - |\langle \widehat h_{\mathrm{unit}},h_{\mathrm{unit}}\rangle|^2\bigr) = 2\sin^2\theta \leq \frac{\varepsilon^2}{nN}.
\end{equation}
Since $\widehat H=\sqrt{n}\sum_v(\widehat h_{\mathrm{unit}})_vP_v$ and $H=\sqrt{n}\sum_v(h_{\mathrm{unit}})_vP_v$, each Pauli has operator norm $1$, so applying Cauchy--Schwarz and \eqref{2norm} gives
\begin{equation} \|\widehat{H}-sH\|_{\mathrm{op}}\le \sqrt{n}\sum_{v\in \mathcal{V}}\bigl|(\widehat{h}_{\mathrm{unit}})_v-s (h_{\mathrm{unit}})_v\bigr|\le \sqrt{nN}\,\|\widehat h_{\mathrm{unit}}-s\,h_{\mathrm{unit}}\|_2\le \varepsilon.\end{equation}
as claimed. \eqref{singular-gap-thm-claim} follows from applying Lemma~\ref{lem:error-matrix-bound} to $\sigma_2(\widehat X)$ and $\sigma_1(\widehat X)$
\begin{equation}\label{eq:lower-bound-diff-singular-values}
\sigma_2(\widehat X)-\sigma_1(\widehat X)\ge \sigma_2(X)-2\|E\|_{\mathrm{op}} \ge \left(1-\frac{\varepsilon}{\sqrt{2nN}}\right)\sigma_2(X)\ge \Omega\Bigg(\sqrt{N\log(N/\eta)}\Bigg).
\end{equation}
\end{proof}

\subsection{Estimating a representative evolution time }\label{sec:time-estimation}

The procedure in this subsection is an optional coherent post-processing step and is not part of Algorithm~\ref{alg:learn-H-general}. Unlike the learning algorithm, which uses product-state probes and randomized Pauli measurements, the time-estimation routine assumes the ability to prepare a maximally entangled state, to apply $U(t_\star)$ coherently to half of that state, to coherently implement $e^{i\widehat H g}$, and to measure the projector onto the maximally entangled state. It does not require controlled access to $U(t_\star)$.

There is an unavoidable change in the loss function. Without controlled access to $U(t_\star)$, the global phase of $U(t_\star)$ is not observable. Therefore the natural target is projective Frobenius distance, or equivalently distance between the induced unitary channels. For a candidate sign and time pair $(s,g) \in \{\pm 1\} \times [0,T]$, define
\begin{equation}
D_{\rm proj}(s,g) := \min_{\phi \in \mathbb R} \frac{1}{2d}\left\|e^{i\phi}e^{-is\widehat H g} - U(t_\star)\right\|_F^2.
\end{equation}
For unitaries, this has the closed form
\begin{equation}
D_{\rm proj}(s,g) = 1 - \left|\frac{1}{d}\operatorname{tr}\!\left(e^{is\widehat H g}U(t_\star)\right)\right|.
\label{eq:projective-loss}
\end{equation}

Let
\begin{equation}
\ket{\Phi_d} := \frac{1}{\sqrt d}\sum_{k = 1}^{d} \ket{k}\ket{k}
\end{equation}
be the maximally entangled state. For every $X \in \mathbb C^{d \times d}$,
\begin{equation}
\bra{\Phi_d}(X \otimes I)\ket{\Phi_d} = \frac{1}{d}\operatorname{tr}(X).
\end{equation}
Hence, if we define
\begin{equation}
W_{s,g} := e^{is\widehat H g}U(t_\star),
\end{equation}
then the Choi-overlap experiment
\begin{equation}
\ket{\Phi_d} \mapsto (W_{s,g} \otimes I)\ket{\Phi_d}
\end{equation}
followed by a measurement of the projector $\ket{\Phi_d}\!\bra{\Phi_d}$ succeeds with probability
\begin{equation}
p(s,g) = \left|\bra{\Phi_d}(W_{s,g} \otimes I)\ket{\Phi_d}\right|^2 = \left|\frac{1}{d}\operatorname{tr}\!\left(e^{is\widehat H g}U(t_\star)\right)\right|^2.
\label{eq:choi-overlap-probability}
\end{equation}
Thus maximizing $p(s,g)$ over a grid is equivalent to minimizing the projective loss \eqref{eq:projective-loss}.

\begin{algorithm}[H]
\caption{Projective Frobenius time estimation without controlled-$U$}
\label{alg:projective-time-estimation}
\begin{algorithmic}[1]
\Require
Known normalized Hamiltonian direction $\widehat H$ with $\tfrac{1}{d}\operatorname{tr}(\widehat H^2) = n$; ordinary coherent query access to $U(t_\star)$ for some unknown $t_\star \in [0,T]$; target projective loss $\varepsilon \in (0,1)$; failure probability $\delta \in (0,1)$.
\Ensure
A sign and time pair $(\widehat s,\widehat t) \in \{\pm 1\} \times [0,T]$ with small projective Frobenius loss.

\State Set the grid spacing
\begin{equation}
\tau := \frac{\varepsilon}{16\|\widehat H\|_{\mathrm{op}}}.
\end{equation}
\State Form the grid
\begin{equation}
\mathcal G := \left\{0,\tau,2\tau,\ldots,\left\lfloor \frac{T}{\tau}\right\rfloor\tau\right\}\cup\{T\}.
\end{equation}
\State Set
\begin{equation}
m := \left\lceil \frac{32}{\varepsilon^2}\log\!\left(\frac{4|\mathcal G|}{\delta}\right)\right\rceil.
\end{equation}
\For{$s \in \{\pm 1\}$ and $g \in \mathcal G$}
    \For{$j = 1,\ldots,m$}
        \State Prepare
        \begin{equation}
        \ket{\Phi_d} := \frac{1}{\sqrt d}\sum_{k = 1}^{d}\ket{k}\ket{k}.
        \end{equation}
        \State Apply $U(t_\star)$ to the first register.
        \State Apply $e^{is\widehat H g}$ to the first register.
        \State Measure the two-register projector $\ket{\Phi_d}\!\bra{\Phi_d}$, and record the Bernoulli outcome $Y_{s,g,j} \in \{0,1\}$.
    \EndFor
    \State Set
    \begin{equation}
    \widehat p(s,g) := \frac{1}{m}\sum_{j = 1}^{m} Y_{s,g,j}.
    \end{equation}
\EndFor
\State Return any maximizer
\begin{equation}
(\widehat s,\widehat t) \in \arg\max_{s \in \{\pm 1\},\, g \in \mathcal G} \widehat p(s,g).
\end{equation}
\end{algorithmic}
\end{algorithm}

\begin{theorem}[Representative time estimation without controlled-$U$]
\label{thm:projective-time-estimation}
Assume first that $\widehat H = H$ exactly and that $U(t_\star) = e^{-iHt_\star}$ for some $t_\star \in [0,T]$. With probability at least $1 - \delta$, Algorithm~\ref{alg:projective-time-estimation} returns $(\widehat s,\widehat t)$ satisfying
\begin{equation}
\min_{\phi \in \mathbb R} \frac{1}{2d}\left\|e^{i\phi}e^{-i\widehat s H\widehat t} - U(t_\star)\right\|_F^2 \leq \varepsilon.
\end{equation}
Supposing that $T \|H\|_{\rm op}/\varepsilon \gtrsim 1$, the algorithm uses
\begin{equation}
O\!\left(\frac{T\|H\|_{\mathrm{op}}}{\varepsilon^3}\log\frac{T\|H\|_{\mathrm{op}}}{\varepsilon\delta}\right)
\end{equation}
ordinary coherent queries to $U(t_\star)$, together with coherent implementations of $e^{\pm iHg}$ and measurements of $\ket{\Phi_d}\!\bra{\Phi_d}$.
\end{theorem}

\begin{proof}
For each fixed pair $(s,g)$, Hoeffding's inequality gives
\begin{equation}
\Pr\!\left(\left|\widehat p(s,g) - p(s,g)\right| > \frac{\varepsilon}{8}\right) \leq 2\exp\!\left(-\frac{m\varepsilon^2}{32}\right).
\end{equation}
By the choice of $m$ and a union bound over the $2|\mathcal G|$ sign and time pairs, with probability at least $1 - \delta$,
\begin{equation}
\left|\widehat p(s,g) - p(s,g)\right| \leq \frac{\varepsilon}{8}
\end{equation}
simultaneously for every $s \in \{\pm 1\}$ and every $g \in \mathcal G$.

Let $g_0 \in \mathcal G$ satisfy $|g_0 - t_\star| \leq \tau$. Such a point exists by construction of the grid. Since $\widehat H = H$,
\begin{equation}
\left\|e^{-iHg_0} - e^{-iHt_\star}\right\|_{\mathrm{op}} \leq \|H\|_{\mathrm{op}}|g_0 - t_\star| \leq \|H\|_{\mathrm{op}}\tau = \frac{\varepsilon}{16}.
\end{equation}
Using \eqref{eq:projective-loss} with $s = +1$, this implies
\begin{equation}
D_{\rm proj}(+1,g_0) \leq \frac{\varepsilon}{16}.
\end{equation}
Therefore
\begin{equation}
p(+1,g_0) = \left(1 - D_{\rm proj}(+1,g_0)\right)^2 \geq 1 - \frac{\varepsilon}{8}.
\end{equation}
On the simultaneous estimation event, the maximizer $(\widehat s,\widehat t)$ satisfies
\begin{equation}
p(\widehat s,\widehat t) \geq \widehat p(\widehat s,\widehat t) - \frac{\varepsilon}{8} \geq \widehat p(+1,g_0) - \frac{\varepsilon}{8} \geq p(+1,g_0) - \frac{\varepsilon}{4} \geq 1 - \frac{3\varepsilon}{8}.
\end{equation}
Consequently,
\begin{equation}
D_{\rm proj}(\widehat s,\widehat t) = 1 - \sqrt{p(\widehat s,\widehat t)} \leq \frac{3\varepsilon}{8} \leq \varepsilon.
\end{equation}
This proves the loss guarantee. The query complexity follows from
\begin{equation}
|\mathcal G| = O\!\left(\frac{T\|H\|_{\mathrm{op}}}{\varepsilon}\right)
\end{equation}
and the value of $m$.
\end{proof}

\begin{corollary}[Representative time estimation after Hamiltonian learning]
\label{cor:projective-time-estimation-after-learning}
Let $U = e^{-iHt_\star}$ for some $t_\star \in [0,T]$, with $T > 0$. Suppose Algorithm~\ref{alg:learn-H-general} returns $\widehat H$ with the same normalization $\tfrac{1}{d}\operatorname{tr}(\widehat H^2) = n$ and satisfying
\begin{equation}
\min_{s \in \{\pm 1\}}\|\widehat H - sH\|_{\mathrm{op}} \leq \varepsilon'.
\end{equation}
Run Algorithm~\ref{alg:projective-time-estimation} with $\widehat H$ in place of $H$. Then, with probability at least $1 - \delta$, the output $(\widehat s,\widehat t)$ satisfies
\begin{equation}
\min_{\phi \in \mathbb R} \frac{1}{2d}\left\|e^{i\phi}e^{-i\widehat s\widehat H\widehat t} - U(t_\star)\right\|_F^2 \leq O(\varepsilon + T\varepsilon').
\end{equation}
In particular, if $T\varepsilon' \leq \varepsilon$, then the projective Frobenius loss is $O(\varepsilon)$.
\end{corollary}

\begin{proof}
Let $s_0 \in \{\pm 1\}$ satisfy
\begin{equation}
\|s_0\widehat H - H\|_{\mathrm{op}} \leq \varepsilon'.
\end{equation}
For every $g \in [0,T]$, Duhamel's formula gives
\begin{equation}
\left\|e^{-is_0\widehat H g} - e^{-iHg}\right\|_{\mathrm{op}} \leq g\|s_0\widehat H - H\|_{\mathrm{op}} \leq T\varepsilon'.
\end{equation}
The projective loss changes by at most this operator-norm perturbation, since for any unitaries $V,V',U$,
\begin{equation}
\left|\left|\frac{1}{d}\operatorname{tr}(V^\dagger U)\right| - \left|\frac{1}{d}\operatorname{tr}((V')^\dagger U)\right|\right| \leq \frac{1}{d}\left|\operatorname{tr}\!\left((V - V')^\dagger U\right)\right| \leq \|V - V'\|_{\mathrm{op}}.
\end{equation}
Thus the exact-Hamiltonian argument in Theorem~\ref{thm:projective-time-estimation}, applied to the candidate sign $s_0$, incurs an additional additive loss at most $T\varepsilon'$. The empirical maximization step contributes the same $O(\varepsilon)$ term as in Theorem~\ref{thm:projective-time-estimation}. Combining the two contributions gives the claim.
\end{proof}

\section{Weak equilibration}
\label{sec:equilibration}

Equilibration asks whether, under the unitary dynamics generated by a fixed Hamiltonian $H$, simple observables lose memory of their initial configuration at long times, at least up to small fluctuations. This question is physically important because the microscopic dynamics is reversible: any irreversible-looking decay of local correlations must arise from dephasing and from the complexity of the spectrum, rather than from coupling to an external bath.

In this section we derive a correlation consequence of the dynamical commutator gap. The result is weaker than thermalization. It does not show that local autocorrelations vanish, and it does not control arbitrary local correlators. It shows that every normalized local observable orthogonal to $H$ loses an inverse-polynomial amount of its infinite-temperature autocorrelation at the sampled time. Equivalently, within the chosen local span, the Hamiltonian direction is the only direction that can remain almost fixed under the sampled evolution.

Throughout this section, let $\mathcal{V}$ be a generic finite Pauli family of size $N=|\mathcal{V}|$, for example $\mathcal{V}_{\mathrm{inj}}$ in the Pauli-injective regime of Appendix~\ref{subsec:pauli-injective-commutator-gaps} or $\mathcal{V}_{k,R}$ in the geometrically local regime of Appendix~\ref{subsec:geometric-local-nondegeneracy}, and let $\mathcal{O}_{\mathcal{V}}^{\mathrm{op}}$ denote the corresponding real local Pauli span. We write $U(t) = e^{-iHt}$, and for a Hermitian observable $A$, we use the Heisenberg convention $A(t) := U(t)^\dagger A U(t)$.

The input to this section is a uniform dynamical commutator gap. Namely, suppose that for a given time $t$, or on a high-probability event over the choice of $t$, one has
\begin{align}
\pi_U:=\min_{\substack{A \in \mathcal{O}_{\mathcal{V}}^{\mathrm{op}}\\ \tfrac{1}{d}\tr(A^2) = 1,\ \tr(AH) = 0}} \frac{1}{d}\|[U(t),A]\|_F^2 \geq \underline{\pi}_U.
\label{eq:equilibration-input-piU}
\end{align}
In the rest of the paper, such a $\underline{\pi}_U$ is obtained by combining the static commutator gap $\underline{\pi}_H$ bound with the uniform sinc reduction.

\begin{lemma}[Weak local autocorrelation decay]
\label{lem:weak-equilibration}
Assume the uniform dynamical commutator gap \eqref{eq:equilibration-input-piU}. Then every Hermitian $A \in \mathcal{O}_{\mathcal{V}}^{\mathrm{op}}$ satisfying $\frac{1}{d}\tr(A^2) = 1$ and $\tr(AH) = 0$ obeys
\begin{align}
\frac{1}{d}\tr(A(t)A) \leq 1 - \frac{\pi_U}{2} \leq 1 - \frac{\underline{\pi}_U}{2}.
\label{eq:weak-equilibration-bound}
\end{align}
Consequently, if $\underline{\pi}_U \geq 1/\operatorname{poly}(n)$, then every normalized local observable orthogonal to $H$ loses at least an inverse-polynomial fraction of its infinite-temperature autocorrelation.
\end{lemma}

\begin{proof}
Since $U(t)$ is unitary,
\begin{align}
\|[U(t),A]\|_F = \|U(t)^\dagger[U(t),A]\|_F = \|A - U(t)^\dagger A U(t)\|_F = \|A - A(t)\|_F.
\end{align}
Because $A$ and $A(t)$ are Hermitian, we can expand the square:
\begin{align}
\frac{1}{d}\|A - A(t)\|_F^2 &= \frac{1}{d}\tr(A^2) + \frac{1}{d}\tr(A(t)^2) - \frac{2}{d}\tr(A(t)A) = 2 - \frac{2}{d}\tr(A(t)A),
\end{align}
where we used $\tfrac{1}{d}\tr(A^2) = 1$ and $\tfrac{1}{d}\tr(A(t)^2) = \tfrac{1}{d}\tr(A^2) = 1$. Therefore
\begin{align}
\frac{1}{2d}\|[U(t),A]\|_F^2 = 1 - \frac{1}{d}\tr(A(t)A).
\label{eq:commutator-autocorrelation-identity}
\end{align}
Using \eqref{eq:equilibration-input-piU}, we obtain
\begin{align}
1 - \frac{1}{d}\tr(A(t)A) \geq \frac{\pi_U}{2} \geq \frac{\underline{\pi}_U}{2},
\end{align}
which gives \eqref{eq:weak-equilibration-bound}.
\end{proof}

Combining Lemma~\ref{lem:weak-equilibration} with the uniform sinc reduction gives the following direct corollary in terms of the static gap.

\begin{corollary}[Autocorrelation decay from the static commutator gap]
\label{cor:autocorrelation-decay-from-static-gap}
Suppose $H$ is a traceless Hamiltonian normalized so that $\tfrac{1}{d}\tr(H^2)=n$ and that the static commutator gap
\begin{align}
\min_{\substack{A \in \mathcal{O}_{\mathcal{V}}^{\mathrm{op}}\\ \tfrac{1}{d}\tr(A^2) = 1,\ \tr(AH) = 0}} \frac{1}{d}\|[A,H]\|_F^2 \geq \underline{\pi}_H
\end{align}
holds. Suppose also that the uniform sinc reduction gives, for $t \sim \operatorname{Unif}([0,T])$, the simultaneous bound
\begin{align}
\frac{1}{d}\|[U(t),A]\|_F^2 \geq \frac{\delta^2\, \min\{T^2,1/n\}}{2048|\mathcal{V}|^3}\,\underline{\pi}_H
\end{align}
for all normalized $A \in \mathcal{O}_{\mathcal{V}}^{\mathrm{op}}$ orthogonal to $H$, with probability at least $1-\delta$. Then, with the same probability, all such observables satisfy
\begin{align}
\frac{1}{d}\tr(A(t)A) \leq 1 - \frac{\delta^2\, \min \{T^2,1/n\}}{4096|\mathcal{V}|^3}\,\underline{\pi}_H.
\label{eq:autocorrelation-from-static-gap}
\end{align}
In particular, if $\underline{\pi}_H \geq 1/\operatorname{poly}(n)$, $|\mathcal{V}| = \operatorname{poly}(n)$, and $\min\{T^2,1/n\} \geq 1/\operatorname{poly}(n)$, then the right-hand side of \eqref{eq:autocorrelation-from-static-gap} is at most $1 - \frac{1}{\operatorname{poly}(n)}$.
\end{corollary}

\begin{proof}
Apply Lemma~\ref{lem:weak-equilibration} with
\begin{align}
\underline{\pi}_U := \frac{\delta^2\min\{T^2,1/n\}}{2048|\mathcal{V}|^3}\,\underline{\pi}_H.
\end{align}
Then
\begin{align}
\frac{1}{d}\tr(A(t)A) \leq 1 - \frac{\underline{\pi}_U}{2} = 1 -\frac{\delta^2\, \min\{T^2,1/n\}}{4096|\mathcal{V}|^3}\,\underline{\pi}_H,
\end{align}
as claimed.
\end{proof}

The physical interpretation is that the local sector has no almost-conserved direction other than the Hamiltonian direction itself. At infinite temperature, the quantity
\begin{align}
C_A(t) := \frac{1}{d}\,\tr(A(t)A)
\end{align}
is the normalized autocorrelation of the observable $A$. If $C_A(t)$ is close to $1$, then $A(t)$ remains very close to $A$, so $A$ behaves as an approximately conserved local observable at time $t$. The commutator gap rules this out uniformly over all normalized local $A$ orthogonal to $H$.

This is a weak equilibration statement. It rules out local observables orthogonal to $H$ that remain nearly fixed at the sampled time, but it does not imply that $C_A(t)$ is close to zero. It also does not imply state-dependent thermalization, monotone relaxation, or decay of general cross-correlations such as $\frac{1}{d}\tr(A(t)B)$. A unitary could still rotate one local direction into another. The commutator gap excludes only the obstruction relevant for learning, namely a non-Hamiltonian local direction that remains close to itself under the sampled dynamics.

For the random Hamiltonian ensembles considered in the commutator gap Appendix, the static gap is inverse-polynomial with high probability. The uniform sinc reduction then implies that, for most sampled times in some interval, the weak autocorrelation decay bound above holds simultaneously for every normalized local observable orthogonal to the Hamiltonian.

\section{Additional Numerics}\label{Additional Numerics}

\setcounter{figure}{0}
\renewcommand{\thefigure}{\thesection.\arabic{figure}}
\renewcommand*{\theHfigure}{\thesection.\arabic{figure}}

The preceding appendices give uniform high-probability guarantees for Hamiltonian identifiability and learning. These guarantees are conservative: they hold simultaneously over all normalized local observables orthogonal to $H$, over random Hamiltonian instances, and for most sampled times. The numerical experiments below illustrate the typical finite-size behavior behind these bounds.  We examine reconstruction accuracy from finite product-state probes and finite-shot shadow estimates, and we compute the commutator-gap and weak-equilibration quantities that enter the proof. Across the transverse-field Ising model, the XYZ chain, and a general nearest-neighbor two-local model, the Hamiltonian direction is recovered accurately and the non-Hamiltonian local directions remain separated by a visible commutator gap.

\subsection{Identifiability diagnostics}
\begin{figure}[!htbp]
    \centering
    \captionsetup{font=small}
    \begin{subfigure}{0.48\textwidth}
        \centering
        \includegraphics[width=\linewidth]{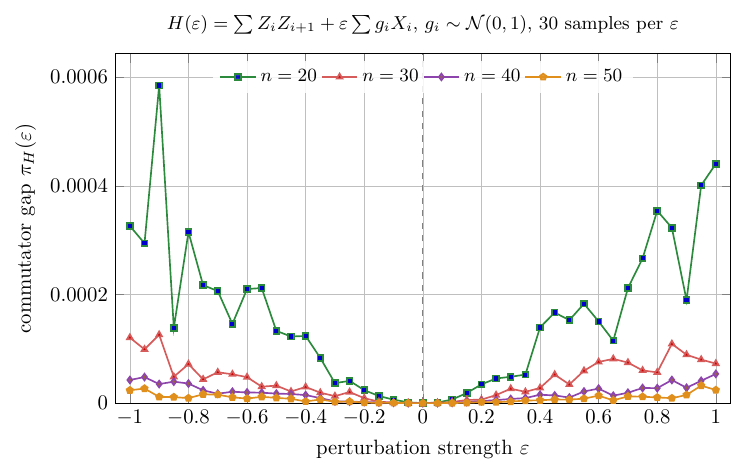}
        \caption{Smoothed commutator gap versus perturbation strength.}
        \label{fig:smoothed-gap-tfim}
    \end{subfigure}
    \hfill
    \begin{subfigure}{0.48\textwidth}
        \centering     \includegraphics[width=\linewidth]{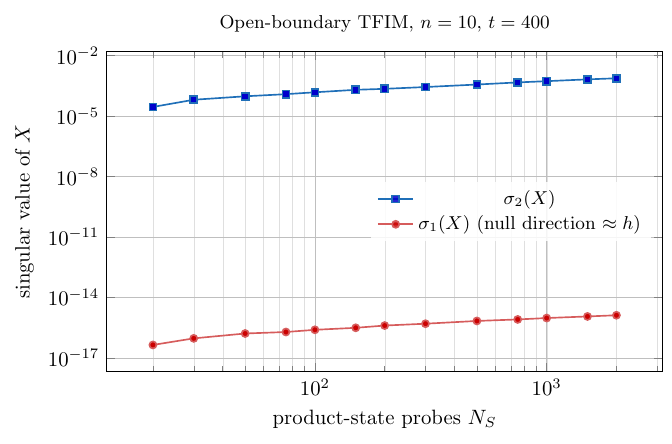}
        \caption{Singular value gap vs number of product-state probes.}
        \label{fig:singular-spectrum-tfim}
    \end{subfigure}
    \caption{TFIM numerics for local identifiability. Panel (a) shows how a random transverse-field perturbation gives a nonzero commutator gap. Panel (b) shows that the second singular value of the learning matrix separates from the null direction as the number of probes increases.}
    \label{fig:tfim-diagnostics}
\end{figure}

Finally, Figure~\ref{fig:tfim-diagnostics}(a) depicts the smoothed transverse-field Ising family
\begin{align}
H(\varepsilon) = \sum_i Z_i Z_{i+1} + \varepsilon \sum_i g_i X_i, \quad g_i \sim \mathcal{N}(0,1),
\end{align}
For this family, we numerically compute the minimum normalized commutator gap after projecting out the Hamiltonian direction. The gap is strongly suppressed near $\varepsilon=0$, reflecting the extra conserved structure of the unperturbed Ising interaction. As $|\varepsilon|$ increases, the random transverse-field perturbation breaks these conservation laws and opens a nonzero local commutator gap. The gap decreases with system size, but remains visibly positive away from the integrable point, supporting the smoothed-analysis picture that a generic local perturbation isolates the Hamiltonian as the unique approximately conserved local direction.

Figure~\ref{fig:tfim-diagnostics}(b) uses the open-boundary TFIM at $n=10$ and $t=400$. We build the learning matrix $X$ from $M_{\mathrm{probe}}$ random product-state probes and plot its two smallest singular values. The smallest singular value $\sigma_1(X)$ remains near numerical precision, as expected because the true Hamiltonian coefficient vector is an exact null direction of the ideal matrix. In contrast, $\sigma_2(X)$ is separated from zero and grows with the number of probes, indicating that directions orthogonal to $H$ are detectably disturbed by the evolution. This empirical singular-value gap is the numerical signature of local identifiability: with sufficiently many probes, the Hamiltonian direction becomes the isolated approximate null vector recovered by the algorithm.

\FloatBarrier
\subsection{Reconstruction resources}
\begin{figure}[!htpb]
\centering
\captionsetup{font=small}
\setlength{\abovecaptionskip}{4pt}
\setlength{\belowcaptionskip}{0pt}

\begin{subfigure}{0.48\textwidth}
    \centering
    \includegraphics[width=\linewidth]{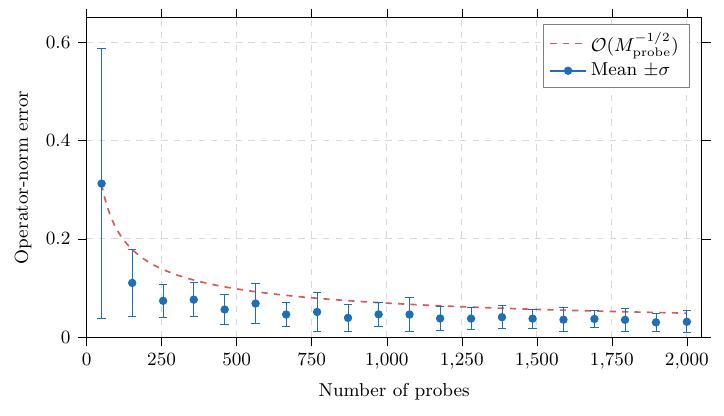}
    \caption{Operator-norm error versus number of probes for TFIM reconstruction at $n=10$, with $t\sim U[1,10]$ and $20{,}000$ shots.}
    \label{fig:probes_opnorm}
\end{subfigure}
\hfill
\begin{subfigure}{0.48\textwidth}
    \centering
    \includegraphics[width=\linewidth]{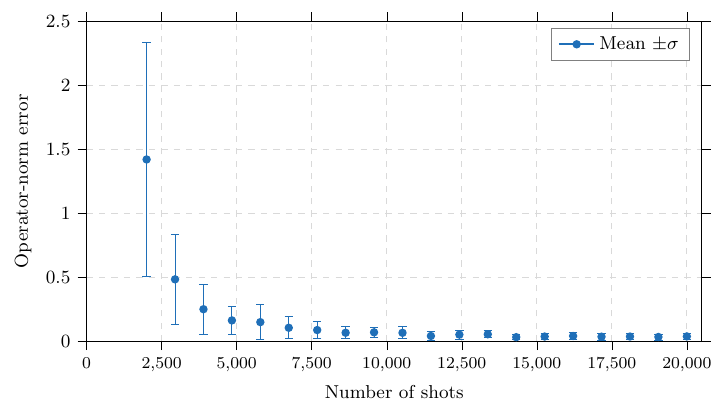}
    \caption{Operator-norm error versus number of shots for TFIM reconstruction at $n=10$, with $t\sim U[1,10]$ and $2000$ probes.}
    \label{fig:shots_opnorm}
\end{subfigure}

\caption{Scaling of reconstruction error with probe and shot number. Each point is averaged over 32 random TFIM reconstructions.}
\label{fig:resource_scaling}
\end{figure}

Figure~\ref{fig:resource_scaling} shows the reconstruction error for the transverse-field Ising model as the number of product-state probes and classical-shadow shots is varied.
In panel~(a), with the number of shots fixed at $20{,}000$, the operator-norm error decreases as the number of probes increases, and the observed scaling is consistent with the $O(M_{\mathrm{probe}}^{-1/2})$ behavior expected from empirical averaging over independent probes.
The error bars are largest at small probe number, where the learning matrix is poorly conditioned, and shrink as additional probes improve the stability of the recovered null vector.
In panel~(b), with $2000$ probes fixed, increasing the number of shots reduces the shadow-estimation noise, leading to a rapid decrease in reconstruction error.
Together, these experiments illustrate the two sources of statistical error in the algorithm: finite-probe covariance estimation and finite-shot shadow estimation.
The monotone decay in both panels supports the predicted polynomial sample behavior of the learning procedure for single-time Hamiltonian reconstruction.

\FloatBarrier
\subsection{Coefficient reconstructions}

\begin{figure}[!htbp]
\centering
\captionsetup{font=small}

\begin{subfigure}[t]{0.31\textwidth}
    \centering
    \includegraphics[width=\linewidth]{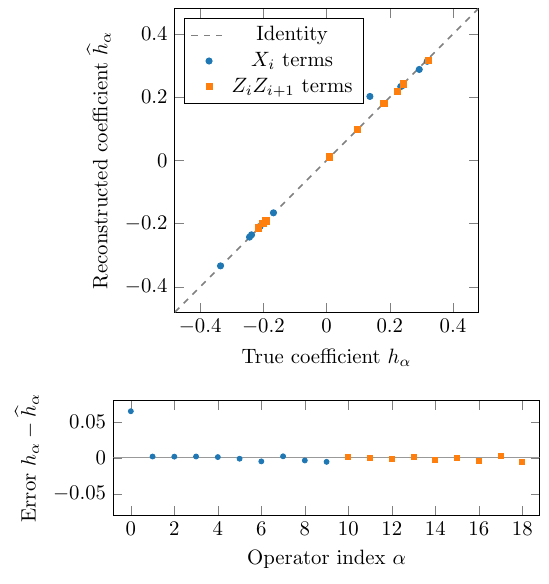}
    \caption{TFIM.}
    \label{fig:tfim_reconstruction_with_residuals}
\end{subfigure}
\hfill
\begin{subfigure}[t]{0.31\textwidth}
    \centering
    \includegraphics[width=\linewidth]{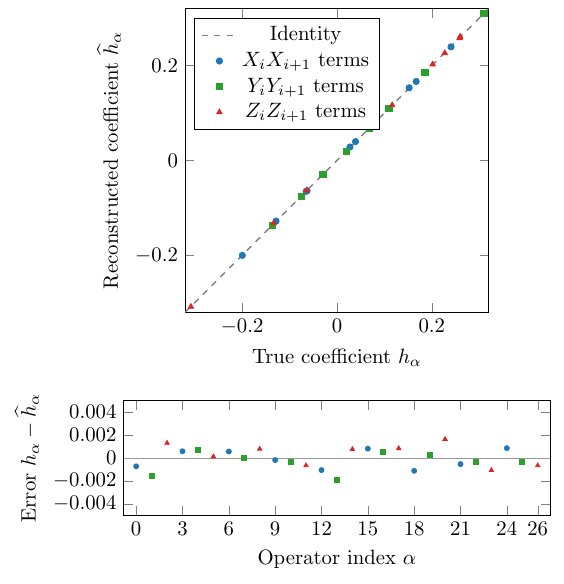}
    \caption{XYZ chain.}
    \label{fig:xyz_reconstruction_with_residuals}
\end{subfigure}
\hfill
\begin{subfigure}[t]{0.31\textwidth}
    \centering
    \includegraphics[width=\linewidth]{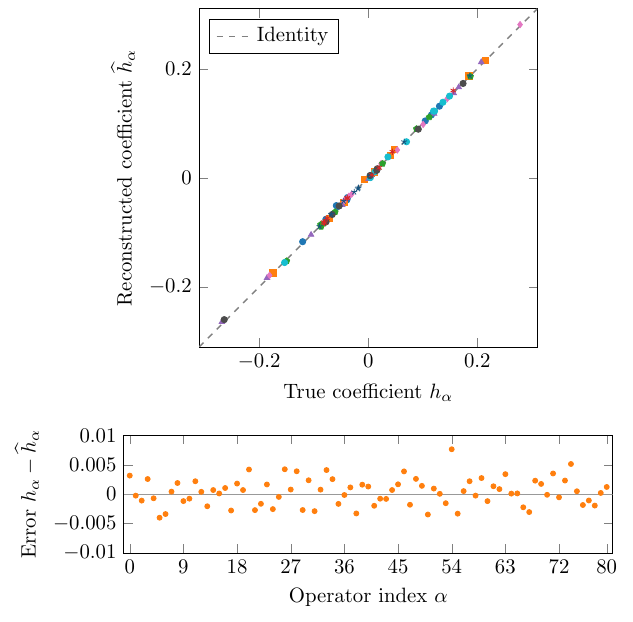}
    \caption{General nearest-neighbor 2-local model.}
    \label{fig:general_2local_nn_reconstruction_with_residuals}
\end{subfigure}

\caption{Representative coefficient reconstructions at $n=10$ with open boundary conditions, $1000$ probes, $20{,}000$ shots, and $t=1453.28$. The operator-norm errors are $0.095451$, $0.017448$, and $0.056607$ for TFIM, XYZ, and the general nearest-neighbor 2-local model, respectively; the corresponding maximum coefficient errors are $0.065303$, $0.0034702$, and $0.0076717$, and the normalized overlaps are $0.997787$, $0.999977$, and $0.999782$.}
\label{fig:reconstruction_examples}
\end{figure}

Figure~\ref{fig:reconstruction_examples} gives representative coefficient reconstructions for three Hamiltonian families at a long evolution time. In all three cases, the recovered coefficients lie close to the identity line and the residual plots remain small across the learned support. These examples are consistent with the algorithm recovering the Hamiltonian direction with inverse-polynomial accuracy even at large time and with finite-shot measurement noise.

\FloatBarrier
\subsection{Gap and equilibration diagnostics}

\begin{figure}[!htbp]
\centering
\captionsetup{font=small}

\begin{subfigure}[t]{0.32\textwidth}
    \centering
    \includegraphics[width=\textwidth]{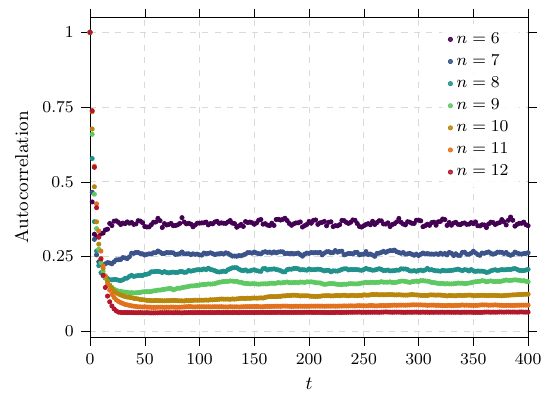}
    \caption{ $\max_{A\perp H}\frac{1}{d}\tr(A(t)A)$ for dense nearest-neighbor Gaussian Hamiltonians.}
    \label{fig:weak-equilibration-numerics}
\end{subfigure}
\hfill
\begin{subfigure}[t]{0.32\textwidth}
    \centering
    \includegraphics[width=\textwidth]{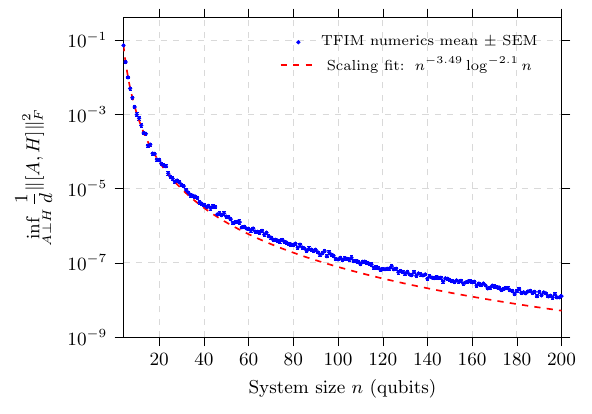}
    \caption{Static commutator gap scaling for Gaussian TFIM.}
    \label{fig:static-gap-numerics}
\end{subfigure}
\hfill
\begin{subfigure}[t]{0.32\textwidth}
    \centering
    \includegraphics[width=\textwidth]{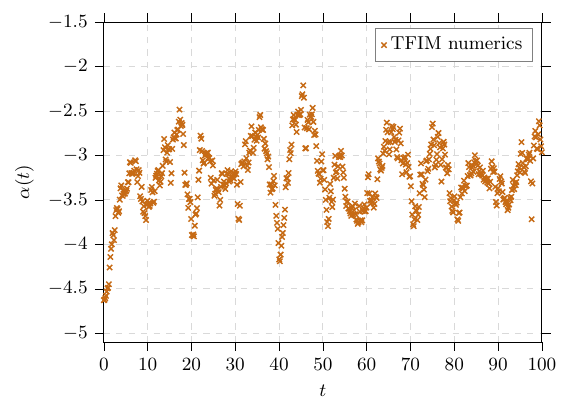}
    \caption{Extracted dynamical exponent $\alpha(t)$ from $\min_{A\perp H}\frac{1}{d}\|[U(t),A]\|_F^2\propto n^{\alpha(t)}$ for Gaussian TFIM.}
    \label{fig:dynamical-gap-exponent}
\end{subfigure}

\caption{Numerics for equilibration and commutator gap scalings.}
\label{fig:combined_numerics}
\end{figure}

Figure~\ref{fig:combined_numerics}(a) directly probes the weak-equilibration quantity appearing in the theoretical analysis. For each time $t$, the plotted value is
\begin{align} \max_{\substack{A\in \mathcal{O}_{\mathcal{V}}^{\mathrm{op}},\\ \tr(AH)=0, \\ \tfrac{1}{d}\tr(A^2)=1}} \frac{1}{d}\,\tr\Big(U(t)^\dagger A U(t)A\Big)\end{align}
where the supremum is over normalized local observables orthogonal to the Hamiltonian. This quantity measures the largest possible infinite-temperature autocorrelation among local directions orthogonal to the Hamiltonian. At $t=0$, the value is $1$, as every observable is perfectly correlated with itself. After a short transient, the curves drop well below $1$ and remain bounded away from $1$. This supports the proof of Appendix \ref{sec:equilibration}.

Figure~\ref{fig:combined_numerics}(b) examines the static commutator gap
\begin{align} \inf_{\substack{A\in \mathcal{O}_{\mathcal{V}}^{\mathrm{op}},\\\tr(AH)=0,\\\tfrac{1}{d}\tr(A^2)=1}}\frac{1}{d}\|[A,H]\|_F^2\end{align}
for the transverse-field Ising model with Gaussian coefficients. The data points show the numerically computed gap as a function of system size. For the normalized Gaussian TFIM bound proved here, the theoretical lower-bound scaling at fixed failure probability is of order $n/(N^7L_{\max}^2)$, which for the open-boundary chain $N=2n-1$ is $\Omega(n^{-6}\log^{-2}n)$. The numerical gaps decay with $n$, as expected, but remain substantially above the proven lower bound over the accessible system sizes. This is not surprising since the proofs are designed to give a uniform high-probability lower bound over all admissible observables and all random instances, so they are conservative. The numerics suggest that the actual finite-size gap is much larger than the guaranteed worst-case lower bound.

Figure~\ref{fig:combined_numerics}(c) studies the dynamical commutator gap
\begin{align} \inf_{\substack{A\in \mathcal{O}_{\mathcal{V}}^{\mathrm{op}},\\\tr(AH)=0,\\\tfrac{1}{d}\tr(A^2)=1}}\frac{1}{d}\|[U(t),A]\|_F^2\end{align}
for TFIM instances with Gaussian coefficients, by fitting the system-size dependence to a power law $n^{\alpha(t)}$. The extracted exponent $\alpha(t)$ fluctuates with time, reflecting the oscillatory dependence of $[U(t),A]$ on the energy gaps of $H$. This is consistent with the uniform sinc-reduction argument: although special times can produce cancellations in individual energy-gap coordinates, a typical sampled time does not allow all local directions orthogonal to $H$ to become simultaneously almost conserved. The observed exponents therefore provide numerical evidence for the dynamical version of the commutator gap that is actually used in the learning guarantee.

Overall, the numerical results support the main conceptual picture of this paper. In particular, the Hamiltonian is recovered accurately from a single unknown long-time evolution and the error decreases with the expected statistical resources. The relevant local commutator gaps are consistent with being inverse polynomial in the tested finite-size regimes.

\bibliographystyle{alpha}
\bibliography{references}

\end{document}